\documentclass[twoside,a4paper]{amsart}
\providecommand{\text}[1]{\mbox{#1}}
\usepackage[OT2, OT1]{fontenc}
  \usepackage[pdfauthor={Vladimir V. Kisil},%
      pdftitle={Classical/Quantum=Commutative/Noncommutative?},%
      pdfsubject={mathematics, physics},%
      backref=page,%
    pdfkeywords={symmetries, quantum mechanics, hypercomplex numbers},
  breaklinks=true,%
  colorlinks=true,%
  bookmarks=true,%
  backref=page,%
  pagebackref=true
  ]{hyperref}
\usepackage[backrefs,msc-links,nobysame,initials]{amsrefs}
\usepackage{graphicx,amssymb}

\def\fileversion{v0.1}
\def\filedate{2012/01/23}
\RequirePackageWithOptions
{amsrefs}\relax

\BibSpecAlias{incollection}{inproceedings}
\BibSpec{article}{%
    +{}  {\PrintAuthors}                {author}
    +{.} { }                            {title}
    +{.} { }                            {part}
    +{:} { }                            {subtitle}
    +{.} { \PrintContributions}         {contribution}
    +{.} { \PrintPartials}              {partial}
    +{.} { \emph}                       {journal}
    +{,} { \textbf}                            {volume}
    +{}  { \parenthesize}               {number}
    +{:} {}                             {pages}
    +{,} { \PrintDateB}                 {date}
    +{,} { }                            {status}
    +{.} { \PrintTranslation}           {translation}
    +{.} { Reprinted in \PrintReprint}  {reprint}
    +{.} { }                            {note}
    +{.} {}                             {transition}
}

\BibSpec{partial}{%
    +{}  {}                             {part}
    +{:} { }                            {subtitle}
    +{.} { \PrintContributions}         {contribution}
    +{.} { \emph}                       {journal}
    +{,} { \textbf}                            {volume}
    +{}  { \parenthesize}               {number}
    +{:} {}                             {pages}
    +{,} { \PrintDateB}                 {date}
}

\BibSpec{book}{%
    +{}  {\PrintPrimary}                {transition}
    +{.} { \emph}                       {title}
    +{.} { }                            {part}
    +{:} { \emph}                       {subtitle}
    +{.} { }                            {series}
    +{,} { \voltext}                    {volume}
    +{.} { Edited by \PrintNameList}    {editor}
    +{.} { Translated by \PrintNameList}{translator}
    +{.} { \PrintContributions}         {contribution}
    +{.} { }                            {publisher}
    +{.} { }                            {organization}
    +{,} { }                            {address}
    +{,} { \PrintEdition}               {edition}
    +{,} { \PrintDateB}                 {date}
    +{.} { }                            {note}
    +{.} {}                             {transition}
    +{.} { \PrintTranslation}           {translation}
    +{.} { Reprinted in \PrintReprint}  {reprint}
    +{.} {}                             {transition}
}

\BibSpec{collection.article}{%
    +{}  {\PrintAuthors}                {author}
    +{.} { }                            {title}
    +{.} { }                            {part}
    +{:} { }                            {subtitle}
    +{.} { \PrintContributions}         {contribution}
    +{.} { \PrintConference}            {conference}
    +{.} { \PrintBook}                  {book}
    +{.} { In \PrintEditorsA}              {editor}
    +{.} { \emph}                         {booktitle}
    +{,} { pages~}                      {pages}
    +{,} { }                            {publisher}
    +{,} { }                            {organization}
    +{,} { }                            {address}
    +{,} { \PrintDateB}                 {date}
    +{.} { \PrintTranslation}           {translation}
    +{.} { Reprinted in \PrintReprint}  {reprint}
    +{.} { }                            {note}
    +{.} {}                             {transition}
}

 \BibSpec{inproceedings}{%
     +{}  {\PrintAuthors}                {author}
     +{.} { }                            {title}
     +{.} { }                            {part}
     +{:} { }                            {subtitle}
     +{.} { \PrintContributions}         {contribution}
     +{.} { \PrintConference}            {conference}
     +{.} { \PrintBook}                  {book}
     +{.} { In \PrintEditorsA}              {editor}
     +{.} {  \emph}                         {booktitle}
     +{,} { pages~}                      {pages}
     +{,} { }                            {publisher}
     +{,} { }                            {organization}
     +{,} { }                            {address}
     +{,} { \PrintDateB}                 {date}
     +{.} { \PrintTranslation}           {translation}
     +{.} { Reprinted in \PrintReprint}  {reprint}
     +{.} { }                            {note}
     +{.} {}                             {transition}
 }

\BibSpec{conference}{%
    +{}  {}                        {title}
    +{}  {\PrintConferenceDetails} {transition}
}

\BibSpec{innerbook}{%
    +{.} { \emph}                       {title}
    +{.} { }                            {part}
    +{:} { \emph}                       {subtitle}
    +{.} { }                            {series}
    +{,} { \voltext}                    {volume}
    +{.} { Edited by \PrintNameList}    {editor}
    +{.} { Translated by \PrintNameList}{translator}
    +{.} { \PrintContributions}         {contribution}
    +{.} { }                            {publisher}
    +{.} { }                            {organization}
    +{,} { }                            {address}
    +{,} { \PrintEdition}               {edition}
    +{,} { \PrintDateB}                 {date}
    +{.} { }                            {note}
    +{.} {}                             {transition}
}

\BibSpec{report}{%
    +{}  {\PrintPrimary}                {transition}
    +{.} { \emph}                       {title}
    +{.} { }                            {part}
    +{:} { \emph}                       {subtitle}
    +{.} { \PrintContributions}         {contribution}
    +{.} { Technical Report }           {number}
    +{,} { }                            {series}
    +{.} { }                            {organization}
    +{,} { }                            {address}
    +{,} { \PrintDateB}                 {date}
    +{.} { \PrintTranslation}           {translation}
    +{.} { Reprinted in \PrintReprint}  {reprint}
    +{.} { }                            {note}
    +{.} {}                             {transition}
}

\BibSpec{thesis}{%
    +{}  {\PrintAuthors}                {author}
    +{,} { \emph}                       {title}
    +{:} { \emph}                       {subtitle}
    +{.} { \PrintThesisType}            {type}
    +{.} { }                            {organization}
    +{,} { }                            {address}
    +{,} { \PrintDateB}                 {date}
    +{.} { \PrintTranslation}           {translation}
    +{.} { Reprinted in \PrintReprint}  {reprint}
    +{.} { }                            {note}
    +{.} {}                             {transition}
}

\makeatletter
  \providecommand{\limlike}[1]{\mathop {\operator@font #1}}
  \providecommand{\loglike}[1]{\mathop {\operator@font #1}\nolimits}
\makeatother

\providecommand{\href}[2]{#2}
\providecommand{\SL}[1][2]{\ensuremath{\mathrm{SL}_{#1}(\Space{R}{})}}
\providecommand{\Mp}[1][2]{\ensuremath{\mathrm{Mp}(#1)}}
\providecommand{\Sp}[1][n]{\ensuremath{\mathrm{Sp}(#1)}}
\newcommand{\Aprime}{A'}
\providecommand{\Space}[3][]{\ensuremath{\mathbb{#2}^{#3}_{#1}{}}}
  \providecommand{\FSpace}[3][]{\ensuremath{\ifx#2l \ell_{#3}^{#1}{}\else
  \mathcal{#2}_{#3}^{#1}{}\fi}} 
\providecommand{\norm}[2][\relax]{\left\|#2\right\|\ifx#1\relax\else_{#1}\fi}
\providecommand{\modulus}[2][\relax]{\left| #2 \right|\ifx#1\relax\else_{#1}\fi}
\providecommand{\scalar}[3][\relax]{\left\langle #2,#3 
        \right\rangle\ifx#1\relax\else_{#1}\fi}
\providecommand{\myh}{h}
\providecommand{\such}{\,\mid\,}
\providecommand{\alli}{\iota}
\providecommand{\myhbar}{\hslash}
\providecommand{\notingiq}{}
\providecommand{\rmi}{\mathrm{i}}
\providecommand{\rmd}{\mathrm{d}}
\providecommand{\rme}{\mathrm{e}}
\providecommand{\rmh}{\mathrm{j}}
\providecommand{\rmp}{\varepsilon}

\providecommand{\map}[1]{\mathsf{#1}}
\providecommand{\spec}[1][]{\ensuremath{\mathbf{sp}}\,}
\providecommand{\cosp}{\loglike{cosp}}
\providecommand{\sinp}{\loglike{sinp}}

\providecommand{\ladder}[2][]{L_{#1}^{\!#2}}
\providecommand{\linv}[2][\relax]{\mathfrak{L}^{#2}_{#1}}
\providecommand{\anti}{\mathcal{A}}
\providecommand{\ub}[3][]{\left\{\!#1\left[#2,#3\right]\!#1\right\}}
\providecommand{\cycle}[3][\relax]{\ifx#1\relax C\else{\ifx#1\tilde S
    \else G \fi}\fi^{#2}_{#3}}
\providecommand{\Cliff}[2][\comment]{{\ensuremath{%
\mathcal{C}\kern-0.18em\ell(#1,#2)}}}
\providecommand{\comment}[1]{}
\providecommand{\uir}[3][0]{\ifcase #1{\rho^{#2}_{#3}}%
\or {\breve{\rho}^{#2}_{#3}}%
\or {\tilde{\rho}^{#2}_{#3}}\fi}

\providecommand{\algebra}[1]{\ensuremath{\mathfrak{#1}}}
  \providecommand{\Zbl}[1]{Zbl\href{http://www.emis.de:80/cgi-bin/zmen/ZMATH/en/zmathf.html?first=1&maxdocs=3&type=html&an=#1&format=complete}{#1}}

\providecommand{\myeprint}[2]{E-print: \href{#1}{\texttt{#2}}}
\providecommand{\doi}[1]{doi: \href{http://dx.doi.org/#1}{#1}}

\providecommand{\wiki}[2]{#2}
\providecommand{\oper}[1]{\mathcal{#1}}

        \newtheorem{theorem}{Theorem}[section]

    \PassOptionsToPackage{british}{babel}
    \newtheorem{proposition}[theorem]{Proposition}
    \newtheorem{lemma}[theorem]{Lemma}
    \newtheorem{corollary}[theorem]{Corollary}
    
    \theoremstyle{definition}

    \newtheorem{definition}[theorem]{Definition}
    \newtheorem{notation}[theorem]{Notation}
\newtheorem{principle}[theorem]{Principle}
    \newtheorem{example}[theorem]{Example}

    \newtheorem{exercise}[theorem]{Exercise}

    \theoremstyle{remark}
    \newtheorem{remark}[theorem]{Remark}

\makeindex

\usepackage[all]{xy}
\DeclareFontFamily{U}{mathx}{\hyphenchar\font45}
\DeclareFontShape{U}{mathx}{m}{n}{
      <5> <6> <7> <8> <9> <10>
      <10.95> <12> <14.4> <17.28> <20.74> <24.88>
      mathx10
      }{}
\DeclareSymbolFont{mathx}{U}{mathx}{m}{n}
\DeclareFontSubstitution{U}{mathx}{m}{n}
\DeclareMathAccent{\wideparen}{0}{mathx}{"75}
\DeclareFontFamily{OT1}{cyr}{}
\DeclareFontShape{OT1}{cyr}{m}{n}
   {  <5> <6> <7> <8> <9> gen * wncyr
      <10> <10.95> <12> <14.4> <17.28> <20.74> <24.88> wncyr10}{}
\DeclareFontShape{OT1}{cyr}{m}{it}
    {
       <5> <6> <7> <8> <9> gen * wncyi
      <10> <10.95> <12> <14.4> <17.28> <20.74> <24.88>wncyi10
      }{}
\DeclareFontShape{OT1}{cyr}{m}{ss}
    {
       <5> <6> <7> <8> wncyss8
       <9> wncy9
      <10> <10.95> <12> <14.4> <17.28> <20.74> <24.88>wncyss10
      }{}
\DeclareFontShape{OT1}{cyr}{m}{sc}
    {
       <5> <6> <7> <8> <9> <10> <10.95> <12> <14.4> <17.28> <20.74> <24.88>wncysc10
      }{}
\DeclareFontShape{OT1}{cyr}{bx}{n}
   {
       <5> <6> <7> <8> <9> gen * wncyb
      <10> <10.95> <12> <14.4> <17.28> <20.74> <24.88>wncyb10
      }{}
\DeclareTextFontCommand{\textcyr}{\fontfamily{cyr}\selectfont}
\newcommand{\cyr}{\fontfamily{cyr}\selectfont\def\cprime{\~}}
\providecommand{\cprime}{'}

\begin{document}

\title[Symmetry, Geometry, Quantization with Hypercomplex Numbers]
{Symmetry, Geometry, and Quantization\\ with Hypercomplex Numbers}

\author[Vladimir V. Kisil]%
{Vladimir V. Kisil}

\thanks{On  leave from Odessa University.}

\address{School of Mathematics,
University of Leeds, 
Leeds, LS2\,9JT, England}

\email{\href{mailto:kisilv@maths.leeds.ac.uk}{kisilv@maths.leeds.ac.uk}}

\urladdr{\href{http://www.maths.leeds.ac.uk/~kisilv/}%
{http://www.maths.leeds.ac.uk/\~{}kisilv/}}


\begin{abstract}
  These notes describe some links between the group \(\SL\),
  the Heisenberg group and hypercomplex numbers\,--\,complex, dual and
  double numbers. Relations between quantum and classical mechanics
  are clarified in this framework.  In particular, classical mechanics
  can be obtained as a theory with \emph{noncommutative} observables
  and a \emph{non-zero} Planck constant if we replace complex numbers
  in quantum mechanics by dual numbers. Our consideration is based on
  induced representations which are build from complex-/\-dual-\-/double-valued characters. Dynamic equations, rules of additions of
  probabilities, ladder operators and uncertainty relations are
  discussed. Finally, we prove a Calder\'on--Vaillancourt-type norm
  estimation for relative convolutions.
\end{abstract}

\maketitle

\tableofcontents

\section*{Introduction}
These lecture notes describe some links between the
group \(\SL\),
the Heisenberg group and hypercomplex numbers. The
described relations appear in a natural way without any enforcement
from our side. The discussion is
illustrated by mathematical models of various physical systems.

By hypercomplex numbers we mean two-dimensional real associative
commutative algebras. It is known~\cite{LavrentShabat77}, that any
such algebra is isomorphic either to complex, dual or double numbers,
that is collection of elements \(a+\alli b\),
where \(a\),
\(b\in\Space{R}{}\)
and \(\alli^2=-1\),
\(0\)
or \(1\).
Complex numbers are crucial in quantum mechanics (or, in fact, any
wave process), dual numbers similarly serve classical mechanics and
double numbers are perfect to encode relativistic
space-time\footnote{The last case is not discussed much in these
  notes, see~\cite{BocCatoniCannataNichZamp07} for details.}.

Section~\ref{sec:overview} contains an easy-reading overview of the
r\^ole of complex numbers in quantum mechanics and indicates that
classical mechanics can be described as a theory with
\emph{noncommutative} observables and a \emph{non-zero} Planck
constant if we replace complex numbers by dual numbers. The Heisenberg
group is the main ingredient for both\,--\,quantum and
classic\,--\,models. The detailed exposition of the theory is provided
in the following sections.

Section~\ref{sec:groups-homog-spac} introduces the group \(\SL\)
and describes all its actions on two-di\-men\-si\-onal homogeneous
spaces: it turns out that they are M\"obius transformations of
complex, dual and double numbers. We also re-introduce the Heisenberg
group in more details. In particular, we point out Heisenberg group's
automorphisms from the symplectic action of \(\SL\).

Section~\ref{sec:induc-repr} uses Mackey's induced representation to
construct linear representations of \(\SL\)
and the Heisenberg group. We use all sorts (complex, dual and double)
of characters of one-dimensional subgroups to induce representations of
\(\SL\).
The similarity between obtained representations in hypercomplex
numbers is illustrated by corresponding ladder operators.

Section~\ref{sec:quantum-mechanics-1} systematically presents
the Hamiltonian formalism obtained from linear representations of the
Heisenberg group. Using complex, dual and double numbers we recover
principal elements of quantum, classical and hyperbolic
(relativistic?) mechanics. This includes both the Hamilton--Heisenberg
dynamical equation, rules of addition of probabilities and some
examples.

Section~\ref{sec:gaussian} introduces co- and contra-variant
transforms, which are also known under many other names, e.g. wavelet
transform. These transforms intertwine the given representation with
left and right regular representations. We use this observation to
derive a connection between the uncertainty relations and analyticity
condition\,--\,both in the standard meaning for the Heisenberg group
and a new one for \(\SL\).
We also obtain a Calder\'on--Vaillancourt-type norm estimation for
integrated representation.

\section{Preview: Quantum and Classical Mechanics}
\label{sec:overview}
\par
\hfill\parbox{0.6\textwidth}{\footnotesize\ldots it was on a Sunday
  that the idea first occurred to me that \(ab- ba\) might correspond to a
  Poisson bracket.
 \par \hfil P.A.M. Dirac, \url{http://www.aip.org/history/ohilist/4575_1.html}
}\par\medskip
In this section we will demonstrate that a Poisson bracket do not only
corresponds to a commutator, in fact a Poisson bracket is the image of
the commutator under a transformation which uses dual numbers.

\subsection{Axioms of Mechanics}
\label{sec:quant-class-mech}

There is a recent revival of interest in foundations of quantum
mechanics, 
which is essentially motivated by engineering challenges at the
nano-scale. There are strong indications that we need to revise the
development of the quantum theory from its early days.

In 1926, Dirac discussed the idea that quantum mechanics can be obtained
from classical description through a change in the only rule, cf.~\cite{Dirac26a}:
\begin{quote}
  \ldots there is one basic assumption of the classical theory which
  is false, and that if this assumption were removed and replaced by
  something more general, the whole of atomic theory would follow
  quite naturally. Until quite recently, however, one has had no idea
  of what this assumption could be.
\end{quote}

In Dirac's view, such a condition is provided by the Heisenberg commutation
relation of coordinate and momentum variables~\cite{Dirac26a}*{(1)}:
\begin{equation}
  \label{eq:heisenberg-comm-basic}
  q_r p_r-p_r q_r=\rmi \myh.
\end{equation}
Algebraically, this identity declares noncommutativity of \(q_r\) and
\(p_r\). Thus, Dirac stated~\cite{Dirac26a} that classical mechanics is formulated
through commutative quantities (``c-numbers'' in his terms) while
quantum mechanics requires noncommutative quantities (``q-numbers''). The
rest of theory may be unchanged if it does not contradict to the above
algebraic rules. This was explicitly re-affirmed at the
first sentence of the subsequent paper~\cite{Dirac26b}:
\begin{quote}
  The new mechanics of the atom introduced by Heisenberg may be based
  on the assumption that the variables that describe a dynamical
  system do not obey the commutative law of multiplication, but
  satisfy instead certain quantum conditions.
\end{quote}
The same point of view is expressed in his later works \citelist{\cite{DiracDirections}*{p.~6}
\cite{DiracPrinciplesQM}*{p.~26}}.

Dirac's approach was largely approved, especially by researchers
on the mathematical side of the board. Moreover, the vague version\,--\,``quantum is something noncommutative''\,--\,of the original statement
 was lightly reverted to ``everything
noncommutative is quantum''. For example, there is a fashion to label
any noncommutative algebra as a ``quantum space''~\cite{Cuntz01a}. 

Let us carefully review Dirac's  idea about noncommutativity as the
principal source of quantum theory.

\subsection{``Algebra'' of Observables}
\label{sec:algebra-observables}

Dropping the commutativity hypothesis on observables, Dirac
made~\cite{Dirac26a} the following (apparently flexible) assumption:
\begin{quote}
  All one knows about q-numbers is that if \(z_1\) and \(z_2\) are two
  q-numbers, or one q-number and one c-number, there exist the numbers
  \(z_1 + z_2\), \(z_1 z_2\), \(z_2 z_1\), which will in general be
  q-numbers but may be c-numbers.
\end{quote}
Mathematically, this (together with some natural identities) means
that observables form an algebraic structure known as a \emph{ring}.
Furthermore, the linear \emph{superposition principle} imposes a liner
structure upon observables, thus their set becomes an \emph{algebra}.
Some mathematically-oriented texts,
e.g.~\cite{FaddeevYakubovskii09}*{\S~1.2}, directly speak about an
``algebra of observables'' which is not far from the above
quote~\cite{Dirac26a}. It is also deducible from two connected
statements in Dirac's canonical textbook:
\begin{enumerate}
\item ``the linear operators corresponds to the dynamical variables at
  that time''~\cite{DiracPrinciplesQM}*{\S~7, p.~26}.
\item ``Linear operators can be added
  together''~\cite{DiracPrinciplesQM}*{\S~7, p.~23}. 
\end{enumerate}

However, the assumption that any two observables may be added cannot
fit into a physical theory. To admit addition, observables need to
have the same dimensionality. In the simplest example of the
observables of coordinate \(q\) and momentum \(p\), which units shall
be assigned to the expression \(q+p\)? Meters or
\(\frac{\text{kilos}\times\text{meters}}{\text{seconds}}\)? If we get the value \(5\)
for \(p+q\) in the metric units, what is then the result in the imperial
ones? 
Since these questions cannot be answered, the above
Dirac's assumption is not a part of any physical theory.

Another common definition suffering from the same problem is used in
many excellent books written by distinguished mathematicians, see for
example \citelist{\cite{Mackey63}*{\S~2-2} \cite{Folland89}*{\S~1.1}}.
It declares that quantum observables are projection-valued Borel
measures on the \emph{dimensionless} real line. Such a definition
permit an instant construction (through the functional
calculus) of new observables, including
algebraically formed~\cite{Mackey63}*{\S~2-2, p.~63}:
\begin{quote}
  Because of Axiom III, expressions such as \(A^2\), \(A^3+A\),
  \(1-A\), and \( \rme^A\) all make sense whenever \(A\) is an
  observable. 
\end{quote}
However, if \(A\) has a physical dimension (is not a scalar) then the
expression \(A^3+A\) cannot be assigned a dimension in a consistent
manner.

Of course, physical defects of the above (otherwise perfect)
mathematical constructions do not prevent physicists from making
correct calculations
, which are in a good
agreement with experiments. We are not going to analyse methods which
allow researchers to escape the indicated dangers. Instead, 
it will be more beneficial to outline alternative mathematical
foundations of quantum theory, which do not have those shortcomings.

\subsection{Non-Essential  Noncommutativity}
\label{sec:essent-noncomm}

While we can add two observables if they have the same dimension only,
physics allows us to multiply any observables freely. Of course, the
dimensionality of a product is the product of dimensionalities, thus
the commutator \([A,B]=AB-BA\) is well defined for any two observables
\(A\) and \(B\). In particular, the
commutator~\eqref{eq:heisenberg-comm-basic} is also well-defined, but
is it indeed so important?

In fact, it is easy to argue that noncommutativity of observables is
not an essential prerequisite for quantum mechanics: there are
constructions of quantum theory which do not relay on it at all. The
most prominent example is the Feynman path integral. To focus on the
really cardinal moments, we firstly take the popular
lectures~\cite{Feynman1990qed}, which present the main elements in a
very enlightening way. Feynman managed to tell the fundamental
features of quantum electrodynamics without any reference to
(non-)commutativity: the entire text does not mention it anywhere.

Is this an artefact of the popular nature of these lecture? Take the
academic presentation of path integral technique given in
\cite{FeynHibbs65}. It mentioned (non-)com\-mu\-ta\-ti\-vi\-ty only on
pages~115--6 and 176. In addition, page~355 contains a remark on
noncommutativity of quaternions, which is irrelevant to our topic.
Moreover, page~176 highlights that noncommutativity of quantum
observables is a consequence of the path integral formalism rather
than an indispensable axiom.

But what is the mathematical source of quantum theory if noncommutativity
is not? The vivid presentation in~\cite{Feynman1990qed} uses stopwatch
with a single hand to explain the calculation of path
integrals. The angle of stopwatch's hand presents the \emph{phase} 
for a path \(x(t)\) between two points in the configuration space.
The mathematical expression for the path's phase is
\cite{FeynHibbs65}*{(2-15)}:
\begin{equation}
  \label{eq:phase-path}
  \phi[x(t)]=\mathrm{const}\cdot  \rme^{(\rmi/\myhbar)S[x(t)]}\notingiq
\end{equation}
where \(S[x(t)]\) is the \emph{classic action} along the path
\(x(t)\). Summing up contributions~\eqref{eq:phase-path} along all
paths between two points \(a\) and \(b\) we obtain the amplitude
\(K(a,b)\). This amplitude presents very accurate description of many
quantum phenomena. Therefore, expression~\eqref{eq:phase-path} is also
a strong contestant for the r\^ole of the cornerstone of quantum
theory.

Is there anything common between two ``principal''
identities~\eqref{eq:heisenberg-comm-basic} and~\eqref{eq:phase-path}?
Seemingly, not. A more attentive reader may say that there are only two
common elements there (in order of believed significance):
\begin{enumerate}
\item The non-zero Planck constant \(\myhbar\).
\item The imaginary unit \(\rmi\).
\end{enumerate}

The Planck constant was the first manifestation of quantum (discrete) 
behaviour and it is at the heart of the whole theory. In contrast, classical
mechanics is oftenly obtained as a semiclassical limit \(\myhbar
\rightarrow 0\). Thus, the non-zero Planck constant looks like a clear
marker of quantum world in its opposition to the classical
one. Regrettably, there is a common practice to ``chose our units such
that \(\myhbar=1\)''. Thus, the Planck constant becomes oftenly invisible in many
formulae even being implicitly present there. Note also, that \(1\) in
the identity \(\myhbar=1\) is not a scalar but a physical quantity
with the dimensionality of the action. Thus, the simple omission of the Planck
constant invalidates dimensionalities of physical equations.

The complex imaginary unit is also a mandatory element of quantum
mechanics in all its possible formulations. It is enough to point out
that the popular lectures~\cite{Feynman1990qed} managed to avoid any
mention of noncommutativity but did uses complex numbers both
explicitly (see the Index there) and implicitly (as rotations of the
hand of a stopwatch). However, it is a common perception that complex
numbers are a useful but manly technical tool in quantum theory.

\subsection{Quantum Mechanics from the Heisenberg Group}
\label{sec:heisenberg-group-qm}

Looking for a source of quantum theory we again return to the
Heisenberg commutation relations~\eqref{eq:heisenberg-comm-basic}:
they are an important part of quantum mechanics (either as a
prerequisite or as a consequence). It was observed for a long time
that these relations are a representation of the structural identities
of the Lie algebra of the Heisenberg
group~\cites{Folland89,Howe80a,Howe80b}. In the simplest case of
one dimension, the Heisenberg group \(\Space{H}{}=\Space{H}{1}\) can
be realised by the Euclidean space \(\Space{R}{3}\) with the group
law:
\begin{equation}
  \label{eq:H-n-group-law0}
  \textstyle
  (s,x,y)*(s',x',y')=(s+s'+\frac{1}{2}\omega(x,y;x',y'),x+x',y+y')\notingiq
\end{equation} 
where \(\omega\)
is the \emph{symplectic form} on
\(\Space{R}{2}\)~\cite{Arnold91}*{\S~37},
which is behind the entire classical Hamiltonian formalism:
\begin{equation}
  \label{eq:symplectic-form}
  \omega(x,y;x',y')=xy'-x'y.
\end{equation}
Here, like for the path integral, we see another example of a quantum notion
being defined through a classical object. 

The Heisenberg group is noncommutative since
\(\omega(x,y;x',y') =-\omega(x',y';x,y)\).
The collection of points \((s,0,0)\)
forms the centre of \(\Space{H}{}\),
that is \((s,0,0)\)
commutes with any other element of the group. We are interested in the
unitary irreducible representations (UIRs) \(\uir{}{}\)
of \(\Space{H}{}\)
in an infinite-dimensional Hilbert space \(H\),
that is a group homomorphism
(\(\uir{}{}(g_1)\uir{}{}(g_2)=\uir{}{}(g_1*g_2)\))
from \(\Space{H}{}\)
to unitary operators on \(H\).
By Schur's lemma, for such a representation \(\uir{}{}\),
the action of the centre shall be multiplication by an unimodular
complex number, i.e. \(\uir{}{}(s,0,0)= \rme^{2\pi\rmi\myhbar s} I\)
for some real \(\myhbar\neq 0\).

Furthermore, the celebrated Stone--von~Neumann
theorem~\cite{Folland89}*{\S~1.5} tells that all UIRs of
\(\Space{H}{}\)
in complex Hilbert spaces with the same value of \(\myhbar\)
are unitary equivalent. In particular, this implies that any
realisation of quantum mechanics, e.g. the Schr\"odinger wave
mechanics, which provides the commutation
relations~\eqref{eq:heisenberg-comm-basic} shall be unitary equivalent
to the Heisenberg matrix mechanics based on these relations.

In particular, any UIR of \(\Space{H}{}\) is equivalent to a
subrepresentation of the following representation on
\(\FSpace{L}{2}(\Space{R}{2})\):  
\begin{equation}
  \label{eq:stone-inf0}
  \textstyle
  \uir{}{\myhbar}(s,x,y): f (q,p) \mapsto 
   \rme^{-2\pi\rmi(\myhbar s+qx+py)}
  f \left(q-\frac{\myhbar}{2} y, p+\frac{\myhbar}{2} x\right).
\end{equation}
Here \(\Space{R}{2}\)
has the physical meaning of the classical \emph{phase space} with
\(q\)
representing the coordinate in the configurational space and
\(p\)---the
respective momentum. The function \(f(q,p)\)
in~\eqref{eq:stone-inf0} presents a state of the physical system as an
amplitude over the phase space.  Thus, the
action~\eqref{eq:stone-inf0} is more intuitive and has many technical
advantages~\cites{Howe80b,Zachos02a,Folland89} in comparison with the
well-known Schr\"odinger representation
(cf.~\eqref{eq:schroedinger-rep-conf}), to which it is unitary
equivalent, of course.

Infinitesimal generators of the one-parameter subgroups
\(\uir{}{\myhbar}(0,x,0)\) and \(\uir{}{\myhbar}(0,0,y)\)
from~\eqref{eq:stone-inf0} are the operators
\(\frac{1}{2}\myhbar\partial_p-2\pi\rmi q\) and
\(-\frac{1}{2}\myhbar\partial_q-2\pi\rmi p\). For these, we can
directly verify the commutator identity:
\begin{displaymath}
\textstyle   [-\frac{1}{2}\myhbar\partial_q-2\pi\rmi p, 
\frac{1}{2}\myhbar\partial_p-2\pi\rmi q]= \rmi \myh,\quad
  \text{ where } \myh =2\pi\myhbar.
\end{displaymath}
Since we have a representation
of~\eqref{eq:heisenberg-comm-basic}, these operators can be used as
a model of the quantum coordinate and momentum.

For a Hamiltonian \(H(q,p)\) we can integrate the representation
\(\uir{}{\myhbar}\) with the Fourier transform \(\hat{H}(x,y)\) of
\(H(q,p)\):
\begin{equation}
  \label{eq:weyl-quantisation}
  \tilde{H}=\int_{\Space{R}{2}} \hat{H}(x,y)\,
  \uir{}{\myhbar}(0,x,y)\,\rmd x\,\rmd y
\end{equation}
and obtain (possibly unbounded) operator \(\tilde{H}\) on
\(\FSpace{L}{2}(\Space{R}{2})\).  This assignment of the operator
\(\tilde{H}\) (quantum observable) to a function \(H(q,p)\) (classical
observable) is known as the Weyl quantization or a Weyl
calculus~\cite{Folland89}*{\S~2.1}.  The Hamiltonian \(\tilde{H}\)
defines the dynamics of a quantum observable \(\tilde{k}\) by the
\emph{Heisenberg equation}:
\begin{equation}
  \label{eq:Heisenberg-dynamics}
   \rmi\myh \frac{\rmd \tilde{k}}{\rmd t}=\tilde{H} \tilde{k} - \tilde{k} \tilde{H}.
\end{equation}
This is sketch of the well-known construction of quantum mechanics
from infinite-di\-men\-sional UIRs of the Heisenberg group, which can
be found in numerous sources~\cites{Kisil02e,Folland89,Howe80b}.
%

\subsection{Classical Noncommutativity}
\label{sec:class-noncomm}

Now we are going to show that the priority of importance in quantum
theory shall be shifted from the Planck constant towards the imaginary
unit.  Namely, we describe a model of \emph{classical} mechanics with a
\emph{non-zero} Planck constant but with a different hypercomplex unit.
Instead of the imaginary unit with the property \(\rmi^2=-1\) we will
use the nilpotent unit \(\rmp\) such that \(\rmp^2=0\). The \emph{dual
  numbers} generated by nilpotent unit were already known for there
connections with Galilean relativity~\cites{Yaglom79,Gromov90a}\,--\,the
fundamental symmetry of classical mechanics\,--\,thus its appearance in
our discussion shall not be very surprising after all. Rather, we may
wander why the following construction was unnoticed for such a
long time.

Another important feature of our scheme is that the classical
mechanics is presented by a noncommutative model. Therefore, it will
be a refutation of Dirac's claim about the exclusive r\^ole of
noncommutativity for quantum theory. Moreover, the model is developed
from the same Heisenberg group, which were used above to describe the
quantum mechanics.

Consider a four-dimensional algebra \(\algebra{C}\) spanned by
\(1\), \(\rmi\), \(\rmp\) and \(\rmi\rmp\).  We can define the following
representation \(\uir{}{\rmp\myh}\) of the Heisenberg group in a space
of \(\algebra{C}\)-valued smooth functions~\cites{Kisil10a,Kisil11c}:
\begin{eqnarray}
  \label{eq:dual-repres0}
  \lefteqn{\uir{}{\rmp\myh}(s,x,y):\  f(q,p) \mapsto}\\
  &&  \rme^{-2\pi\rmi(xq+yp)}\left(f(q,p)
    +\rmp\hbar \left(2\pi s f(q,p)
      -\frac{\rmi y}{2}f'_q(q,p)+\frac{\rmi x}{2}f'_p(q,p)\right)\right).
  \nonumber 
\end{eqnarray}
A simple calculation shows the representation property
\begin{displaymath}
\uir{}{\rmp\myh}(s,x,y)
\uir{}{\rmp\myh}(s',x',y')=\uir{}{\rmp\myh}((s,x,y)*(s',x',y'))  
\end{displaymath}
for the multiplication~\eqref{eq:H-n-group-law0} on \(\Space{H}{}\).
Since this is not a unitary representation in a complex-valued Hilbert
space its existence does not contradict the Stone--von~Neumann
theorem.  Both representations~\eqref{eq:stone-inf0}
and~\eqref{eq:dual-repres0} are \emph{noncommutative} and act on
functions over the
phase space. The important distinction is:
\begin{itemize}
\item The representation~\eqref{eq:stone-inf0} is induced (in the
  sense of Mackey~\cite{Kirillov76}*{\S~13.4}) by the
  \emph{complex-valued} unitary character
  \(\uir{}{\myhbar}(s,0,0)= \rme^{2\pi\rmi\myhbar s}\)
  of the centre of \(\Space{H}{}\).
\item The representation~\eqref{eq:dual-repres0} is similarly induced
  by the \emph{dual number-valued} character
  \(\uir{}{\rmp\myh}(s,0,0)= \rme^{\rmp\myh s}=1+\rmp\myh s\) of the
  centre of \(\Space{H}{}\), cf.~\cite{Kisil09c}. Here dual numbers
  are the associative and commutative two-dimensional algebra spanned
  by \(1\) and \(\rmp\).
\end{itemize}

Similarity between~\eqref{eq:stone-inf0} and~\eqref{eq:dual-repres0}
is even more striking  if~\eqref{eq:dual-repres0} is written\footnote{I am grateful to Prof.~N.A.Gromov, who
  suggested this expression.} 
as:
\begin{equation}
  \label{eq:dual-as-SB}
  \uir{}{\myhbar}(s,x,y): f (q,p) \mapsto 
   \rme^{-2\pi(\rmp \myhbar s+\rmi(qx+py))}
  f \left(q-\frac{\rmi\myhbar}{2} \rmp y, p+\frac{\rmi\myhbar}{2} \rmp
    x\right).
\end{equation}
Here, for a differentiable function \(k\) of a real variable \(t\),
the expression \(k(t+\rmp a)\) is understood as \(k(t)+\rmp a k'(t)\),
where \(a\in\Space{C}{}\) is a constant. For a real-analytic function
\(k\) this can be justified through its Taylor's expansion,
see~\citelist{\cite{CatoniCannataNichelatti04}
  \cite{Zejliger34}*{\S~I.2(10)} \cite{Gromov90a} \cite{Dimentberg78a}
\cite{Dimentberg78b}}. The same expression appears
within the non-standard analysis based on the idempotent unit
\(\rmp\)~\cite{Bell08a}.

The infinitesimal generators of one-parameter subgroups
\(\uir{}{\rmp\myh}(0,x,0)\) and \(\uir{}{\rmp\myh}(0,0,y)\) in~\eqref{eq:dual-repres0} are
\begin{displaymath}
 \rmd\uir{X}{\rmp\myh}= -2\pi\rmi q-\frac{\rmp\myh}{4\pi\rmi}\partial_p \quad \text{ and }
  \quad
 \rmd\uir{Y}{\rmp\myh}= -2\pi\rmi p+\frac{\rmp\myh}{4\pi\rmi}\partial_q,
\end{displaymath}
respectively. We calculate their commutator:
\begin{equation}
  \label{eq:dual-classical-commutator}
  \rmd\uir{X}{\rmp\myh}\cdot  \rmd\uir{Y}{\rmp\myh}-
   \rmd\uir{Y}{\rmp\myh}\cdot  \rmd\uir{X}{\rmp\myh}=\rmp\myh.
\end{equation}
It is similar to the Heisenberg
relation~\eqref{eq:heisenberg-comm-basic}: the commutator is non-zero
and is proportional to the Planck constant. The only difference is the
replacement of the imaginary unit by the nilpotent one. The radical
nature of this change becomes clear if we integrate this representation
with the Fourier transform \(\hat{H}(x,y)\) of a Hamiltonian function
\(H(q,p)\):
\begin{equation}
  \label{eq:classical-int-represe}
  \mathring{H} =  \int_{\Space{R}{2n}}\hat{H}(x,y)\,
  \uir{}{\rmp\myh}(0,x,y)\,\rmd x\,\rmd y
  =H+\frac{\rmp\myh}{2} \left(\frac{\partial  H}{\partial p}\frac{\partial\  }{\partial q}
    - \frac{\partial  H}{\partial q} \frac{\partial\ }{\partial p}\right). 
\end{equation}
This is a first order differential operator on the phase space. It 
generates a dynamics of a classical observable \(k\)\,--\,a smooth real-valued
function on the phase space\,--\,through the equation isomorphic to the
Heisenberg equation~\eqref{eq:Heisenberg-dynamics}:
\begin{displaymath}
  \rmp \myh \frac{\rmd  \mathring{k}}{\rmd  t}= \mathring{H} \mathring{k} - \mathring{k} \mathring{H}.
\end{displaymath}
Making a substitution from~\eqref{eq:classical-int-represe} and using the
identity \(\rmp^2=0\) we obtain:
\begin{equation}
  \label{eq:hamilton-poisson}
  \frac{\rmd  {k}}{\rmd  t} =\frac{\partial  H}{\partial p}\frac{\partial k }{\partial q}
    - \frac{\partial  H}{\partial q} \frac{\partial k }{\partial p}\, .
\end{equation}
This is, of course, the \emph{Hamilton equation} of classical
mechanics based on the \emph{Poisson bracket}.  Dirac suggested, see
the paper's epigraph, that  the commutator  \emph{corresponds} to the
Poisson bracket. However, the commutator in the
representation~\eqref{eq:dual-repres0} \emph{is exactly} the Poisson
bracket. 

Note also, that both the Planck constant and the nilpotent unit
disappeared from~\eqref{eq:hamilton-poisson}, however we did use the
fact \(\myh\neq 0\) to make this cancellation.  Also, the shy
disappearance of the nilpotent unit \(\rmp\) at the very last minute
can explain why its r\^ole remain unnoticed for a long time.

\subsection{Summary}
\label{sec:conclusions}

We revised mathematical foundations of quantum and classical
mechanics and the r\^ole of hypercomplex units \(\rmi^2=-1\) and
\(\rmp^2=0\) there. 
To make the consideration complete, one may wish
to consider the third logical possibility of the hyperbolic unit
\(\rmh\) with the property
\(\rmh^2=1\)~\cites{Hudson66a,Khrennikov09book,Kisil10a,Ulrych10a,Pilipchuk10a,Kisil12a,Kisil09c},
see Section~\ref{sec:hyperb-repr-addt}.

The above discussion
provides the following observations~\cite{Kisil12c}:
\begin{enumerate}
\item Noncommutativity is not a crucial prerequisite for
  quantum theory, it can be obtained as a consequence of other
  fundamental assumptions.
\item Noncommutativity is not a distinguished feature of quantum
  theory, there are noncommutative formulations of classical mechanics
  as well.
\item The non-zero Planck constant is compatible with classical
  mechanics. Thus, there is no a necessity to consider the
  semiclassical limit \(\myhbar \rightarrow 0\), where the
  \emph{constant} has to tend to zero.
\item There is no a necessity to request that physical observables
  form an algebra, which is a physical non-sense since we cannot add two
  observables of different dimensionalities. Quantization can be
  performed by the Weyl recipe, which requires only a structure of a
  linear space in the collection of all observables with the same
  physical dimensionality.   
\item It is the imaginary unit in~\eqref{eq:heisenberg-comm-basic},
  which is ultimately responsible for most of quantum effects.
  Classical mechanics can be obtained from the similar commutator
  relation~\eqref{eq:dual-classical-commutator} using the nilpotent
  unit \(\rmp^2=0\).
\end{enumerate}
In Dirac's opinion, quantum noncommutativity was so important because
it guaranties a non-trivial commutator, which is required to
substitute the Poisson bracket. In our model, multiplication of
classical observables is also non-commutative and the Poisson bracket
exactly is the commutator. Thus, these elements do not separate quantum
and classical models anymore.

Our consideration illustrates the following statement on the
exceptional r\^ole of the complex numbers in quantum
theory~\cite{Penrose78a}: 
\begin{quote}
  \ldots for the first time, the complex field \(\mathbb{C}\) was brought
  into physics at a fundamental and universal level, not just as a
  useful or elegant device, as had often been the case earlier for
  many applications of complex numbers to physics, but at the very
  basis of physical law.
\end{quote}

Thus, Dirac may be right that we need to change a single assumption
to get a transition between classical mechanics and quantum. But, it
shall not be a move from commutative to noncommutative. Instead, we
need to replace a representation of the Heisenberg group induced from
a dual number-valued character by the representation induced by a
complex-valued character. Our conclusion can be stated like a
proportionality: 
\begin{quote}
  Classical\(/\)Quantum\(=\)Dual numbers\(/\)Complex numbers.
\end{quote}

\section{Groups, Homogeneous Spaces and Hypercomplex Numbers}
\label{sec:groups-homog-spac}
This section shows that the group \(\SL\) naturally requires complex,
dual and double numbers to describe its action on homogeneous
space. And the group \(\SL\) acts by automorphism on the Heisenberg
group, thus the Heisenberg group is naturally linked to hypercomplex
numbers as well.

\subsection{The Group $\SL$ and Its Subgroups}
\label{sec:group-sl-its}

The \(\SL\) group%
\index{$\SL$ group}%
\index{group!$\SL$}~\cite{Lang85} consists of \(2\times 2\)
real matrices with unit determinant. This is the smallest semis-simple
Lie group, its Lie algebra is formed by zero-trace \(2\times 2\)
real matrices. The \wiki{Affine_transformation}{$ax+b$ group}%
\index{$ax+b$ group}%
\index{group!$ax+b$}%
\index{group!affine|see{$ax+b$ group}}%
\index{affine group|see{$ax+b$ group}}, which is used wavelet theory
and harmonic analysis~\cite{Kisil12d}, is only a subgroup of \(\SL\)
consisting of the upper-triangular matrices \(
 \begin{pmatrix}
   a^{1/2}&b\\0&a^{-1/2}
 \end{pmatrix}\).

Consider the Lie
algebra \(\algebra{sl}_2\) of the group \(\SL\). Pick up the following
basis in \(\algebra{sl}_2\)~\cite{MTaylor86}*{\S~8.1}:  
\begin{equation}
  \label{eq:sl2-basis}
  A= \frac{1}{2}
  \begin{pmatrix}
    -1&0\\0&1
  \end{pmatrix},\quad 
  B= \frac{1}{2} \
  \begin{pmatrix}
    0&1\\1&0
  \end{pmatrix}, \quad 
  Z=
  \begin{pmatrix}
    0&1\\-1&0
  \end{pmatrix}.
\end{equation}
The commutation relations between the elements are:
\begin{equation}
  \label{eq:sl2-commutator}
  [Z,A]=2B, \qquad [Z,B]=-2A, \qquad [A,B]=- \frac{1}{2} Z.
\end{equation} 

Any element \(X\)
of the Lie algebra \(\algebra{sl}_2\)
defines a one-parameter continuous subgroup \(A(t)\) of \(\SL\)
through the exponentiation: \(A(t)=\exp(tX)\).
There are only \emph{three} different types of such subgroups under
the matrix similarity \(A(t)\mapsto MA(t)M^{-1}\)
for some constant \(M\in\SL\).
\begin{proposition}
  \label{pr:ank-subgroups}
  Any continuous one-parameter subgroup of \(\SL\) is conjugate to one of
  the following subgroups:%
  \index{$A$ subgroup|indef}%
  \index{subgroup!$A$|indef}%
  \index{$N$ subgroup|indef}%
  \index{subgroup!$N$|indef}%
  \index{$K$ subgroup|indef}%
  \index{subgroup!$K$|indef}
  \begin{eqnarray}
    \label{eq:a-subgroup}
    A&=&\left\{  
      \begin{pmatrix}  \rme^{-t/2} & 0\\0& \rme^{t/2}
      \end{pmatrix}=\exp \begin{pmatrix} -t/2 & 0\\0&t/2
      \end{pmatrix},\  t\in\Space{R}{}\right\},\\
    \label{eq:n-subgroup}
    N&=&\left\{   {\begin{pmatrix} 1&t \\0&1
        \end{pmatrix}=\exp \begin{pmatrix} 0 & t\\0&0
        \end{pmatrix},}\  t\in\Space{R}{}\right\},\\
    \label{eq:k-subgroup}
    K&=&\left\{ {\begin{pmatrix}
          \cos t &  \sin t\\
          -\sin t & \cos t
        \end{pmatrix}=   \exp \begin{pmatrix} 0& t\\-t&0
        \end{pmatrix}},\ t\in(-\pi,\pi]\right\}.
  \end{eqnarray}
\end{proposition}

\subsection{Action of $\SL$ as a Source of Hypercomplex Numbers}
\label{sec:homogeneous-spaces}

We recall the following standard
construction~\cite{Kirillov76}*{\S~13.2}. Let \(H\)
be a closed subgroup of a Lie group \(G\).
Let \( \Omega=G / H\)
be the corresponding homogeneous space and \(s: \Omega \rightarrow G\)
be a smooth section, which is a right inverse to the natural
projection \(p: G\rightarrow \Omega \).
The choice of \(s\)
is inessential in the sense that by a smooth map
\(\Omega\rightarrow \Omega\) we can always reduce one to another.

Any \(g\in G\) has a unique decomposition of the form
\(g=s(\omega)h\), where \(\omega=p(g)\in \Omega\) and \(h\in H\).
Note that \(\Omega \) is a left homogeneous space with the
\(G\)-action defined in terms of \(p\) and \(s\) as follows:
\begin{equation}
  \label{eq:g-action}
  g: \omega  \mapsto g\cdot \omega=p(g* s(\omega))\notingiq
\end{equation}
where \(*\) is the multiplication on \(G\). This is also illustrated
by the following commutative diagram:
\begin{displaymath}
        \xymatrix{ 
        G \ar@<.5ex>[d]^p \ar[r]^{g*} & G \ar@<.5ex>[d]^p \\ 
        \Omega   \ar@<.5ex>[u]^s \ar[r]^{g\cdot}   & \Omega     \ar@<.5ex>[u]^s }
\end{displaymath}

We want to describe homogeneous spaces obtained from \(G=\SL\)
and \(H\)
be one-dimensional continuous subgroup of \(\SL\).
For \(G=\SL\),
as well as for other semisimple groups, it is common to consider only
the case of \(H\)
being the maximal compact subgroup \(K\).
However, in this paper we admit \(H\)
to be any one-dimensional continuous subgroup. Due to
Prop.~\ref{pr:ank-subgroups} it is sufficient to take \(H=K\),
\(N\)
or \(A\).
Then \(\Omega\)
is a two-dimensional manifold and for any choice of \(H\)
we define~\cite{Kisil97c}*{Ex.~3.7(a)}:
\begin{equation}
  \label{eq:s-map}
  s: (u,v) \mapsto
  \frac{1}{\sqrt{v}}
  \begin{pmatrix}
    v & u \\ 0 & 1
  \end{pmatrix}, \qquad (u,v)\in\Space{R}{2},\  v>0.
\end{equation}
A direct (or computer algebra~\cite{Kisil07a}) calculation show that:
\begin{proposition}
  \label{pr:sl2-act-brute}
  The \(\SL\) action~\eqref{eq:g-action} associated to the map
  \(s\)~\eqref{eq:s-map} is:
  \begin{equation}
    \label{eq:sl2-act-brute}
    \begin{pmatrix}
      a&b\\c&d
    \end{pmatrix}: (u,v)\mapsto
    \left(\frac{(au+b)(c u+d) -\sigma cav^2}{( c u+d)^2 -\sigma (cv)^2},
      \frac{v}{( c u+d)^2 -\sigma (cv)^2}\right)\notingiq
  \end{equation}
  where \(\sigma=-1\), \(0\) and \(1\) for the subgroups \(K\), \(N'\)
  and \(\Aprime\) respectively.
\end{proposition}
The expression in~\eqref{eq:sl2-act-brute} does not look very
appealing, however an introduction of hypercomplex numbers makes it
more attractive:
\begin{proposition}
  \label{pr:sl2-act}
  Let a hypercomplex unit \(\alli\) be such that \(\alli^2=\sigma\), then the
  \(\SL\) action~\eqref{eq:sl2-act-brute} becomes:
  \begin{equation}
    \label{eq:sl2-act}
    \begin{pmatrix}
      a&b\\c&d
    \end{pmatrix}: w\mapsto \frac{aw+b}{cw+d}, \qquad
    \text{where } w=u+\alli v\notingiq
  \end{equation}
  for all three cases parametrised by \(\sigma\) as in Prop.~\ref{pr:sl2-act-brute}.
\end{proposition}
\begin{remark}
  We wish to stress that the hypercomplex numbers were not introduced
  here by our intention, arbitrariness or ``generalising
  attitude''~\cite{Pontryagin86a}*{p.~4}. They were naturally created by
  the \(\SL\) action. 
\end{remark}

Notably the action~\eqref{eq:sl2-act} is a group homomorphism of the
group \(\SL\) into transformations of the ``upper half-plane'' on
hypercomplex numbers.  Although dual and double numbers are
algebraically trivial, the respective geometries in the spirit of
\wiki{Erlangen_program}{Erlangen programme} are refreshingly
inspiring~\cites{Kisil05a,Kisil12a,Kisil08a} and provide useful insights even in the
elliptic case~\cite{Kisil06a}.  In order to treat divisors of zero, we
need to consider M\"obius transformations~\eqref{eq:sl2-act} of
conformally completed plane~\cites{HerranzSantander02b,Kisil06b}.

The arithmetic of dual and double numbers is different from 
complex numbers mainly in the following aspects:
\begin{enumerate}
\item They have zero divisors%
  \index{divisor!zero}%
  \index{zero!divisor}. However, we are still able to define
  their transforms by~\eqref{eq:sl2-act} in most cases. The proper
  treatment of zero divisors will be done through corresponding
  compactification~\cite{Kisil12a}*{\S~8.1}.
\item They are not algebraically closed. However, this property of
  complex numbers is not used very often in analysis.
\end{enumerate}

 Three
possible values \(-1\), \(0\) and \(1\) of \(\sigma:=\alli^2\)%
\index{$sigma$@$\sigma$ ($\sigma:=\alli^2$)} will be referred to here as
\emph{elliptic}, \emph{parabolic} and \emph{hyperbolic}%
\index{elliptic!case}%
\index{hyperbolic!case}%
\index{parabolic!case}%
\index{case!elliptic}%
\index{case!parabolic}%
\index{case!hyperbolic} cases, respectively.  This separation into
three cases will be referred to as the \emph{EPH classification}%
\index{EPH classification}.  Unfortunately, there is a clash here with
the already established label for the \emph{Lobachevsky geometry}%
\index{Lobachevsky!geometry}%
\index{geometry!Lobachevsky}. It is often called hyperbolic geometry%
\index{hyperbolic!geometry}%
\index{geometry!hyperbolic} because it can be realised as a Riemann
geometry%
\index{Riemann!geometry}%
\index{geometry!Riemann} on a two-sheet hyperboloid. However, within
our framework, the Lobachevsky geometry should be called elliptic and
it will have a true hyperbolic counterpart.

\begin{notation}
  We denote the space \(\Space{R}{2}\) of vectors \(u+v \alli\) by
  \(\Space[e]{R}{}\)\index{$C@$\Space[e]{R}{}$ (complex numbers)}, %
  \(\Space[p]{R}{}\)\index{$D@$\Space[p]{R}{}$ (dual numbers)} or %
  \(\Space[h]{R}{}\)\index{$O@$\Space[h]{R}{}$ (double numbers)} to
  highlight which number system (complex, dual or double,
  respectively) is used. The notation \(\Space[\sigma]{R}{}\)%
  \index{$A@$\Space[\sigma]{R}{}$ (point space)} is used for a generic
  case. 
\end{notation}

\subsection{Orbits of the Subgroup Actions}
\label{sec:orbits-subgr-acti}

We start our investigation of the M\"obius
transformations~\eqref{eq:sl2-act}
\begin{displaymath}
  \begin{pmatrix}
    a&b\\c&d
  \end{pmatrix}:\ w\mapsto  \frac{aw+b}{c w+d}
\end{displaymath}
on the hypercomplex numbers \(w=u+\alli v\) from a description of
orbits produced by the subgroups \(\Aprime\), \(N'\) and \(K\). Due to the
Iwasawa decomposition%
\index{Iwasawa decomposition}%
\index{decomposition!Iwasawa}~\(\SL=ANK\), any M\"obius
transformation can be represented as a superposition of these three
actions.

\begin{figure}[htbp]
  \centering
  \includegraphics[scale=.8]{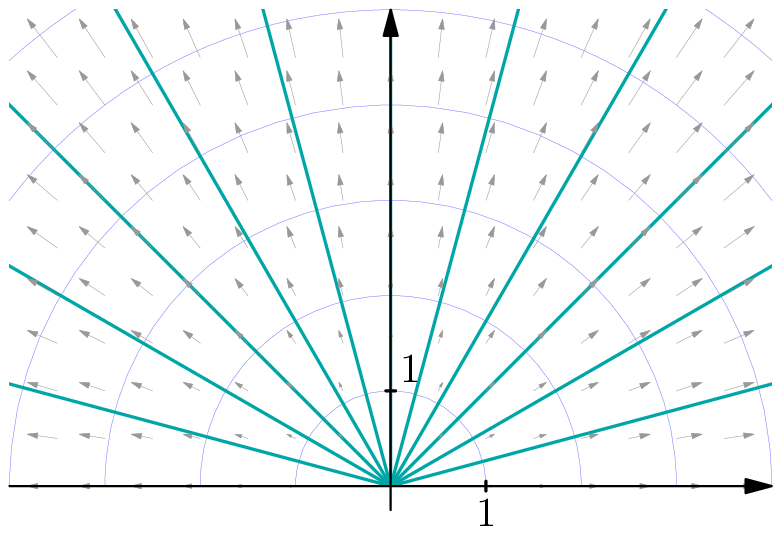}
  \hfill
   \includegraphics[scale=.8]{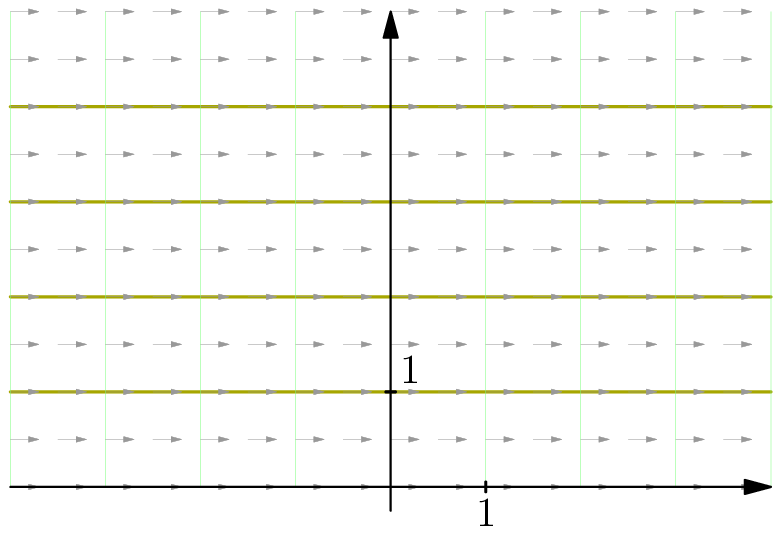}
   \caption[Actions of the subgroups $A$ and $N$ by M\"obius
   transformations]{Actions of the subgroups \(A\) and \(N\) by
     M\"obius transformations. Transverse thin lines are images of the
     vertical axis, grey arrows show the derived action.}
  \label{fig:a-n-action}
\end{figure}
The actions of subgroups \(A\) and \(N\) for any kind of hypercomplex
numbers on the plane are the same as on the real line: \(A\)%
\index{$A$-orbit|indef}%
\index{orbit!subgroup $A$, of|indef}%
\index{subgroup!$A$!orbit|indef} dilates and \(N\)%
\index{$N$-orbit|indef}%
\index{orbit!subgroup $N$, of|indef}%
\index{subgroup!$N$!orbit|indef} shifts\,--\,see Fig.~\ref{fig:a-n-action} for
illustrations. Thin traversal lines in Fig.~\ref{fig:a-n-action} join
points of orbits obtained from the vertical axis by the same values of
\(t\) and grey arrows represent ``local velocities''\,--\,vector fields
of derived representations.

\begin{figure}[htbp]
  \centering
  \includegraphics[scale=.8]{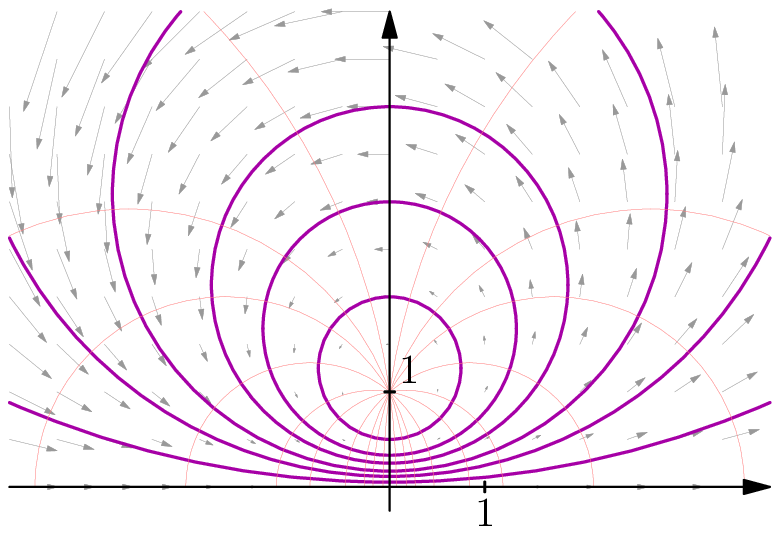} 
  \hfill
  \includegraphics[scale=.8]{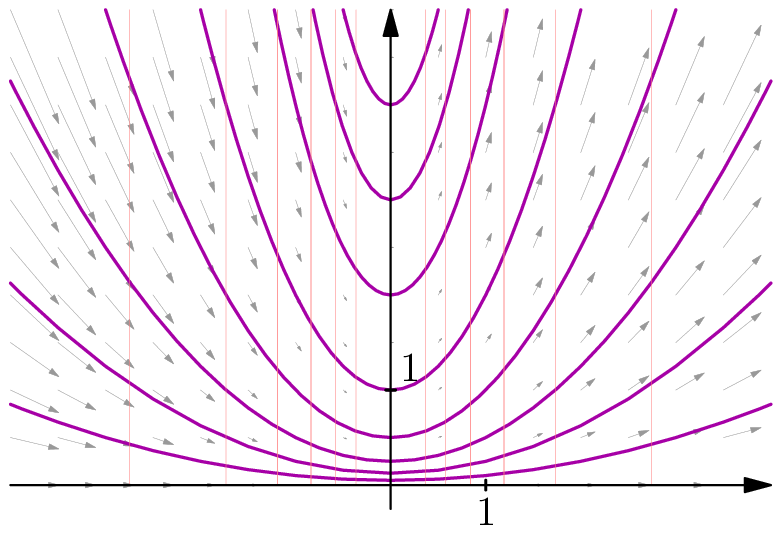}\\
  \includegraphics[scale=.8]{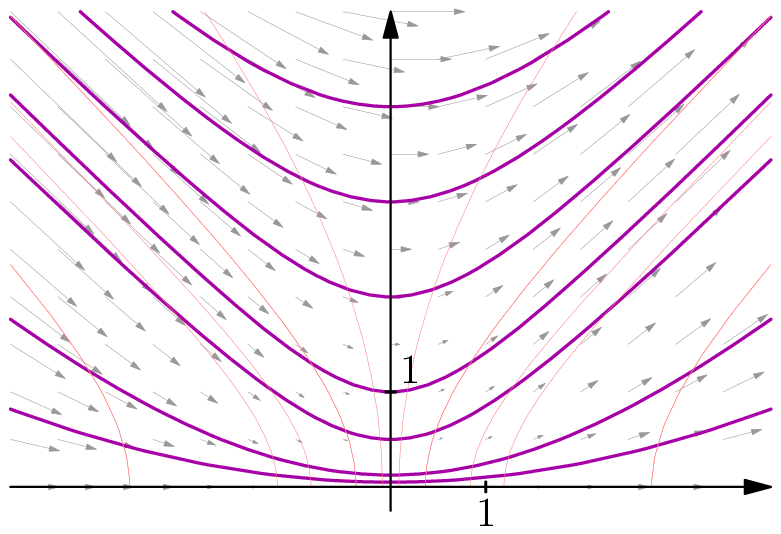}
  \caption[Action of the subgroup $K$]{Action of the subgroup \(K\).
    The corresponding orbits are circles, parabolas and hyperbolas
    shown by thick lines.  Transverse thin lines are images of the
    vertical axis, grey arrows show the derived action.}
  \label{fig:int-k-subgroup}
\end{figure}
By contrast, the action of the third matrix from the subgroup \(K\)
sharply depends on \(\sigma=\alli^2\), as illustrated by
Fig.~\ref{fig:int-k-subgroup}. In elliptic, parabolic and hyperbolic
cases, \(K\)-orbits%
\index{$K$-orbit|indef}%
\index{orbit!subgroup $K$, of|indef}%
\index{subgroup!$K$!orbit|indef} are circles, parabolas and
(equilateral) hyperbolas, respectively. The meaning of traversal lines
and vector fields is the same as on the previous figure.

\index{isotropy subgroup|(}%
\index{subgroup!isotropy|(} 

At the beginning of this subsection we described how
subgroups  generate homogeneous spaces.
The following exercise goes it in the
opposite way: from the group action on a homogeneous space to the
corresponding subgroup, which fixes the certain point.
\begin{exercise} 
  \label{le:fix-subgroups}
  Let \(\SL\) act by M\"obius transformations~\eqref{eq:sl2-act}
  on the three number systems.  Show that the isotropy subgroups
  of the point \(\alli\) are:
  \begin{enumerate}
  \item\label{it:fix-group-ell} 
    The subgroup \(K\) in the elliptic case. Thus, the
    elliptic upper half-plane%
    \index{upper half-plane!elliptic}%
    \index{half-plane!upper!elliptic}%
    \index{elliptic!upper half-plane} is a model for the homogeneous space
    \(\SL/K\).
  \item \label{it:fix-group-par}
    The subgroup \(N'\)%
    \index{$N'$ subgroup}%
    \index{subgroup!$N'$}~\eqref{eq:np-subgroup} of matrices
    \begin{equation}
      \label{eq:np-subgroup}
      \begin{pmatrix}
        1&0\\
         \nu & 1
      \end{pmatrix}
      =
      \begin{pmatrix}
        0&-1\\
        1 & 0
      \end{pmatrix}
      \begin{pmatrix}
        1&\nu \\
         0& 1
      \end{pmatrix}
      \begin{pmatrix}
        0&1\\
        -1 & 0
      \end{pmatrix}
    \end{equation}
    in the parabolic case.  It also fixes any point \(\rmp v\) on the
    vertical axis, which is the set of zero divisors%
    \index{divisor!zero}%
    \index{zero!divisor} in dual numbers.  The subgroup \(N'\) is
    conjugate to subgroup \(N\), thus the \emph{parabolic upper half-plane}%
    \index{upper half-plane!parabolic}%
    \index{half-plane!upper!parabolic}%
    \index{parabolic!upper half-plane} is
    a model for the homogeneous space \(\SL/N\).
  \item \label{it:fix-group-hyp}
    The subgroup \(\Aprime\)%
    \index{$A'$ subgroup}%
    \index{subgroup!$A'$}~\eqref{eq:ap-subgroup} of matrices
    \begin{equation}
      \label{eq:ap-subgroup}
      \begin{pmatrix}
        \cosh\tau & \sinh\tau\\
        \sinh\tau & \cosh\tau 
      \end{pmatrix} =
      \frac{1}{2}
      \begin{pmatrix}
        1 & -1\\
        1 & 1
      \end{pmatrix}
      \begin{pmatrix}
         \rme^\tau & 0\\
        0 &  \rme^{-\tau}
      \end{pmatrix}
      \begin{pmatrix}
        1 & 1\\
        -1 & 1
      \end{pmatrix}
    \end{equation}
    in the hyperbolic case. These transformations also fix the light
    cone%
    \index{light!cone}%
    \index{cone!light} centred at \(\rmh\), that is, consisting of
    \(\rmh+\text{zero divisors}\)%
    \index{divisor!zero}%
    \index{zero!divisor}.  The subgroup \(\Aprime\) is conjugate to
    the subgroup \(A\), thus two copies of the upper half-plane are a
    model for \(\SL/A\).
  \end{enumerate}
\end{exercise}
\begin{figure}[htbp]
  \centering
  \includegraphics[scale=.82]{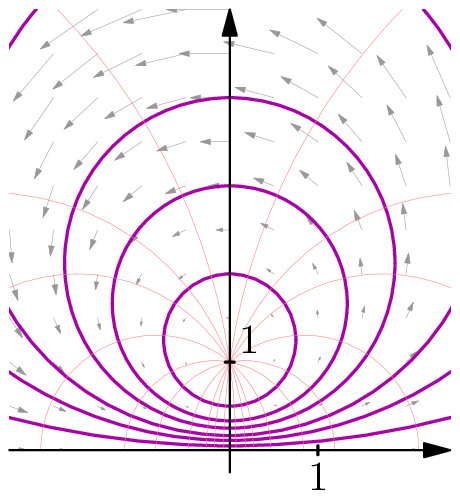} \hspace{0em plus 10em}%
  \includegraphics[scale=.82]{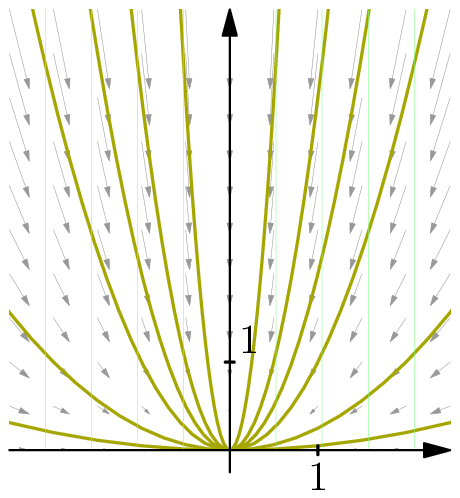} \hspace{0em plus 10em}%
  \includegraphics[scale=.82]{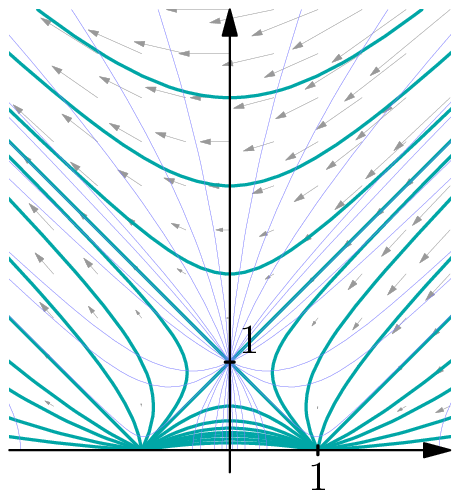}
  \caption[Actions of isotropy subgroups]{Actions of isotropy subgroups
    \(K\), \(N'\) and \(\Aprime\), which
    fix point \(\alli\) in three EPH cases.} 
  \label{fig:fix-sbroups}
\end{figure}
Figure~\ref{fig:fix-sbroups} shows actions of the above isotropic
subgroups on the respective numbers, we call them \emph{rotations}%
\index{rotation|indef} around \(\alli\). Note, that in parabolic and
hyperbolic cases they fix larger sets connected with zero divisors.

It is inspiring to compare the action of subgroups \(K\), \(N'\) and
\(\Aprime\) on three number systems, this is presented on
Fig.~\ref{fig:fix-groups-all-spaces}. Some features are preserved if
we move from top to bottom along the same column, that is, keep the
subgroup and change the metric of the space. We also note the same
system of a gradual transition if we compare pictures from left to
right along a particular row.  Note, that Fig.~\ref{fig:fix-sbroups}
extracts diagonal images from Fig.~\ref{fig:fix-groups-all-spaces},
this puts three images from Fig.~\ref{fig:fix-sbroups} into a context,
which is not obvious from Fig.~\ref{fig:fix-groups-all-spaces}.

\begin{figure}[htbp]
  \centering
  \includegraphics[scale=.82]{k_orb_ell1.eps} \hspace{1em plus 10em}%
  \includegraphics[scale=.82]{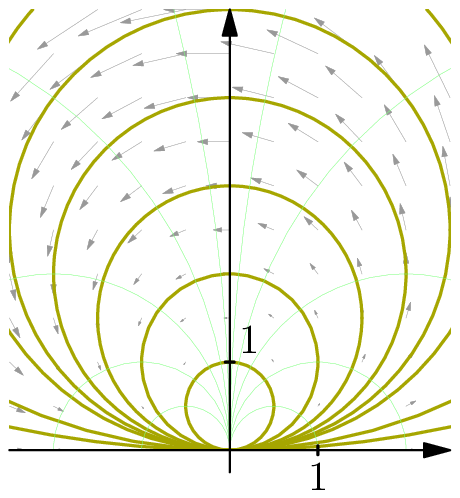} \hspace{1em plus 10em}%
  \includegraphics[scale=.82]{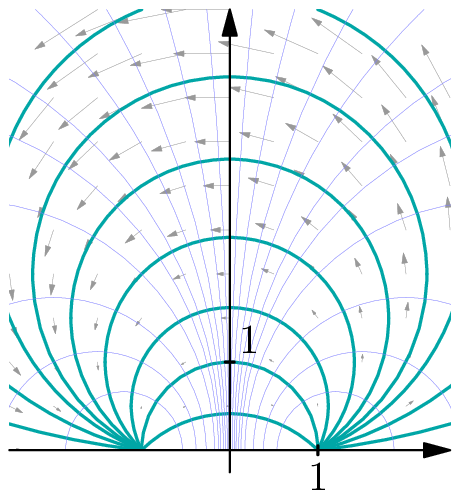}\\[2em]
  \includegraphics[scale=.82]{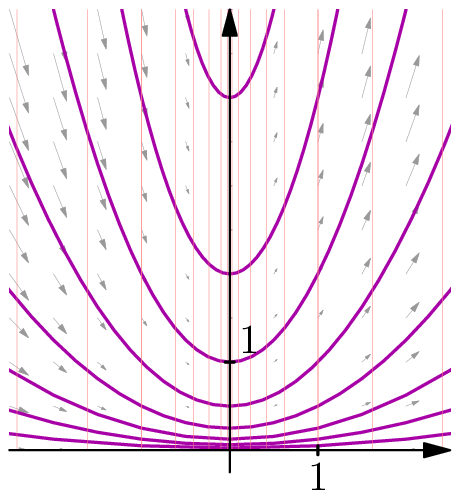} \hspace{1em plus 10em}%
  \includegraphics[scale=.82]{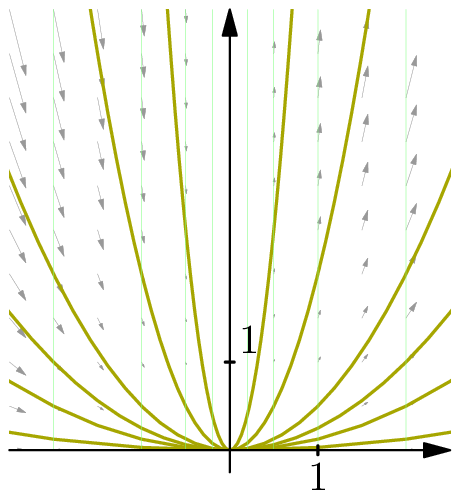} \hspace{1em plus 10em}%
  \includegraphics[scale=.82]{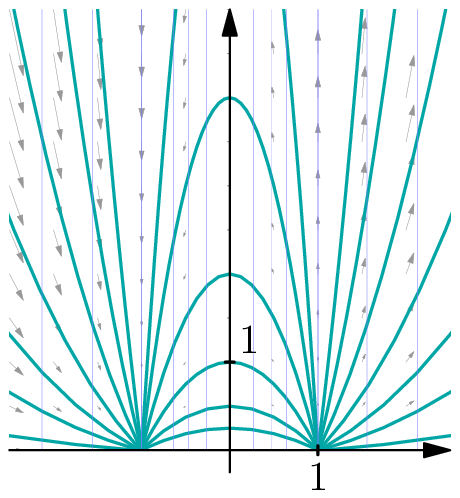}\\[2em]
  \includegraphics[scale=.82]{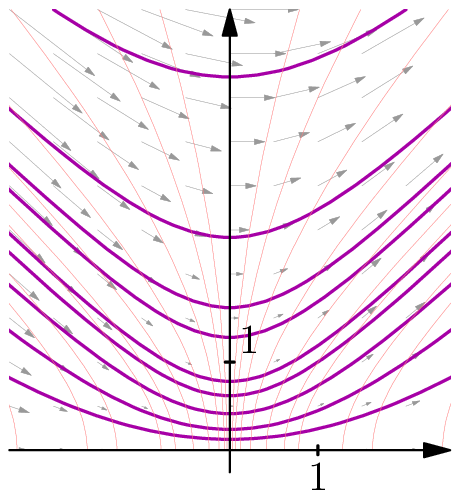} \hspace{1em plus 10em}%
  \includegraphics[scale=.82]{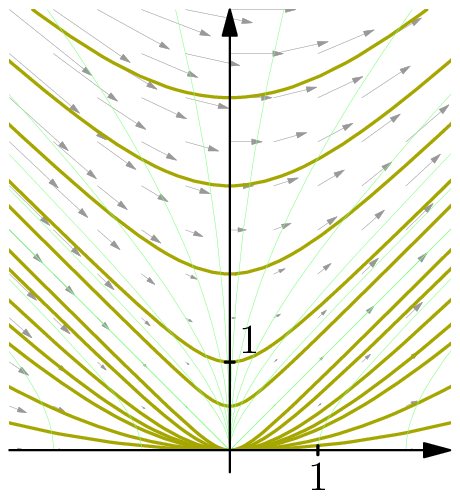} \hspace{1em plus 10em}%
  \includegraphics[scale=.82]{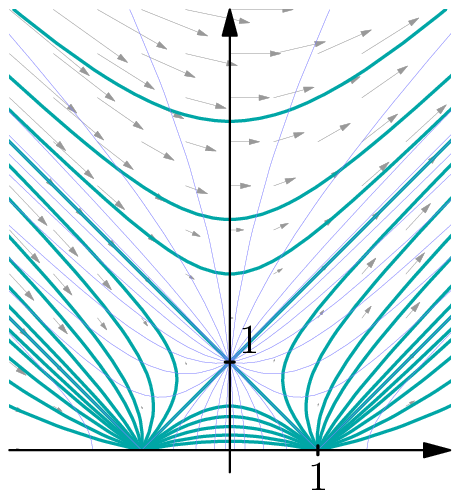}
  
  \caption[Actions of the subgroups $K$, $N'$, $\Aprime$]{Actions of
    the subgroups $K$, $N'$, $\Aprime$ are shown in the first, middle
    and last columns respectively. The elliptic, parabolic and
    hyperbolic spaces are presented in the first, middle and last rows
    respectively. The diagonal drawings comprise
    Fig.~\ref{fig:fix-sbroups} and the first column Fig.~\ref{fig:int-k-subgroup}.}
  \label{fig:fix-groups-all-spaces}
\end{figure}



\subsection{The Heisenberg Group and Symplectomorphisms}
\label{sec:heisenberg-group-intro}

Let \((s,x,y)\), where \(s\), \(x\), \(y\in \Space{R}{}\), be an
element of the one-dimensional \emph{Heisenberg group}
\(\Space{H}{}\)%
\index{Heisenberg!group}%
\index{group!Heisenberg}%
\index{$\Space{H}{}$, Heisenberg group}~\cites{Folland89,Howe80b}  also known as Weyl or Heisenberg-Weyl group.
Consideration of the general case of the \(n\)-dimensional Heisenberg
group \(\Space{H}{n}\) will be similar,
but is beyond the scope of present paper. The group law on
\(\Space{H}{}\) is given as follows:
\begin{equation}
  \label{eq:H-n-group-law}
  \textstyle
  (s,x,y)\cdot(s',x',y')=(s+s'+\frac{1}{2}\omega(x,y;x',y'),x+x',y+y')\notingiq
\end{equation} 
where the non-commutativity is due to \(\omega\)\,--\,the
\emph{symplectic form}%
\index{symplectic!form}%
\index{form!symplectic} on \(\Space{R}{2}\)~\eqref{eq:symplectic-form}, which is the central
object of the classical mechanics%
\index{classical mechanics}%
\index{mechanics!classical}~\cite{Arnold91}*{\S~37}:
\begin{equation}
  \label{eq:symplectic-form1}
  \omega(x,y;x',y')=xy'-x'y.
\end{equation}
The Heisenberg group is a non-commutative Lie
group with the centre
\begin{displaymath}
  Z=\{(s,0,0)\in \Space{H}{}, \ s \in \Space{R}{}\}.
\end{displaymath}
The left shifts
\begin{equation}
  \label{eq:left-right-regular}
  \Lambda(g): f(g') \mapsto f(g^{-1}g')  
\end{equation}
act as a representation of \(\Space{H}{}\) on a certain linear space
of functions. For example, an action on \(\FSpace{L}{2}(\Space{H}{},dg)\) with
respect to the Haar measure%
\index{Heisenberg!group!invariant measure}%
\index{group!Heisenberg!invariant measure}%
\index{invariant!measure}%
\index{measure!invariant} \(dg=\rmd s\,\rmd x\,\rmd y\) is the \emph{left regular}
representation%
\index{left regular representation}%
\index{representation!left regular}, which is unitary. 

The Lie algebra \(\algebra{h}\) of \(\Space{H}{}\) is spanned by
left-(right-)invariant vector fields
\begin{equation}
\textstyle  S^{l(r)}=\pm{\partial_s}, \quad
  X^{l(r)}=\pm\partial_{ x}-\frac{1}{2}y{\partial_s},  \quad
 Y^{l(r)}=\pm\partial_{y}+\frac{1}{2}x{\partial_s}
  \label{eq:h-lie-algebra}
\end{equation}
on \(\Space{H}{}\) with the Heisenberg \emph{commutator relation}%
\index{Heisenberg!commutator relation}%
\index{commutator relation}
\begin{equation}
  \label{eq:heisenberg-comm}
  [X^{l(r)},Y^{l(r)}]=S^{l(r)} 
\end{equation}
and all other commutators vanishing. This is encoded in the phrase
\(\Space{H}{}\)
is a \emph{nilpotent step \(2\)}
Lie group. For simplicity, we will sometimes omit the superscript
\(l\) for left-invariant field.

The group of outer automorphisms of \(\Space{H}{}\),
which trivially acts on the centre of \(\Space{H}{}\),
is the symplectic group \(\Sp[2]\)\index{$\Sp[2]$}
It is the group of symmetries of the symplectic form
\(\omega\)~\eqref{eq:symplectic-form1}
~\citelist{\cite{Folland89}*{Thm.~1.22} \cite{Howe80a}*{p.~830}}. The
symplectic group is isomorphic to \(\SL\)
considered in Sec.~\ref{sec:homogeneous-spaces}. The explicit action of \(\Sp[2]\)
on the Heisenberg group is:
\begin{equation}
  \label{eq:sympl-auto}
  g: h=(s,x,y)\mapsto g(h)=(s,x',y')\notingiq
\end{equation}
where 
\begin{displaymath}
  g=\begin{pmatrix}
    a&b\\
    c&d
  \end{pmatrix}\in\SL, \quad\text{ and }\quad
  \begin{pmatrix}
    x'\\y'
  \end{pmatrix}
  =\begin{pmatrix}
    a&b\\
    c&d
  \end{pmatrix}
  \begin{pmatrix}
    x\\y
  \end{pmatrix}.
\end{displaymath}
Due to appearance of half-integer weight in the Shale--Weil
representation below, we need to consider the metaplectic group
\(\Mp\) which is the double cover of  \(\Sp[2]\). 
Then we can build the semidirect product
\(G=\Space{H}{}\rtimes\Mp\) with the standard group law:
\begin{equation}
  \label{eq:schrodinger-group}
  (h,g)*(h',g')=(h*g(h'),g*g'), \qquad \text{where } 
  h,h'\in\Space{H}{}, \quad g,g'\in\Mp\notingiq
\end{equation}
and the stars denote the respective group operations while the action
\(g(h')\) is defined as the composition of the projection map
\(\Mp\rightarrow {\mathrm{Sp}}(2)\) and the
action~\eqref{eq:sympl-auto}. This group is sometimes called the
\emph{Schr\"odinger group}%
\index{Schr\"odinger!group}%
\index{group!Schr\"odinger} and it is known as the maximal kinematical
invariance group of both the free Schr\"odinger equation and the
quantum harmonic oscillator~\cite{Niederer73a}. This group is of
interest not only in quantum mechanics but also in
optics\index{optics}~\cites{ATorre10a,ATorre08a}. 

Consider the Lie algebra \(\algebra{sl}_2\)
of the group \(\SL\)
(as well as groups \(\Sp[2]\) and \(\Mp\)).
We again use the basis \(A\),
\(B\),
\(Z\)~\eqref{eq:sl2-basis}
with commutators~\eqref{eq:sl2-commutator}.  Vectors \(Z\),
\(B-Z/2\)
and \(B\)
are generators of the one-parameter subgroups \(K\),
\(N'\)
and \(\Aprime\)
\eqref{eq:k-subgroup}, \eqref{eq:np-subgroup} and
\eqref{eq:ap-subgroup} respectively.  Furthermore we can consider the
basis \(\{S, X, Y, A, B, Z\}\)
of the Lie algebra \(\algebra{g}\)
of the Lie group \(G=\Space{H}{}\rtimes\Mp\).
All non-zero commutators besides those already listed
in~\eqref{eq:heisenberg-comm} and~\eqref{eq:sl2-commutator} are:
\begin{align}
  \label{eq:cross-comm}
  [A,X]&=\textstyle\frac{1}{2}X,&
  [B,X]&=\textstyle-\frac{1}{2}Y,&
  [Z,X]&=Y;\\
  \label{eq:cross-comm1}
  [A,Y]&=\textstyle-\frac{1}{2}Y,&
  [B,Y]&=\textstyle-\frac{1}{2}X,&
  [Z,Y]&=-X.
\end{align}

\section{Linear Representations and Hypercomplex Numbers}
\label{sec:induc-repr}

A consideration of the symmetries in analysis is natural to start from
\emph{linear representations}%
\index{representations!linear}. The above geometrical
actions~\eqref{eq:sl2-act} can be naturally extended to such
representations by
induction%
\index{representation!induced}%
\index{induced!representation}~\citelist{\cite{Kirillov76}*{\S~13.2}
  \cite{Kisil97c}*{\S~3.1}} from a representation of a subgroup \(H\).
If \(H\) is one-dimensional then its irreducible representation is a
character, which is commonly supposed to be a complex valued.  However,
hypercomplex number naturally appeared in the \(\SL\)
action~\eqref{eq:sl2-act}, see~\cite{Kisil09c,Kisil12a}, why
shall we admit only \(\rmi^2=-1\) to deliver a character then?

\subsection{Hypercomplex Characters}
\label{sec:hyperc-char}

\begin{figure}[htbp]
  \centering
  \includegraphics[scale=.77]{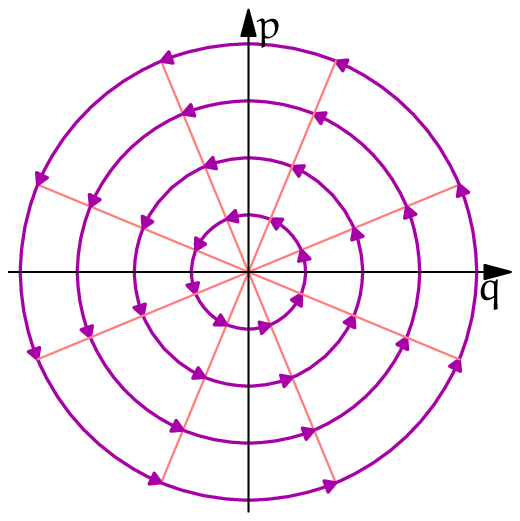}\hfill
  \includegraphics[scale=.77]{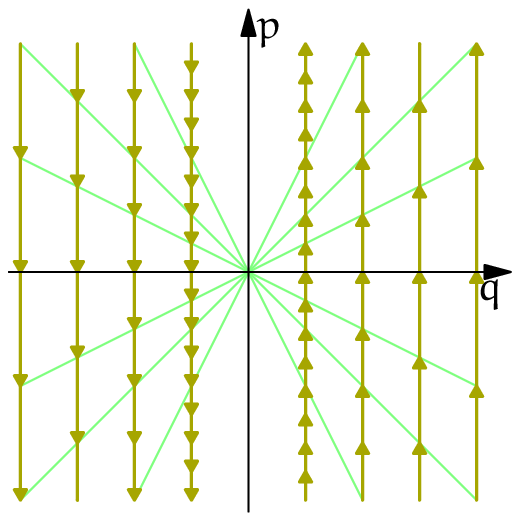}\hfill
  \includegraphics[scale=.77]{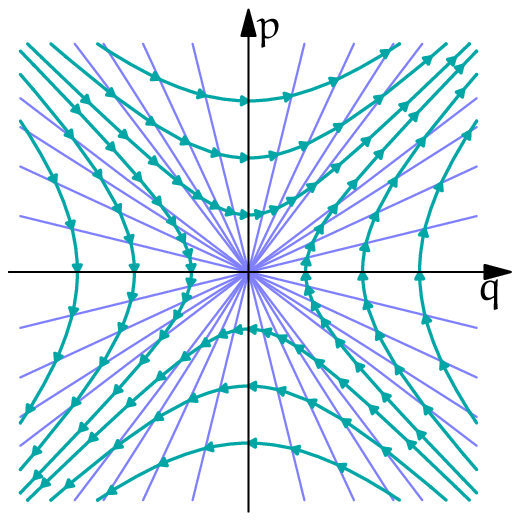}
  \caption[Rotations of wheels]{Rotations of algebraic wheels, i.e.
    the multiplication by \( \rme^{\alli t}\): elliptic (\(E\)), trivial
    parabolic (\(P_0\)) and hyperbolic (\(H\)). All blue orbits are
    defined by the identity \(x^2-\alli^2y^2=r^2\). Thin ``spokes''
    (straight lines from the origin to a point on the orbit) are
    ``rotated'' from the real axis. This is symplectic linear
    transformations of the classical phase space as well.}
  \label{fig:rotations}
\end{figure}%
\index{symplectic!transformation}%
\index{transformation!symplectic}
As we already mentioned, the typical discussion of induced
representations of \(\SL\) is centred around the case \(H=K\) and a
complex valued character of \(K\).  A linear transformation defined by
a matrix~\eqref{eq:k-subgroup} in \(K\) is a rotation of
\(\Space{R}{2}\) by the angle \(t\).  After identification
\(\Space{R}{2}=\Space{C}{}\) this action is given by the
multiplication \( \rme^{\rmi t}\), with \(\rmi^2=-1\).  The rotation
preserve the (elliptic) metric given by:
\begin{equation}
  \label{eq:ell-metric}
  x^2+y^2=(x+\rmi y)(x-\rmi y).
\end{equation}
Therefore the orbits of rotations are circles, any line passing the
origin (a ``spoke'') is rotated by the angle \(t\).  Dual%
\index{dual!number}%
\index{number!dual} and double numbers%
\index{number!double}%
\index{double!number} produces the most straightforward adaptation of
this result, see Fig.~\ref{fig:rotations} for all three cases.  The
correspondences between the respective algebraic aspects is shown at
Fig.~\ref{fig:algebr-corr-CDO}.

\begin{figure}[htbp]
  \centering
    \begin{tabular}{||c|c|c||}
      \hline\hline
      Elliptic & Parabolic & Hyperbolic\\
      \hline\hline
      \(\strut^{\strut}\rmi^{2}=-1\)&    \(\rmp^2=0\)&    \(\rmh^2=1\)
      \\
      \hline
      \(w=x+\rmi y\) &\(w=x+\rmp y\) &\(w=x+\rmh y\) 
      \\
      \hline
      \(\bar{w}=x-\rmi y\) &\(\bar{w}=x-\rmp y\) &\(\bar{w}=x-\rmh y\) 
      \\
      \hline
      \(\strut^{\strut} \rme^{\rmi t} = \cos t +\rmi \sin t\)&\( \rme^{\rmp t} = 1 +\rmp t\)&\( \rme^{\rmh t} = \cosh t +\rmh \sinh t\)
      \\
      \hline
      \(\strut^{\strut}\modulus[e]{w}^{ 2}=w\bar{w}=x^2+y^2\) &\(\modulus[p]{w}^2=w\bar{w}=x^2\) &\(\modulus[h]{w}^2=w\bar{w}=x^2-y^2\) 
      \\
      \hline
      \(\arg w = \tan^{-1} \frac{y}{x} \frac{\strut}{\strut}\)& \(\arg w = \frac{y}{x}\)&\(\arg w = \tanh^{-1} \frac{y}{x}\)
      \\
      \hline
       unit circle \(\strut^{\strut} \modulus[e]{w}^{2}=1\) & ``unit'' strip \(x=\pm 1\) & unit hyperbola \(\modulus[h]{w}^2=1\)
      \\
      \hline\hline
    \end{tabular}
    \caption[Algebraic correspondence between complex, dual and double
    numbers]{Algebraic correspondence between complex, dual and double
    numbers.}
  \label{fig:algebr-corr-CDO}
  \end{figure}  

Explicitly, parabolic rotations associated with \(\rme^{\rmp t}\) acts on dual
numbers%
\index{dual!number|(}%
\index{number!dual|(} as follows:
\begin{equation}
  \label{eq:parab-rot-triv}
  \rme^{\rmp x}: a+\rmp b \mapsto a+\rmp (a x+b).
\end{equation}
This links the parabolic case with the Galilean
group~\cite{Yaglom79} of symmetries of the classic mechanics, with
the absolute time disconnected from space.

The obvious algebraic similarity and the connection to classical
kinematic is a wide spread justification for the following viewpoint on
the parabolic case, cf.~\cites{HerranzOrtegaSantander99a,Yaglom79}:
\begin{itemize}
\item The parabolic trigonometric functions are trivial:
  \begin{equation}
    \label{eq:par-trig-0}
    \cosp t =\pm 1, \qquad \sinp t=t.
  \end{equation}
\item The parabolic distance is independent from \(y\) if \(x\neq 0\):
  \begin{equation}
    \label{eq:par-metr-0}
    x^2=(x+\rmp y)(x-\rmp y).
  \end{equation}
\item The polar decomposition of a dual number is defined by~\cite{Yaglom79}*{App.~C(30')}:
  \begin{equation}
    \label{eq:p-polar-yaglom}
    u+\rmp v = u(1+\rmp \frac{v}{u}), \quad \text{ thus }
    \quad \modulus{u+\rmp v}=u, \quad \arg(u+\rmp v)=\frac{v}{u}.
  \end{equation}
\item The parabolic wheel looks rectangular, see  Fig.~\ref{fig:rotations}.
\end{itemize}

Those algebraic analogies are quite explicit and widely accepted as an
ultimate source for parabolic
trigonometry~\cites{LavrentShabat77,HerranzOrtegaSantander99a,Yaglom79}.
Moreover, those three rotations are all non-isomorphic symplectic
linear transformations%
\index{symplectic!transformation}%
\index{transformation!symplectic}  of the phase space%
\index{phase!space}%
\index{space!phase}, which makes them useful in
the context of classical and quantum
mechanics~\cites{Kisil10a,Kisil11a}, see
Section~\ref{sec:quantum-mechanics-1}.  There exist also
alternative characters~\cite{Kisil09a} based on M\"obius
transformations with geometric motivation and connections 
to equations of mathematical physics.%
\index{dual!number|)}%
\index{number!dual|)}

\subsection{Induced Representations}
\label{sec:concl-induc-repr}

Let \(G\) be a group, \(H\) be its closed subgroup with the
corresponding homogeneous space \(X=G/H\) with an invariant measure.
Now we wish to linearise the action~\eqref{eq:g-action} through the
\wiki{Induced_representations}{induced
  representations}~\citelist{\cite{Kirillov76}*{\S~13.2}
  \cite{Kisil97c}*{\S~3.1}}.  We define a map \(\map{r}: G\rightarrow H\)
associated to the natural projection \(\map{p}: G \rightarrow G/H\)
and a continuous section \(\map{s}: G/H \rightarrow H\) from the identities:
\begin{equation}
  \label{eq:r-map}
  \map{r}(g)={(\map{s}(\omega))}^{-1}g, \qquad \text{where}\quad]
  \omega=\map{p}(g)\in \Omega .
\end{equation}
Let \(\chi\) be an irreducible representation of \(H\) in a vector space \(V\), then
it induces a representation of \(G\) in the sense of
Mackey~\cite{Kirillov76}*{\S~13.2}.  For a character \(\chi\)
of \(H\)
we can define a \emph{lifting}\index{lifting}
\(\oper{L}_{\chi}: \FSpace{L}{2}(G/H) \rightarrow
\FSpace[\chi]{L}{2}(G)\) as follows:
\begin{equation}
  \label{eq:lifting}
  [\oper{L}_{\chi} f](g)=\chi(\mathbf{r}(g))f(\mathbf{p}(g))\qquad
  \text{where}\quad] f(x)\in \FSpace{L}{2}(G/H). 
\end{equation}
The image space of the lifting \(\oper{L}_{\chi}\)
is invariant under left shifts.  We also define the
\emph{pulling}\index{pulling}
\(\oper{P}:\FSpace[\chi]{L}{2}(G) \rightarrow \FSpace{L}{2}(G/H)\),
which is a left inverse of the lifting and explicitly cab be given,
for example, by \([\oper{P}F](x)=F(\mathbf{s}(x))\).
Then the induced representation on \(\FSpace{L}{2}(G/H)\)
is generated by the formula
\(\uir{}{\chi}(g)=\oper{P}\circ\Lambda (g)\circ\oper{L}_\chi\).
This representation has the
realisation \(\uir{}{\chi}\) in the space 
of \(V\)-valued functions by the
formula~\cite{Kirillov76}*{\S~13.2.(7)--(9)}:
\begin{equation} 
  \label{eq:def-ind}
  [\uir{}{\chi}(g) f](\omega)= \chi(\map{r}(g^{-1} * \map{s}(\omega)))  f(g^{-1}\cdot \omega)\notingiq
\end{equation}
where \(g\in G\), \(\omega\in\Omega\), \(h\in H\) and \(\map{r}: G
\rightarrow H\), \(\map{s}: \Omega \rightarrow G\) are maps defined
above; \(*\)~denotes multiplication on \(G\) and \(\cdot\) denotes the
action~\eqref{eq:g-action} of \(G\) on \(\Omega\).

An alternative construction of induced representations is as
follow~\cite{Kirillov76}*{\S~13.2}. 

Let
\(\FSpace[\chi]{F}{2}(\Space{H}{n})\) be the space of functions on
\(\Space{H}{n}\) having the properties:
\begin{equation}
  \label{eq:induced-prop}
  f(gh)=\chi(h)f(g), \qquad \text{ for all}\quad] g\in \Space{H}{n},\ h\in Z
\end{equation}
and
\begin{equation}
  \label{eq:L2-condition}
  \int_{\Space{R}{2n}} \modulus{f(0,x,y)}^2\,\rmd x\,\rmd y<\infty.
\end{equation}
Then \(\FSpace[\chi]{F}{2}(\Space{H}{n})\) is invariant under the left
shifts and those shifts restricted to
\(\FSpace[\chi]{F}{2}(\Space{H}{n})\) make a representation
\(\uir{}{\chi}\) of \(\Space{H}{n}\) induced by \(\chi\).

Consider this scheme for representations of \(\SL\) induced from
characters of its one-dimensional subgroups. We can notice that only
the subgroup \(K\) requires a complex valued character due to the fact
of its compactness. For subgroups \(N'\) and \(\Aprime\) we can
consider characters of all three types---elliptic, parabolic and
hyperbolic.  Therefore we have seven essentially different induced
representations. We will write explicitly only three of them here.

\begin{example}%
  \index{representation!$\SL$ group|(}%
  \index{$\SL$ group!representation|(}
  Consider the subgroup \(H=K\), due to its compactness we are limited
  to complex valued characters of \(K\) only. All of them are of the
  form \(\chi_k\):
  \begin{equation}
    \label{eq:k-character}
    \chi_k\begin{pmatrix}
      \cos t &  \sin t\\
      -\sin t & \cos t
    \end{pmatrix}= \rme^{-\rmi k t}, \qquad \text{ where}\quad]
    k\in\Space{Z}{}.
  \end{equation}
  Using the explicit form~\eqref{eq:s-map} of the map \(s\) we find 
  the map \(r\) given in~\eqref{eq:r-map} as follows:
  \begin{displaymath}
    r
    \begin{pmatrix}
      a&b\\c&d
    \end{pmatrix}
    =\frac{1}{\sqrt{c^2+d^2}}
    \begin{pmatrix}
      d&-c\\c&d
    \end{pmatrix}\in K.
  \end{displaymath}
  Therefore:
  \begin{displaymath}
    \map{r}(g^{-1} * \map{s}(u,v))  =  
    \frac{1}{\sqrt{(c u+d)^2 +(cv)^2}}
    \begin{pmatrix}
      cu+d&-cv\\cv&cu+d
    \end{pmatrix}\notingiq
  \end{displaymath} 
  where  \(g^{-1}=    \begin{pmatrix}
      a&b\\c&d
    \end{pmatrix}\in\SL\).
    Substituting this into~\eqref{eq:k-character} and combining with
    the M\"obius transformation of the domain~\eqref{eq:sl2-act} we
    get the explicit realisation \(\uir{}{k}{}\) of the induced
    representation~\eqref{eq:def-ind}:
  \begin{equation}
    \label{eq:discrete}
    \uir{}{k}{}(g) f(w)=\frac{\modulus{cw+d}^k}{(cw+d)^k}f\left(\frac{aw+b}{cw+d}\right),
    \quad \text{ where}\quad] g^{-1}=\begin{pmatrix}a&b\\c&d
    \end{pmatrix}, \ w=u+\rmi v.
  \end{equation}
  This representation acts on complex valued functions in the upper
  half-plane \(\Space[+]{R}{2}=\SL/K\) and belongs to the discrete
  series%
  \index{representation!discrete series}%
  ~\cite{Lang85}*{\S~IX.2}.
  It is common to get rid of the factor \(\modulus{cw+d}^k\) from that
  expression in order to keep analyticity and we will follow this
  practise for a convenience as well.
\end{example}

\begin{example}
  \label{ex:n-induced}
  In the case of the subgroup \(N\) there is a wider choice of
  possible characters.
  \begin{enumerate}
  \item Traditionally only complex valued characters of the subgroup
    \(N\) are considered, they are:
    \begin{equation}
      \label{eq:np-character}
      \chi_{\tau}^{\Space{C}{}}\begin{pmatrix}
        1 &  0\\
        t & 1
      \end{pmatrix}= \rme^{\rmi \tau t}, \qquad \text{ where}\quad]
      \tau\in\Space{R}{}.
    \end{equation}
    A direct calculation shows that:
    \begin{displaymath}
      r
      \begin{pmatrix}
        a&b\\c&d
      \end{pmatrix}
      =
      \begin{pmatrix}
        1&0\\\frac{c}{d}&1
      \end{pmatrix}\in N'.
    \end{displaymath}
    Thus:
    \begin{equation}
      \label{eq:np-char-part}
      \map{r}(g^{-1}*\map{s}(u,v))=
      \begin{pmatrix}
        1&0\\\frac{cv}{d+cu}&1
      \end{pmatrix}, \quad\text{ where}\quad] g^{-1}=    \begin{pmatrix}
        a&b\\c&d
      \end{pmatrix}.
    \end{equation}
    A substitution of this value into the
    character~\eqref{eq:np-character} together with the M\"obius
    transformation~\eqref{eq:sl2-act} we obtain the next realisation of~\eqref{eq:def-ind}:
    \begin{displaymath}
      \uir{\Space{C}{}}{\tau}(g) f(w)= \exp\left(\rmi\frac{\tau c v}{cu+d} \right)
      f\left(\frac{aw+b}{cw+d}\right)\notingiq
    \end{displaymath}
    where \(w=u+\rmp v\) and \(g^{-1}=\begin{pmatrix}a&b\\c&d
    \end{pmatrix}\in\SL\).
    The representation acts on the space of \emph{complex} valued
    functions on the upper half-plane \(\Space[+]{R}{2}\), which is
    a subset of \emph{dual} numbers%
    \index{dual!number|(}%
    \index{number!dual|(} as a homogeneous space \(\SL/N'\).
    The mixture of complex and dual numbers in the same expression is
    confusing.
  \item The parabolic character \(\chi_{\tau}\) with the algebraic flavour
    is provided by multiplication~\eqref{eq:parab-rot-triv} with the
    dual number:
    \begin{displaymath}
      \chi_{\tau}\begin{pmatrix}
        1 &  0\\
        t & 1
      \end{pmatrix}= \rme^{\rmp \tau t}=1+\rmp \tau t, \qquad \text{ where}\quad]
      \tau\in\Space{R}{}.
    \end{displaymath}
    If we substitute the value~\eqref{eq:np-char-part} into this
    character, then we receive the representation:
    \begin{displaymath}
      \uir{}{\tau}{}(g) f(w)= \left(1+\rmp\frac{\tau c v}{cu+d} \right)
      f\left(\frac{aw+b}{cw+d}\right)\notingiq
    \end{displaymath}
    where \(w\), \(\tau\) and \(g\) are as above.  The representation
    is defined on the space of dual numbers valued functions on the
    upper half-plane of dual numbers.  Thus expression contains only
    dual numbers with their usual algebraic operations. Thus it is
    linear with respect to them.
  \end{enumerate}
\end{example}
All characters in the previous Example are unitary. Then, the general
scheme~\cite{Kirillov76}*{\S~13.2} implies
unitarity  of induced representations in suitable senses.
\begin{theorem}[\cite{Kisil09c}]
  \label{th:unitarity}
  Both representations of \(\SL\) from Example~\ref{ex:n-induced}
  are unitary on the space of function on the upper half-plane
  \(\Space[+]{R}{2}\) of dual numbers with the inner product:
  \begin{equation}
    \label{eq:inner-product}
    \scalar{f_1}{f_2}=\int_{\Space[+]{R}{2}} f_1(w)
    \bar{f}_2(w)\,\frac{\rmd u\,\rmd v}{v^2}, \qquad \text{ where}\quad] w=u+\rmp v\notingiq
  \end{equation}
  and we use the conjugation and multiplication of functions' values
  in algebras of complex and dual numbers for representations
  \(\uir{\Space{C}{}}{\tau}\) and \(\uir{}{\tau}\) respectively.
\end{theorem}
The inner product~\eqref{eq:inner-product} is positive defined for
the representation \(\uir{\Space{C}{}}{\tau}\) but is not for the
others. The respective spaces are parabolic cousins of the \emph{Krein
  spaces}%
\index{Krein!space}%
\index{space!Krein}~\cite{ArovDym08}, which are hyperbolic in our sense.
  \index{representation!$\SL$ group|)}%
  \index{$\SL$ group!representation|)}%
\index{dual!number|)}%
\index{number!dual|)}

\subsection{Similarity and Correspondence: Ladder Operators}
\label{sec:correspondence}

From the above observation we can deduce the following empirical
principle, which has a heuristic value.

\begin{principle}[Similarity and correspondence]%
  \index{similarity|see{principle of similarity and correspondence}}%
  \index{correspondence|see{principle of similarity and correspondence}}%
  \index{principle!similarity and correspondence}%
  \label{pr:simil-corr-principle}
  \begin{enumerate}
  \item Subgroups \(K\), \(N'\) and \(\Aprime\) play a similar r\^ole in the
    structure of the group \(\SL\) and its representations.
  \item The subgroups shall be swapped simultaneously with the
    respective replacement of hypercomplex%
    \index{number!hypercomplex}%
    \index{hypercomplex!number} unit \(\alli\).
  \end{enumerate}
\end{principle}
The first part of the Principle (similarity) does not look sound
alone. It is enough to mention that the subgroup \(K\) is compact (and
thus its spectrum is discrete) while two other subgroups are not.
However, in a conjunction with the second part (correspondence) the
Principle have received the following confirmations so far,
see~\cite{Kisil09c} for details:
\begin{itemize}
\item The action of \(\SL\) on the homogeneous space \(\SL/H\) for
  \(H=K\), \(N'\) or \(\Aprime\) is given by linear-fractional
  transformations of complex, dual or double numbers respectively.
  Fig.~\ref{fig:fix-groups-all-spaces} provides an illustration.
\item Subgroups \(K\), \(N'\) or \(\Aprime\) are isomorphic to the groups of
  unitary rotations of respective unit cycles in complex, dual or
  double numbers.   
\item Representations induced from subgroups \(K\), \(N'\) or \(\Aprime\)
  are unitary if the inner product spaces of functions with values in
  complex, dual or double numbers.
\end{itemize}
\begin{remark}
  The principle of similarity and correspondence resembles
  \emph{supersymmetry}%
  \index{supersymmetry} between bosons and fermions in particle physics, but
  we have similarity between three different types of entities in our case.
\end{remark}

\subsection{Ladder Operators}
\label{sec:ladder-operators}
We present another illustration to the Principle~\ref{pr:simil-corr-principle}. 
Let \(\uir{}{}\) be a representation of the group \(\SL\) in a space
\(V\). Consider the derived representation \(\rmd\uir{}{}\) of the Lie
algebra \(\algebra{sl}_2\)~\cite{Lang85}*{\S~VI.1}, that is:
\begin{equation}
  \label{eq:deriver-rep-defn}
  \rmd\uir{}{}(X)=\left.\frac{\rmd}{\rmd t} \uir{}{}( \rme^{tX})\right|_{t=0},\qquad
  \text{for any } \quad X\in \algebra{sl}_2.
\end{equation}
We also denote
\(\tilde{X}=\rmd\uir{}{}(X)\) for \(X\in\algebra{sl}_2\). To see the
structure of the representation \(\uir{}{}\) we can decompose the
space \(V\) into eigenspaces of the operator \(\tilde{X}\) for
a suitable \(X\in \algebra{sl}_2\).

\subsubsection{Elliptic Ladder Operators}
\label{sec:ellipt-ladd-oper}

It would not be surprising that we are going to consider three cases:
Let \(X=Z\) be a generator of the subgroup%
\index{generator!of a subgroup}%
\index{subgroup!generator}
\(K\)~\eqref{eq:k-subgroup}. Since this is a compact subgroup the
corresponding eigenspaces%
\index{eigenvalue} \(\tilde{Z} v_k=\rmi k v_k\) are
parametrised by an integer \(k\in\Space{Z}{}\).  The
\emph{raising/lowering} or \emph{ladder operators}%
\index{ladder operator|(}%
\index{operator!ladder|(}%
\index{operator!raising|see{ladder operator}}%
\index{operator!lowering|see{ladder operator}}%
\index{operator!creation|see{ladder operator}}%
\index{operator!annihilation|see{ladder operator}}%
\index{raising operator|see{ladder operator}}%
\index{lowering operator|see{ladder operator}}%
\index{creation operator|see{ladder operator}}%
\index{annihilation operator|see{ladder operator}}
\(\ladder{\pm}\)~\citelist{\cite{Lang85}*{\S~VI.2}
  \cite{MTaylor86}*{\S~8.2}} are defined by the following
commutation relations:
\begin{equation}
  \label{eq:raising-lowering}
  [\tilde{Z},\ladder{\pm}]=\lambda_\pm \ladder{\pm}. 
\end{equation}
In other words \(\ladder{\pm}\) are eigenvectors for operators 
\(\loglike{ad}Z\) of adjoint representation of \(\algebra{sl}_2\)~\cite{Lang85}*{\S~VI.2}.
\begin{remark}
  The existence of such ladder operators follows from the general
  properties of Lie algebras if the element \(X\in\algebra{sl}_2\)
  belongs to a \emph{Cartan subalgebra}%
  \index{Cartan!subalgebra}%
  \index{subalgebra!Cartan}. This is the case for vectors \(Z\)
  and \(B\), which are the only two non-isomorphic types of
  Cartan subalgebras in \(\algebra{sl}_2\). However, the third case
  considered in this paper, the parabolic vector \(B+Z/2\), does
  not belong to a Cartan subalgebra, yet a sort of ladder
  operators is still possible with dual number coefficients.
  Moreover, for the hyperbolic vector \(B\), besides the standard
  ladder operators an additional pair with double number
  coefficients will also be described.
\end{remark}

From the commutators~\eqref{eq:raising-lowering} we deduce that
\(\ladder{+} v_k\) are eigenvectors of \(\tilde{Z}\) as well:
\begin{eqnarray}
  \tilde{Z}(\ladder{+} v_k)&=&(\ladder{+}\tilde{Z}+\lambda_+\ladder{+})v_k=\ladder{+}(\tilde{Z}v_k)+\lambda_+\ladder{+}v_k
  =\rmi k \ladder{+}v_k+\lambda_+\ladder{+}v_k\nonumber \\
  &=&(\rmi k+\lambda_+)\ladder{+}v_k.\label{eq:ladder-action}
\end{eqnarray}
Thus action of ladder operators on respective eigenspaces can be
visualised by the diagram:
\begin{equation}
  \label{eq:ladder-chain-1D}
  \xymatrix@1{
    \ldots\, \ar@<.4ex>[r]^{\ladder{+}} &
    \,V_{\rmi k-\lambda}\,  \ar@<.4ex>[l]^{\ladder{-}}\ar@<.4ex>[r]^{\ladder{+}} &
    \,V_{\rmi k}\, \ar@<.4ex>[l]^{\ladder{-}} \ar@<.4ex>[r]^{\ladder{+}} &
    \,V_{\rmi k+ \lambda}\,\ar@<.4ex>[l]^{\ladder{-}}  \ar@<.4ex>[r]^{\ladder{+}}
    &
    \,\ldots\ar@<.4ex>[l]^{\ladder{-}}}
\end{equation}
Assuming \(\ladder{+}=a\tilde{A}+b\tilde{B}+c\tilde{Z}\) from the
relations~\eqref{eq:sl2-commutator} and defining
condition~\eqref{eq:raising-lowering} we obtain linear equations
with unknown \(a\), \(b\) and \(c\): 
\begin{displaymath}
  c=0, \qquad 2a=\lambda_+ b, \qquad -2b=\lambda_+ a.
\end{displaymath}
The equations have a solution if and only if \(\lambda_+^2+4=0\),
therefore the raising/lowering operators are 
\begin{equation}
  \label{eq:elliptic-ladder}
  \ladder{\pm}=\pm\rmi \tilde{A}+\tilde{B}.
\end{equation}

\subsubsection{Hyperbolic Ladder Operators}
\label{sec:hyperb-ladd-oper}
Consider the case \(X=2B\) of a generator of the subgroup
\(\Aprime\)~\eqref{eq:ap-subgroup}. The subgroup is not compact
and eigenvalues of the operator \(\tilde{B}\) can be arbitrary,
however raising/lowering operators are still
important~\citelist{\cite{HoweTan92}*{\S~II.1}
  \cite{Mazorchuk09a}*{\S~1.1}}.  We again seek a solution in the
form \(\ladder[h]{+}=a\tilde{A}+b\tilde{B}+c\tilde{Z}\) for the commutator
\([2\tilde{B},\ladder[h]{+}]=\lambda \ladder[h]{+}\). We will get the system:
\begin{displaymath}
  4c=\lambda a,\qquad
  b=0,\qquad
  {a}=\lambda c.
\end{displaymath}
A solution exists if and only if \(\lambda^2=4\). There are
obvious values \(\lambda=\pm 2\) with the ladder operators
\(\ladder[h]{\pm}=\pm2\tilde{A}+\tilde{Z}\),
see~\citelist{\cite{HoweTan92}*{\S~II.1}
  \cite{Mazorchuk09a}*{\S~1.1}}. Each indecomposable
\(\algebra{sl}_2\)-module is formed by a one-dimensional chain of
eigenvalues\index{eigenvalue} with a transitive action of ladder
operators.

Admitting double numbers%
\index{number!double}%
\index{double!number} we have an extra possibility to satisfy
\(\lambda^2=4\) with values \(\lambda=\pm2\rmh\).  Then there is an
additional pair of hyperbolic ladder operators
\(\ladder[\rmh]{\pm}=\pm2\rmh\tilde{A}+\tilde{Z}\), which shift eigenvectors
in the ``orthogonal'' direction to the standard operators \(\ladder[h]{\pm}\).
Therefore an indecomposable \(\algebra{sl}_2\)-module can be
parametrised by a two-dimensional lattice of eigenvalues on the
double number plane,  see
Fig.~\ref{fig:2D-lattice}

\begin{figure}[htbp]
  \centering
  \(  \xymatrix@R=2.5em@C=1.5em@M=.5em{
    & 
    \,\ldots\, \ar@<.4ex>[d]^{\ladder[\rmh]{+}} &  
    \,\ldots\, \ar@<.4ex>[d]^{\ladder[\rmh]{+}} & 
    \,\ldots\,  \ar@<.4ex>[d]^{\ladder[\rmh]{+}}  & 
    \\
    \ldots\, \ar@<.4ex>[r]^-{\ladder[h]{+}} & 
    \,V_{(n-2)+\rmh (k-2)}\,  \ar@<.4ex>[l]^-{\ladder[h]{-}}\ar@<.4ex>[r]^{\ladder[h]{+}}
    \ar@<.4ex>[u]^{\ladder[\rmh]{-}} \ar@<.4ex>[d]^{\ladder[\rmh]{+}} &  
    \,V_{n+\rmh (k-2)}\, \ar@<.4ex>[l]^{\ladder[h]{-}} \ar@<.4ex>[r]^{\ladder[h]{+}}
    \ar@<.4ex>[u]^{\ladder[\rmh]{-}} \ar@<.4ex>[d]^{\ladder[\rmh]{+}} & 
    \,V_{(n+2)+\rmh (k-2)}\,\ar@<.4ex>[l]^{\ladder[h]{-}}  \ar@<.4ex>[r]^-{\ladder[h]{+}}
    \ar@<.4ex>[u]^{\ladder[\rmh]{-}} \ar@<.4ex>[d]^{\ladder[\rmh]{+}}    & 
    \,\ldots\ar@<.4ex>[l]^-{\ladder[h]{-}}\\
    \ldots\, \ar@<.4ex>[r]^-{\ladder[h]{+}} & 
    \,V_{(n-2)+\rmh k}\,  \ar@<.4ex>[l]^-{\ladder[h]{-}}\ar@<.4ex>[r]^{\ladder[h]{+}}
    \ar@<.4ex>[u]^{\ladder[\rmh]{-}} \ar@<.4ex>[d]^{\ladder[\rmh]{+}} &  
    \,V_{n+\rmh k}\, \ar@<.4ex>[l]^{\ladder[h]{-}} \ar@<.4ex>[r]^{\ladder[h]{+}}
    \ar@<.4ex>[u]^{\ladder[\rmh]{-}} \ar@<.4ex>[d]^{\ladder[\rmh]{+}}& 
    \,V_{(n+2)+\rmh k}\,\ar@<.4ex>[l]^{\ladder[h]{-}}  \ar@<.4ex>[r]^-{\ladder[h]{+}}
    \ar@<.4ex>[u]^{\ladder[\rmh]{-}} \ar@<.4ex>[d]^{\ladder[\rmh]{+}}    & 
    \,\ldots\ar@<.4ex>[l]^-{\ladder[h]{-}}\\
    \ldots\, \ar@<.4ex>[r]^-{\ladder[h]{+}} & 
    \,V_{(n-2)+\rmh (k+2)}\,  \ar@<.4ex>[l]^-{\ladder[h]{-}}\ar@<.4ex>[r]^{\ladder[h]{+}}
    \ar@<.4ex>[u]^{\ladder[\rmh]{-}} \ar@<.4ex>[d]^{\ladder[\rmh]{+}} &  
    \,V_{n+\rmh (k+2)}\, \ar@<.4ex>[l]^{\ladder[h]{-}} \ar@<.4ex>[r]^{\ladder[h]{+}} 
    \ar@<.4ex>[u]^{\ladder[\rmh]{-}} \ar@<.4ex>[d]^{\ladder[\rmh]{+}}& 
    \,V_{(n+2)+\rmh (k+2)}\,\ar@<.4ex>[l]^{\ladder[h]{-}}  \ar@<.4ex>[r]^-{\ladder[h]{+}}
    \ar@<.4ex>[u]^{\ladder[\rmh]{-}} \ar@<.4ex>[d]^{\ladder[\rmh]{+}}    & 
    \,\ldots\ar@<.4ex>[l]^-{\ladder[h]{-}}\\
    & 
    \,\ldots\, \ar@<.4ex>[u]^{\ladder[\rmh]{-}} &  
    \,\ldots\, \ar@<.4ex>[u]^{\ladder[\rmh]{-}} & 
    \,\ldots\,  \ar@<.4ex>[u]^{\ladder[\rmh]{-}}  & }
  \)
  \caption[The action of hyperbolic ladder operators on a 2D
  lattice of eigenspaces]{The action of hyperbolic ladder operators on a 2D
    lattice of eigenspaces. Operators \(\ladder[h]{\pm}\) move the
    eigenvalues by \(2\), 
    making shifts in the horizontal direction. Operators
    \(\ladder[\rmh]{\pm}\) change the eigenvalues by \(2\rmh\), 
    shown as vertical shifts.}  
  \label{fig:2D-lattice}
\end{figure}
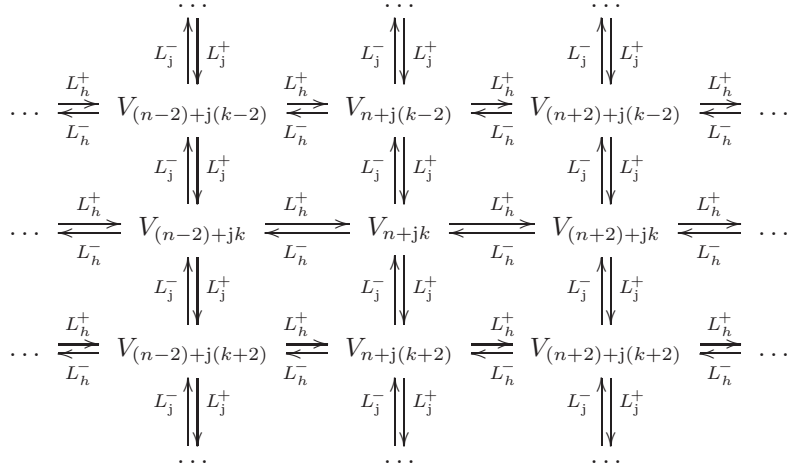

\subsubsection{Parabolic Ladder Operators}
\label{sec:parab-ladd-oper}

Finally consider the case of a generator \(X=-B+Z/2\) of the
subgroup \(N'\)~\eqref{eq:np-subgroup}. According to the above
procedure we get the equations:
\begin{displaymath}
  b+2c=\lambda a,\qquad
  -a=\lambda b,\qquad
  \frac{a}{2}=\lambda c\notingiq
\end{displaymath}
which can be resolved if and only if \(\lambda^2=0\). If we restrict
ourselves with the only real (complex) root \(\lambda=0\), then the
corresponding operators \(\ladder[p]{\pm}=-\tilde{B}+\tilde{Z}/2\)
will not affect eigenvalues and thus are useless in the above context.
However the dual number roots \(\lambda =\pm \rmp t\),
\(t\in\Space{R}{}\) lead to the operators \(\ladder[\rmp]{\pm}=\pm
\rmp t\tilde{A}-\tilde{B}+\tilde{Z}/2\). These operators are suitable
to build an \(\algebra{sl}_2\)-modules with a one-dimensional chain of
eigenvalues.
\begin{remark}
  \label{re:hyper-number-necessity}
  The following r\^oles of hypercomplex numbers are noteworthy%
  \index{number!hypercomplex}%
  \index{hypercomplex!number}:
  \begin{itemize}
  \item the introduction of complex numbers is a necessity for the
    \emph{existence} of ladder operators in the elliptic
    case;
  \item in the parabolic case we need dual numbers%
    \index{dual!number}%
    \index{number!dual} to make
    ladder operators \emph{useful};
  \item in the hyperbolic case double numbers%
    \index{number!double}%
    \index{double!number} are not required 
    neither for the existence or for the usability of ladder operators, but
    they do provide an enhancement. 
  \end{itemize}
\end{remark}
We summarise the above consideration with a focus on the Principle of
similarity and correspondence:%
\index{principle!similarity and correspondence}
\begin{proposition}
  \label{pr:ladder-sim-eq}
  Let a vector \(X\in\algebra{sl}_2\) generates the subgroup \(K\),
  \(N'\) or \(\Aprime\), that is \(X=Z\), \(B-Z/2\), or
  \(B\) respectively. Let \(\alli\) be the respective hypercomplex unit.   

  Then raising/lowering operators \(\ladder{\pm}\) satisfying to the
  commutation relation:
  \begin{displaymath}
    [X,\ladder{\pm}]=\pm\alli \ladder{\pm},\qquad [\ladder{-},\ladder{+}]=2\alli X.
  \end{displaymath}
  are:
  \begin{displaymath}
    \ladder{\pm}=\pm\alli \tilde{A} +\tilde{Y}.
  \end{displaymath}
  Here \(Y\in\algebra{sl}_2\) is a linear combination of  \(B\) and
  \(Z\) with the properties:
  \begin{itemize}
  \item \(Y=[A,X]\).
  \item \(X=[A,Y]\).
  \item Killings form \(K(X,Y)\)~\cite{Kirillov76}*{\S~6.2} vanishes.
  \end{itemize}
  Any of the above properties defines the vector \(Y\in\loglike{span}\{B,Z\}\)
  up to a real constant factor.
\end{proposition}
\index{ladder operator|)}%
\index{operator!ladder|)}

\subsection{Induced Representations of the Heisenberg Group}
\label{sec:induc-repr-heis}
In this subsection we calculate representations of of the Heisenberg
group induced by a complex valued character. Representations induced
by hypercomplex characters and their physical interpretation will be
discussed in the next section. 

Take a maximal (two dimensional) abelian subgroup
\(H=\{(s,0,y)\in\Space{H}{}\}\),
then the homogeneous space can be parametrised by a real number
\(x\).
We define the natural projection \(\map{p}(s,x,y)=x\)
and the continuous section \(\map{s}(x)=(0,x,0)\).
Then the map \(\map{r}: \Space{H}{1}\rightarrow H'_x\)
is \(\map{r}(s,x,y)=(s-\frac{1}{2}xy,0,y)\).
For the character \(\chi_\myhbar(s,0,y)= \rme^{2\pi \rmi (\myhbar s)}\),
the representation of \(\Space{H}{1}\)
on \(\FSpace{L}{2}(\Space{R}{1})\) is:
\begin{equation}
  \label{eq:H1-schroedinger-rep}
  \textstyle[\uir{}{\chi}(s,x,y)f](x')= \exp(2\pi \rmi (\myhbar
  (-s+yx'-\frac{1}{2}xy)))\, f(x'-x).
\end{equation}

Then the Fourier transform \(x\rightarrow q\)
produces the Schr\"odinger representation%
\index{representation!Heisenberg group!Schr\"odinger}%
\index{Schr\"odinger!representation}%
\index{representation!Schr\"odinger}~\cite{Folland89}*{\S~1.3} of
\(\Space{H}{}\)
in \(\FSpace{L}{2}(\Space{R}{})\), that is \cite{Kisil10a}*{(3.5)}:
\begin{equation}
  \label{eq:H1-schroedinger-rep-q-space}
    [\uir{}{\myhbar}(s,x,y) f\,](q)= \rme^{2\pi\rmi\myhbar (s-xy/2)
      +2\pi\rmi x q}\,f(q-\myhbar y).  
  \end{equation}
The variable \(q\)
is treated as the coordinate on the configurational space of a
particle.  The action of the derived representation on the Lie algebra
\(\algebra{h}\) is:
\begin{equation}
  \label{eq:schroedinger-rep-conf-der}
  \uir{}{\myhbar}(X)=2\pi\rmi q,\qquad \uir{}{\myhbar}(Y)=-\myhbar \frac{\rmd}{\rmd q},
  \qquad
  \uir{}{\myhbar}(S)=2\pi\rmi\myhbar I.
\end{equation}

The Shale--Weil theorem~\citelist{\cite{Folland89}*{\S~4.2}
  \cite{Howe80a}*{p.~830}} states that any representation
\(\uir{}{\myhbar}\) of the Heisenberg groups generates a unitary
\emph{oscillator}%
\index{oscillator!representation}%
\index{representation!oscillator!} (or \emph{metaplectic}) representation%
\index{representation!metaplectic|see{oscillator representation}}%
\index{metaplectic representation|see{oscillator representation}}%
\index{representation!Shale--Weil|see{oscillator representation}}%
\index{Shale--Weil representation|see{oscillator representation}}
\(\uir{\text{SW}}{\myhbar}\) of the \(\Mp\),
the two-fold cover of the symplectic group~\cite{Folland89}*{Thm.~4.58}.
The Shale--Weil theorem allows us to expand any representation
\(\uir{}{\myhbar}\) of the Heisenberg group to the representation
\(\uir{2}{\myhbar}=\uir{}{\myhbar}\oplus\uir{\text{SW}}{\myhbar}\) of the
group Schr\"odinger group~\eqref{eq:schrodinger-group}.
Of course, there is the derived form of the Shale--Weil representation
for \(\algebra{g}\). It can often be explicitly written in contrast
to the Shale--Weil representation.
\begin{example}
  TheShale--Weil representation of \(\SL\) in
  \(\FSpace{L}{2}(\Space{R}{})\)  associated  to the Schr\"odinger
  representation~\eqref{eq:H1-schroedinger-rep-q-space} has the
  derived action,  cf.~\citelist{\cite{ATorre08a}*{(2.2)} \cite{Folland89}*{\S~4.3}}:
  \begin{equation}
    \label{eq:shale-weil-der}
    \uir{\text{SW}}{\myhbar}(A) =-\frac{q}{2}\frac{\rmd}{\rmd q}-\frac{1}{4},\quad
    \uir{\text{SW}}{\myhbar}(B)=-\frac{\myhbar\rmi}{8\pi}\frac{\rmd ^2}{\rmd q^2}-\frac{\pi\rmi q^2}{2\myhbar},\quad
    \uir{\text{SW}}{\myhbar}(Z)=\frac{\myhbar\rmi}{4\pi}\frac{\rmd ^2}{\rmd q^2}-\frac{\pi\rmi q^2}{\myhbar}.
  \end{equation}
  We can verify commutators~\eqref{eq:sl2-commutator} and
  \eqref{eq:heisenberg-comm}, \eqref{eq:cross-comm1} for
  operators~\eqref{eq:schroedinger-rep-conf-der}--\eqref{eq:shale-weil-der}.
  It is also obvious that in this representation the following
  algebraic relations hold:
  \begin{eqnarray}
    \label{eq:quadratic-A}
    \qquad\uir{\text{SW}}{\myhbar}(A) &=&
    \frac{\rmi}{4\pi\myhbar}(\uir{}{\myhbar}(X)\uir{}{\myhbar}(Y)-{\textstyle\frac{1}{2}}\uir{}{\myhbar}(S))\\
    &=&\frac{\rmi}{8\pi\myhbar}(\uir{}{\myhbar}(X)\uir{}{\myhbar}(Y)+\uir{}{\myhbar}(Y)\uir{}{\myhbar}(X) ), \nonumber\\ 
    \label{eq:quadratic-B}
    \uir{\text{SW}}{\myhbar}(B) &=&
    \frac{\rmi}{8\pi\myhbar}(\uir{}{\myhbar}(X)^2-\uir{}{\myhbar}(Y)^2), \\
    \label{eq:quadratic-Z}
    \uir{\text{SW}}{\myhbar}(Z)
    &=&\frac{\rmi}{4\pi\myhbar}(\uir{}{\myhbar}(X)^2+\uir{}{\myhbar}(Y)^2). 
  \end{eqnarray}
  Thus it is common in quantum optics to name \(\algebra{g}\) as a Lie
  algebra with  quadratic generators%
  \index{generator!quadratic}%
  \index{quadratic!generator}, see~\cite{Gazeau09a}*{\S~2.2.4}.
\end{example}
Note that \(\uir{\text{SW}}{\myhbar}(Z)\) is the Hamiltonian of the
harmonic oscillator%
\index{harmonic!oscillator}%
\index{oscillator!harmonic} (up to a factor). Then we can consider
\(\uir{\text{SW}}{\myhbar}(B)\) as the Hamiltonian of a repulsive
(hyperbolic) oscillator. The operator
\(\uir{\text{SW}}{\myhbar}(B-Z/2)=\frac{\myhbar\rmi}{4\pi}\frac{\rmd ^2}{\rmd q^2}\)
is the parabolic analog. A graphical representation of all three
transformations defined by those Hamiltonian is given in
Fig.~\ref{fig:rotations} and a further discussion of these
Hamiltonians can be found in~\cite{Wulfman10a}*{\S~3.8}.

An important observation, which is often missed, is that the
three linear symplectic transformations are unitary rotations in the
corresponding hypercomplex algebra, cf.~\cite{Kisil09c}*{\S~3}. This
means, that the symplectomorphisms generated by operators \(Z\),
\(B-Z/2\), \(B\) within time \(t\) coincide with the
multiplication of hypercomplex number \(q+\alli p\) by \( \rme^{\alli
  t}\), see Subsection~\ref{sec:hyperc-char} and
Fig.~\ref{fig:rotations}, which is just another illustration of the
Similarity and Correspondence Principle~\ref{pr:simil-corr-principle}.
\begin{example}
  There are many advantages of considering representations of the
  Heisenberg group on the phase
  space~\citelist{\cite{Howe80b}*{\S~1.7}
    \cite{Folland89}*{\S~1.6} \cite{deGosson08a}}. A convenient
  expression for Fock--Segal--Bargmann%
  \index{Fock--Segal--Bargmann!representation}%
  \index{representation!Fock--Segal--Bargmann}%
  \index{FSB!representation|see{Fock--Segal--Bargmann representation}}%
  \index{representation!FSB|see{Fock--Segal--Bargmann representation}}
  (FSB) representation on the phase space%
  \index{phase!space}%
  \index{space!phase} is,
  cf.~\S~\ref{sec:schr-segal-bargm} and~\citelist{\cite{Kisil02e}*{(2.9)}
  \cite{deGosson08a}*{(1)}}:
  \begin{equation}
    \label{eq:stone-inf}
    \textstyle
    [\uir{}{F}(s,x,y) f] (q,p)=
     \rme^{-2\pi\rmi(\myhbar s+qx+py)}
    f \left(q-\frac{\myhbar}{2} y, p+\frac{\myhbar}{2} x\right).
  \end{equation}
  Then the derived representation of \(\algebra{h}\) is:
  \begin{equation}
    \label{eq:fock-rep-conf-der-par1}
    \textstyle
    \uir{}{F}(X)=-2\pi\rmi q+\frac{\myhbar}{2}\partial_{p},\qquad
    \uir{}{F}(Y)=-2\pi\rmi p-\frac{\myhbar}{2}\partial_{q},
    \qquad
    \uir{}{F}(S)=-2\pi\rmi\myhbar I.
  \end{equation}
  This produces the derived form of the Shale--Weil representation:
  \begin{equation}
    \label{eq:shale-weil-der-ell}
    \textstyle
    \uir{\text{SW}}{F}(A) =\frac{1}{2}\left(q\partial_{q}-p\partial_{p}\right),\quad
    \uir{\text{SW}}{F}(B)=-\frac{1}{2}\left(p\partial_{q}+q\partial_{p}\right),\quad
    \uir{\text{SW}}{F}(Z)=p\partial_{q}-q\partial_{p}.
  \end{equation}
  Note that this representation does not contain the parameter
  \(\myhbar\) unlike the equivalent
  representation~\eqref{eq:shale-weil-der}. Thus, the FSB model%
  \index{Fock--Segal--Bargmann!space}%
  \index{space!Fock--Segal--Bargmann} explicitly shows the equivalence
  of \(\uir{\text{SW}}{\myhbar_1}\) and \(\uir{\text{SW}}{\myhbar_2}\)
  if \(\myhbar_1 \myhbar_2>0\)~\cite{Folland89}*{Thm.~4.57}.
\end{example}

As we will also see below the FSB-type representations%
\index{Fock--Segal--Bargmann!representation}%
\index{representation!Fock--Segal--Bargmann} in hypercomplex numbers
produce almost the same Shale--Weil representations.

\section{Mechanics and Hypercomplex Numbers}
\label{sec:quantum-mechanics-1}

Complex valued representations of the Heisenberg group%
\index{Heisenberg!group}%
\index{group!Heisenberg} provide a natural framework for quantum
mechanics~\cites{Howe80b,Folland89}. These representations provide the
fundamental example of induced representations%
\index{representation!induced}%
\index{induced!representation}, the Kirillov orbit method%
\index{orbit!method}%
\index{method!orbits, of} and geometrical quantisation%
\index{geometrical quantisation}%
\index{quantisation!geometrical}
technique~\cites{Kirillov99,Kirillov94a}.  Following the presentation
in Section~\ref{sec:induc-repr} we will consider representations of
the Heisenberg group which are induced by hypercomplex characters of
its centre: complex (which correspond to the elliptic case), dual
(parabolic) and double (hyperbolic).

To describe dynamics of a physical system we use a universal equation
based on inner derivations%
\index{inner!derivation}%
\index{derivation!inner} (commutator%
\index{commutator}) of the convolution
algebra~\citelist{\cite{Kisil00a} \cite{Kisil02e}}.  The complex
valued representations produce the standard framework for quantum
mechanics with the Heisenberg dynamical equation%
\index{Heisenberg!equation}%
\index{equation!Heisenberg}~\cite{Vourdas06a}.

The double number%
\index{number!double}%
\index{double!number} valued representations, with
the hyperbolic unit \(\rmh^2=1\), is a natural source of hyperbolic
quantum mechanics developed for a
while~\cites{Hudson04a,Hudson66a,Khrennikov03a,Khrennikov05a,Khrennikov08a,Ulrych2014a}.
The universal dynamical equation employs hyperbolic commutator in this
case. This can be seen as a \emph{Moyal bracket}%
\index{Moyal bracket!hyperbolic}%
\index{bracket!Moyal!hyperbolic}%
\index{hyperbolic!Moyal bracket} based on the hyperbolic sine
function. The hyperbolic observables act as operators on a Krein
space%
\index{Krein!space}%
\index{space!Krein} with an indefinite inner product. Such spaces are
employed in study of \(\mathcal{PT}\)-symmetric%
\index{PT-symmetry@\(\mathcal{PT}\)-symmetry} Hamiltonians and
hyperbolic unit \(\rmh^2=1\) naturally appear in this
setup~\cite{GuentherKuzhel10a}.

The representations with values in dual numbers%
\index{number!dual}%
\index{dual!number}
provide a convenient description of the classical mechanics. For this
we do not take any sort of semiclassical limit%
  \index{semiclassical!limit}%
  \index{limit!semiclassical}, rather the nilpotency
of the parabolic unit (\(\rmp^2=0\)) does the task. This removes the
vicious necessity to consider the Planck \emph{constant}%
\index{Planck!constant}%
\index{constant!Planck} tending to
zero. 
The dynamical equation takes the Hamiltonian%
\index{Hamilton!equation}%
\index{equation!Hamilton} form. We also describe
classical non-commutative representations of the Heisenberg group
which acts in the first jet space\index{jet}.

\begin{remark}
  It is worth to note that our technique is different from contraction
  technique in the theory of Lie
  groups~\cites{LevyLeblond65a,GromovKuratov05b,Gromov12a}. Indeed a
  contraction of the Heisenberg group \(\Space{H}{n}\)
  is the commutative Euclidean group \(\Space{R}{2n}\)
  which may be identified with the phase space in classical and
  quantum mechanics.
\end{remark}

The considered here approach provides not only three different types
of dynamics, it also generates the respective rules for addition of
probabilities%
\index{probability!quantum}%
\index{quantum!probability} as well.  For example, the quantum
interference is the consequence of the same complex-valued structure,
which directs the Heisenberg equation. The absence of an interference
(a particle behaviour) in the classical mechanics is again the
consequence the nilpotency of the parabolic unit. Double numbers
creates the hyperbolic law of additions of probabilities, which was
extensively investigates~\cites{Khrennikov03a,Khrennikov05a}. There
are still unresolved issues with positivity of the probabilistic
interpretation in the hyperbolic case~\cites{Hudson04a,Hudson66a}.

The fundamental relations of quantum and classical mechanics were
discussed in Section~\ref{sec:overview}. Below we will recover the
existence of three non-isomorphic models of mechanics from the
representation theory. They were already derived
in~\cites{Hudson04a,Hudson66a} from translation invariant formulation,
that is from the group theory as well.  It also hinted that hyperbolic
counterpart is (at least theoretically) as natural as classical and
quantum mechanics are. The approach provides a framework for a
description of aggregate system which have say both quantum and
classical components. This can be used to model quantum computers with
classical terminals~\cite{Kisil09b}.

Remarkably, simultaneously with the work \cite{Hudson66a}
group-invariant axiomatics of geometry leaded
R.I.~Pimenov~\cite{Pimenov65a} to description of \(3^n\) Cayley--Klein
constructions. The connection between group-invariant geometry and respective
mechanics were explored in many works of N.A.~Gromov, see for
example~\cites{Gromov90a,Gromov90b,GromovKuratov05b}. They already
highlighted the r\^ole of three types of hypercomplex units for the
realisation of elliptic, parabolic and hyperbolic geometry and
kinematic.

There is a further connection between representations of the
Heisenberg group and hypercomplex numbers. The symplectomorphism of
phase space%
\index{phase!space}%
\index{space!phase} are also automorphism of the Heisenberg
group~\cite{Folland89}*{\S~1.2}. We recall that the symplectic group
\(\SL\)%
\index{group!$\SL$}%
\index{group!$\SL$|see{also $\SL$}}~\cite{Folland89}*{\S~1.2} is
isomorphic to the group \(\SL\)~\citelist{ \cite{Lang85}
  \cite{HoweTan92} \cite{Mazorchuk09a}} and provides linear
symplectomorphisms%
\index{symplectic!transformation}%
\index{transformation!symplectic} of the two-dimensional phase space.
It has three types of non-iso\-mor\-phic one-dimensional continuous
subgroups~\eqref{eq:a-subgroup}--\eqref{eq:k-subgroup} with symplectic
action on the phase space%
\index{phase!space}%
\index{space!phase} illustrated by Fig.~\ref{fig:rotations}.
Hamiltonians, which produce those symplectomorphism, are of
interest~\citelist{\cite{Wulfman10a}*{\S~3.8} \cite{ATorre08a}
  \cite{ATorre10a}}. An analysis of those Hamiltonians from
Subsection~\ref{sec:correspondence} by means of ladder operators
recreates hypercomplex coefficients as well~\cites{Kisil11a}.

Harmonic oscillators%
\index{harmonic!oscillator}%
\index{oscillator!harmonic}, which we shall use as the main
illustration here, are treated in most textbooks on quantum mechanics%
\index{quantum mechanics}%
\index{quantum mechanics}. This is efficiently done through
creation/annihilation (ladder) operators,
cf.~\S~\ref{sec:correspondence} and~\citelist{\cite{Gazeau09a}
  \cite{BoyerMiller74a}}. The underlying structure is the
representation theory of the Heisenberg%
\index{Heisenberg!group|(}%
\index{group!Heisenberg|(} and symplectic groups%
\index{symplectic!group}%
\index{group!symplectic}~\citelist{\cite{Lang85}*{\S~VI.2}
  \cite{MTaylor86}*{\S~8.2} \cite{Howe80b} \cite{Folland89}}.  As
we will see, they are naturally connected with respective hypercomplex
numbers.  As a result we obtain further illustrations to the
Similarity and Correspondence Principle~\ref{pr:simil-corr-principle}.

We work with the simplest case of a particle with only one
degree of freedom. Higher dimensions and the respective group of
symplectomorphisms \(\Sp[2n]\) may require consideration of Clifford
algebras%
\index{Clifford!algebra}%
\index{algebra!Clifford}~\citelist{\cite{Kisil93c}
  \cite{ConstalesFaustinoKrausshar11a} \cite{CnopsKisil97a}
  \cite{GuentherKuzhel10a} \cite{Porteous95} \cite{Kisil14a} \cite{Ulrych2014a}}.

\subsection{p-Mechanic Formalism}
Here we briefly outline a formalism~\cites{Kisil96a,Prezhdo-Kisil97,Kisil00a,BrodlieKisil03a,Kisil02e}, which allows to unify quantum and
classical mechanics.\index{p-mechanics}

\subsubsection{Convolutions (Observables) on $\Space{H}{}$ and Commutator}
\label{sec:conv-algebra-hg}

Using the invariance of the Lebesgue measure \(dg=\rmd s\,\rmd x\,\rmd y\)
on \(\Space{H}{}\)
we can define the convolution\index{convolution} of two functions:
\begin{eqnarray}
  (k_1 * k_2) (g) &=& \int_{\Space{H}{}} k_1(g_1)\,
  k_2(g_1^{-1}g)\,\rmd g_1  .
  \label{eq:de-convolution}
\end{eqnarray}
Because \(\Space{H}{}\)
is non-commutative, the convolution is a non-commutative operation.
It is meaningful for functions from various spaces including
\(\FSpace{L}{1}(\Space{H}{})=\FSpace{L}{1}(\Space{H}{},dg)\),
the Schwartz space \(\FSpace{S}{}\)
and many classes of distributions, which form algebras under
convolutions.  Convolutions on \(\Space{H}{}\)
are used as \emph{observables}%
\index{observable}%
\index{p-mechanics!observable} in
\(p\)-mechanic~\cites{Kisil96a,Kisil02e}.

A unitary representation \(\uir{}{} \) of \(\Space{H}{}\) extends
 to \(\FSpace{L}{1}(\Space{H}{})\) by the formula:
\begin{equation}
  \label{eq:rho-extended-to-L1}
  \uir{}{} (k) = \int_{\Space{H}{}} k(g)\uir{}{}  (g)\,\rmd g \notingiq
\end{equation}
where the operator-valued integral can be defined in a weak
sense. This is also an algebra homomorphism of convolutions to linear
operators.

For a dynamics of observables we need \emph{inner derivations}%
\index{inner!derivation}%
\index{derivation!inner} \(D_k\) of
the convolution algebra \(\FSpace{L}{1}(\Space{H}{})\), which are
given by the \emph{commutator}\index{commutator}:
\begin{eqnarray}
  \qquad D_k: f \mapsto [k,f]&=&k*f-f*k   \label{eq:commutator}
   \\  &=&
  \int_{\Space{H}{}} k(g_1)\left( f(g_1^{-1}g)-f(gg_1^{-1})\right)\,\rmd g_1
, \quad f,k\in\FSpace{L}{1}(\Space{H}{}).
\nonumber 
\end{eqnarray}

To describe dynamics of a time-dependent observable \(f(t,g)\) we use
the universal equation%
\index{p-mechanics!dynamic equation}%
\index{equation!dynamics in p-mechanics}, cf.~\cites{Kisil94d,Kisil96a}:
\begin{equation}
  \label{eq:universal}
  S\dot{f}=[H,f]\notingiq
\end{equation}
where \(S\) is the left-invariant vector
field~\eqref{eq:h-lie-algebra} generated by the centre of
\(\Space{H}{}\). The presence of operator \(S\) fixes the
dimensionality of both sides of the equation~\eqref{eq:universal} if
the observable \(H\) (Hamiltonian) has the dimensionality of
energy~\cite{Kisil02e}*{Rem~4.1}. 

Alternatively, if we apply a right inverse
\(\anti\) of \(S\) to both sides of the equation~\eqref{eq:universal}
we obtain the equivalent equation
\begin{equation}
  \label{eq:universal-bracket}
  \dot{f}=\ub{H}{f}\notingiq
\end{equation}
based on the universal bracket
\(\ub{k_1}{k_2}=k_1*\anti k_2-k_2*\anti k_1\)~\cite{Kisil02e}.
We will not use this approach in the present paper.
\begin{example}[Harmonic oscillator]
  \label{ex:p-harmonic}
  Let \(H=\frac{1}{2} (mk^2 q^2 + \frac{1}{m}p^2)\) be the
  Hamiltonian of a one-dimensional harmonic oscillator%
  \index{harmonic!oscillator|indef}%
  \index{oscillator!harmonic|indef}, where \(k\) is a constant
  frequency and \(m\) is a constant mass.  Its \emph{p-mechanisation}%
  \index{p-mechanisation} will be the second order differential
  operator on \(\Space{H}{}\)~\cite{BrodlieKisil03a}*{\S~5.1}:
  \begin{displaymath}
    \textstyle
    H=\frac{1}{2} (mk^2 X^2
    + \frac{1}{m}Y^2)\notingiq
  \end{displaymath}
  where we dropped sub-indexes of vector
  fields~\eqref{eq:h-lie-algebra} in one dimensional setting. We can
  express the commutator as a difference between the left and the
  right action of the vector fields:
  \begin{displaymath}
    \textstyle
    [H,f]=\frac{1}{2} (mk^2 ((X^{r})^2-(X^{l})^2)
    + \frac{1}{m}((Y^{r})^2-(Y^{l})^2))f.
  \end{displaymath}
   Thus the equation~\eqref{eq:universal} becomes~\cite{BrodlieKisil03a}*{(5.2)}:
   \begin{equation}
     \label{eq:p-harm-osc-dyn}
     \frac{\partial }{\partial s}\dot{f}= \frac{\partial }{\partial s}
    \left(m k^2 y\frac{\partial}{\partial x}-\frac{1}{m} x
      \frac{\partial}{\partial y} \right) f. 
   \end{equation}
   Of course, the derivative \(\frac{\partial }{\partial s}\) can be
   dropped from both sides of the equation and the general solution
   is found to be:
   \begin{equation}
     \label{eq:p-harm-sol}
     \textstyle
     f(t;s,x,y)  =  f_0\left(s, x\cos(k t) +
      m k y\sin( k t),  -\frac{x}{mk} \sin(k t) + y\cos (k t)\right)\notingiq
   \end{equation}
   where \(f_0(s,x,y)\) is the initial value of an observable on \(\Space{H}{}\).
\end{example}
\begin{example}[Unharmonic oscillator] 
  \label{ex:p-unharmonic}
  We consider unharmonic
  oscillator%
  \index{unharmonic!oscillator}%
  \index{oscillator!unharmonic} with cubic potential, see~\cite{CalzettaVerdaguer06a} and
  references therein:
  \begin{equation}
    \label{eq:unharmonic-hamiltonian}
    H=\frac{mk^2}{2} q^2+\frac{\lambda}{6} q^3
  + \frac{1}{2m}p^2.
  \end{equation}
  Due to the absence of non-commutative products
  in~\eqref{eq:unharmonic-hamiltonian}, its p-mechanisation is again straightforward: 
  \begin{displaymath}
    H=\frac{mk^2}{2}  X^2+\frac{\lambda}{6} X^3
    + \frac{1}{m}Y^2.
  \end{displaymath}
  Similarly to the harmonic case the dynamic equation, after
  cancellation of \(\frac{\partial }{\partial s}\) on both sides,
  becomes:
   \begin{equation}
     \label{eq:p-unharm-osc-dyn}
     \dot{f}=     \left(m k^2 y\frac{\partial}{\partial x}
       +\frac{\lambda}{6}\left(3y\frac{\partial^2}{\partial x^2} 
         +\frac{1}{4}y^3\frac{\partial^2}{\partial s^2}\right)-\frac{1}{m} x
      \frac{\partial}{\partial y} \right) f. 
   \end{equation}
Unfortunately, it cannot be solved analytically as easy as in the
harmonic case.
\end{example}

\subsubsection{States and Probability}
\label{sec:states-probability}

Let an observable \(\uir{}{}(k)\)~\eqref{eq:rho-extended-to-L1} is
defined by a kernel \(k(g)\) on the Heisenberg group and a
representation \(\uir{}{}\) at a Hilbert space \(\mathcal{H}\). A
\emph{state}%
\index{state}%
\index{p-mechanics!state} on the convolution algebra is given by a vector
\(v\in\mathcal{H}\). A simple calculation:
\begin{eqnarray*}
  \scalar[\mathcal{H}]{\uir{}{}(k)v}{v}&=& \scalar[\mathcal{H}]{\int_{\Space{H}{}} k(g)
    \uir{}{}(g)v\,\rmd g}{v}\\
  &=& \int_{\Space{H}{}} k(g) \scalar[\mathcal{H}]{\uir{}{}(g)v}{v}dg\\
  &=& \int_{\Space{H}{}} k(g) \overline{\scalar[\mathcal{H}]{v}{\uir{}{}(g)v}}\,\rmd g
\end{eqnarray*}
can be restated as:
\begin{displaymath}
  \scalar[\mathcal{H}]{\uir{}{}(k)v}{v}=\scalar[]{k}{l}, \qquad \text{where} \quad
  l(g)=\scalar[\mathcal{H}]{v}{\uir{}{}(g)v}.
\end{displaymath}
Here the left-hand side contains the inner product on
\(\mathcal{H}\), while the right-hand side uses a skew-linear pairing
between functions on \(\Space{H}{}\) based on the Haar measure
integration. In other words we obtain,
cf.~\cite{BrodlieKisil03a}*{Thm.~3.11}:
\begin{proposition}
  \label{pr:state-functional}
  A state defined by a vector \(v\in\mathcal{H}\) coincides with the
  linear functional given by the wavelet transform
  \begin{equation}
    \label{eq:kernel-state}
    l(g)=\scalar[\mathcal{H}]{v}{\uir{}{}(g)v}
  \end{equation}
  of \(v\) used as the mother wavelet as well.
\end{proposition}
The addition of vectors in \(\mathcal{H}\) implies the following
operation on states:
\begin{eqnarray}
  \scalar[\mathcal{H}]{v_1+v_2}{\uir{}{}(g)(v_1+v_2)}&=&
    \scalar[\mathcal{H}]{v_1}{\uir{}{}(g)v_1}
    +\scalar[\mathcal{H}]{v_2}{\uir{}{}(g)v_2}\nonumber \\
    &&{}+
    \scalar[\mathcal{H}]{v_1}{\uir{}{}(g)v_2}
   + \overline{\scalar[\mathcal{H}]{v_1}{\uir{}{}(g^{-1})v_2}}
   \label{eq:kernel-add}
\end{eqnarray}
The last expression can be conveniently rewritten for kernels of the
functional as 
\begin{equation}
  \label{eq:addition-functional}
  l_{12}=l_1+l_2+2 A\sqrt{l_1l_2}
\end{equation}
for some real number \(A\). This formula is behind the contextual law
of addition of conditional probabilities%
\index{probability!quantum}%
\index{quantum!probability}~\cite{Khrennikov01a} and will
be illustrated below. Its physical interpretation is an interference,%
\index{interference} say, from two slits. Despite of a common belief,
the mechanism of such interference can be both causal\index{causal}
and local, see~\citelist{\cite{Kisil01c} \cite{KhrenVol01}}.

\subsection{Elliptic Characters and Quantum Dynamics}
\label{sec:ellipt-char-moyal}
In this subsection we consider the representation \(\uir{}{\myh}\) of
the Heisenberg group \(\Space{H}{}\) induced by the elliptic character
\(\chi_\myh(s)= \rme^{\rmi\myh s}\) in complex numbers parametrised by
\(\myh\in\Space{R}{}\). We also use the convenient agreement
\(\myh=2\pi\myhbar\) borrowed from physical literature.

\subsubsection{Fock--Segal--Bargmann and Schr\"odinger Representations}
\label{sec:schr-segal-bargm}
The realisation of \(\uir{}{\myh}\) by the left
shifts~\eqref{eq:left-right-regular} on
\(\FSpace[\myh]{L}{2}(\Space{H}{})\) is rarely used in quantum
mechanics. Instead two unitary equivalent forms are more common: the
Schr\"odinger%
\index{Schr\"odinger!representation}%
\index{representation!Schr\"odinger} and Fock--Segal--Bargmann (FSB) representations.%
\index{Fock--Segal--Bargmann!representations|indef}%
\index{representations!Fock--Segal--Bargmann|indef}

The FSB representation can be obtained from the orbit
method of Kirillov~\cite{Kirillov94a}. It allows spatially separate
irreducible components of the left regular representation, each of
them become located on the orbit of the co-adjoint representation,
see~\citelist{\cite{Kisil02e}*{\S~2.1} \cite{Kirillov94a}} for
details, we only present a brief summary here.

We identify \(\Space{H}{}\) and its Lie algebra \(\algebra{h}\)
through the exponential map~\cite{Kirillov76}*{\S~6.4}. The dual
\(\algebra{h}^*\) of \(\algebra{h}\) is presented by the Euclidean
space \(\Space{R}{3}\) with bi-orthogonal coordinates \((\myhbar,q,p)\).  Then the
pairing of \(\algebra{h}^*\) and \(\algebra{h}\) given by
\begin{displaymath}
  \scalar{(s,x,y)}{(\myhbar,q,p)}=\myhbar s + q \cdot x+p\cdot y.
\end{displaymath}
This pairing can be used to defines the Fourier
transform \(\hat{\ }: \FSpace{L}{2}(\Space{H}{})\rightarrow
\FSpace{L}{2}(\algebra{h}^*)\) given by~\cite{Kirillov99}*{\S~2.3}:
\begin{equation}
  \label{eq:fourier-transform}
  \hat{\phi}(F)=\int_{\algebra{h}^n} \phi(\exp X) 
   \rme^{-2\pi\rmi  \scalar{X}{F}}\,\rmd X, \qquad \textrm{ where}\quad]
  X\in\algebra{h}^n,\ F\in\algebra{h}^*. 
\end{equation}
For a fixed \(\myhbar\) the left regular
representation~\eqref{eq:left-right-regular} is mapped by the Fourier
transform to the FSB type representation~\eqref{eq:stone-inf}.  The
collection of points \((\myhbar,q,p)\in\algebra{h}^*\) for a fixed
\(\myhbar\) is naturally identified with the \emph{phase space}%
\index{phase!space}%
\index{space!phase} of the system.
\begin{remark}
  It is possible to identify the case of \(\myhbar=0\) with classical
  mechanics~\cite{Kisil02e}. Indeed, a substitution of the zero value of \(\myhbar\)
  into~\eqref{eq:stone-inf} produces the commutative representation:
  \begin{equation}
    \label{eq:commut-repres}
      \uir{}{0}(s,x,y): f (q,p) \mapsto 
       \rme^{-2\pi\rmi(qx+py)}
      f \left(q, p\right).
  \end{equation}
  It can be decomposed into the direct integral of one-dimensional
  representations parametrised by the points \((q,p)\) of the phase
  space. The classical mechanics%
  \index{classical mechanics}%
  \index{mechanics!classical}, including the Hamilton equation, can
  be recovered from those representations~\cite{Kisil02e}. However, 
  the condition \(\myhbar=0\) (as well as the \emph{semiclassical limit}%
  \index{semiclassical!limit}%
  \index{limit!semiclassical} \(\myhbar\rightarrow 0\)) is
  not completely physical. Commutativity (and subsequent relative
  triviality) of those representation is the main reason why they are
  oftenly neglected. The commutativity can be outweighed by special
  arrangements, e.g. an antiderivative, see the discussion
  around~\eqref{eq:universal-bracket} and~\cite{Kisil02e}*{(4.1)}. However, the
  procedure is not straightforward, see discussion in~\citelist{\cite{Kisil05c}
    \cite{AgostiniCapraraCiccotti07a}
    \cite{Kisil09a}}. A direct approach using dual
  numbers will  be shown below, cf. Rem.~\ref{re:hamilton-from-nc}.
\end{remark}

To recover the Schr\"odinger representation we use notations and
technique of induced representations from
\S~\ref{sec:concl-induc-repr}, see also~\cite{Kisil98a}*{Ex.~4.1}.
The subgroup \(H=\{(s,0,y)\such s\in\Space{R}{},
y\in\Space{R}{n}\}\subset\Space{H}{}\) defines the homogeneous space
\(X=G/H\), which coincides with \(\Space{R}{n}\) as a manifold. The
natural projection \(\map{p}:G\rightarrow X\) is \(\map{p}(s,x,y)=x\) and its left
inverse \(\map{s}:X\rightarrow G\) can be as simple as \(\map{s}(x)=(0,x,0)\).
For the map \(\map{r}:G\rightarrow H\), \(\map{r}(s,x,y)=(s-xy/2,0,y)\) we have
the decomposition
\begin{displaymath}
  (s,x,y)=\map{s}(\map{p}(s,x,y))*\map{r}(s,x,y)=(0,x,0)*(s-\textstyle\frac{1}{2}xy,0,y).
\end{displaymath}
For a character
\(\chi_{\myh}(s,0,y)= \rme^{\rmi\myh s}\) of \(H\) the lifting
\(\oper{L}_\chi: \FSpace{L}{2}(G/H) \rightarrow
\FSpace[\chi]{L}{2}(G)\) is as follows:
\begin{displaymath}
  [\oper{L}_\chi f](s,x,y)=\chi_{\myh}(\map{r}(s,x,y))\, 
  f(\map{p}(s,x,y))= \rme^{\rmi\myh (s-xy/2)}f(x).  
\end{displaymath}
Thus the representation \(\uir{}{\chi}(g)=\oper{P}\circ\Lambda
(g)\circ\oper{L}\) becomes:
\begin{equation}
  \label{eq:schroedinger-rep}
  [\uir{}{\chi}(s',x',y') f](x)= \rme^{-2\pi\rmi\myhbar (s'+xy'-x'y'/2)}\,f(x-x').  
\end{equation}
After the Fourier transform \(x\mapsto q\)
(similar to one in~\eqref{eq:fourier-transform}) we get the
Schr\"odinger representation%
\index{Schr\"odinger!representation|indef}%
\index{representation!Schr\"odinger|indef} on the \emph{configuration
  space}%
\index{configuration!space}%
\index{space!configuration}:
\begin{equation}
  \label{eq:schroedinger-rep-conf}
  [\uir{}{\chi}(s',x',y') \hat{f}\,](q)= \rme^{-2\pi\rmi\myhbar (s'+x'y'/2)
    -2\pi\rmi x' q}\,\hat{f}(q+\myhbar y').  
\end{equation}
Note that this again turns into a commutative representation
(multiplication by an unimodular function) if \(\myhbar=0\). To get
the full set of commutative representations in this way we need to use the
character \(\chi_{(\myh,p)}(s,0,y)= \rme^{2\pi\rmi(\myhbar+ py)}\) in the
above consideration. 

\subsubsection{Commutator and the Heisenberg Equation}
\label{sec:comm-heis-equat}

The property~\eqref{eq:induced-prop} of
\(\FSpace[\chi]{F}{2}(\Space{H}{})\) implies that the restrictions of
two operators \(\uir{}{\chi} (k_1)\) and \(\uir{}{\chi} (k_2)\) to
this space are equal if
\begin{displaymath}
  \int_{\Space{R}{}} k_1(s,x,y)\,\chi(s)\, ds = \int_{\Space{R}{}} k_2(s,x,y)\,\chi(s)\,\rmd s.
\end{displaymath}
In other words, for a character \(\chi(s)= \rme^{2\pi\rmi \myhbar s}\) the
operator \(\uir{}{\chi} (k)\) depends only on
\begin{displaymath}
  \hat{k}_s(\myhbar,x,y)=\int_{\Space{R}{}} k(s,x,y)\, \rme^{-2\pi\rmi \myhbar s}\,\rmd s\notingiq
\end{displaymath}
which is the partial Fourier transform \(s\mapsto \myhbar\) of
\(k(s,x,y)\). The restriction to \(\FSpace[\chi]{F}{2}(\Space{H}{})\)
of the composition formula for convolutions is~\cite{Kisil02e}*{(3.5)}:
\begin{equation}
  \label{eq:composition-ell}
  (k'*k)\hat{_s}
  =
  \int_{\Space{R}{2n}}  \rme^{ {\rmi \myh}{}(xy'-yx')/2}\, \hat{k}'_s(\myhbar ,x',y')\,
 \hat{k}_s(\myhbar ,x-x',y-y')\,\rmd x'\rmd y'. 
\end{equation}
Under the Schr\"odinger representation~\eqref{eq:schroedinger-rep-conf} the
convolution~\eqref{eq:composition-ell} defines a rule for composition
of two pseudo-differential operators (PDO) in the Weyl
calculus~\citelist{\cite{Howe80b} \cite{Folland89}*{\S~2.3}}.

Consequently the representation~\eqref{eq:rho-extended-to-L1} of
commutator~\eqref{eq:commutator} depends only on its partial Fourier
transform~\cite{Kisil02e}*{(3.6)}: 
\begin{eqnarray}
  [k',k]\hat{_s}
  &=&   2 \rmi  \int_{\Space{R}{2n}}\!\! \sin(\textstyle\frac{\myh}{2}
  (xy'-yx'))\,\label{eq:repres-commutator}\\
   && \qquad\times 
  \hat{k}'_s(\myhbar, x', y')\,
  \hat{k}_s(\myhbar, x-x', y-y')\,\rmd x'\rmd y'. \nonumber 
\end{eqnarray}
Under the Fourier transform~\eqref{eq:fourier-transform} this
commutator is exactly the \emph{Moyal bracket}%
\index{Moyal bracket}%
\index{bracket!Moyal}~\cite{Zachos02a} for of
\(\hat{k}'\) and \(\hat{k}\) on the phase space.

For observables in the space
\(\FSpace[\chi]{F}{2}(\Space{H}{})\) the action of \(S\) is reduced to
multiplication
, e.g. for \(\chi(s)= \rme^{\rmi \myh s}\) the action of \(S\) is
multiplication by \(\rmi \myh\). Thus the
equation~\eqref{eq:universal} reduced to the space
\(\FSpace[\chi]{F}{2}(\Space{H}{})\) becomes the Heisenberg type
equation~\cite{Kisil02e}*{(4.4)}:
\begin{equation}
  \label{eq:heisenberg-eq}
  \dot{f}=\frac{1}{\rmi\myh}  [H,f]\hat{_s}\notingiq
\end{equation}
based on the above bracket~\eqref{eq:repres-commutator}.  The
Schr\"odinger representation~\eqref{eq:schroedinger-rep-conf} transforms
this equation to the original Heisenberg equation%
\index{Heisenberg!equation}%
\index{equation!Heisenberg}.

\begin{example}
  \label{ex:quntum-oscillators}
  \begin{enumerate}
  \item 
    \label{it:q-harmonic}
    Under the Fourier transform \((x,y)\mapsto(q,p)\) the
    p-dynamic equation~\eqref{eq:p-harm-osc-dyn} of the harmonic
    oscillator%
    \index{harmonic!oscillator}%
    \index{oscillator!harmonic} becomes: 
    \begin{equation}
      \label{eq:harmic-osc-dyn}
      \dot{f}=     \left(m k^2 q\frac{\partial}{\partial p}-\frac{1}{m} p
      \frac{\partial}{\partial q} \right) f. 
    \end{equation}
    The same transform creates its solution out
    of~\eqref{eq:p-harm-sol}.
  \item 
    \label{it:q-unharmonic}
    Since \(\frac{\partial}{\partial s}\) acts on
    \(\FSpace[\chi]{F}{2}(\Space{H}{})\) as multiplication by \(\rmi
    \myhbar\), the quantum representation of unharmonic dynamics%
    \index{unharmonic!oscillator}%
    \index{oscillator!unharmonic}
    equation~\eqref{eq:p-unharm-osc-dyn} is:
    \begin{equation}
      \label{eq:q-unhar-dyn}
      \dot{f}=     \left(m k^2 q\frac{\partial}{\partial
          p}+\frac{\lambda}{6}\left(3q^2\frac{\partial}{\partial p}
          -\frac{\myhbar^2}{4}\frac{\partial^3}{\partial p^3}\right)-\frac{1}{m}
        p \frac{\partial}{\partial q} \right) f. 
    \end{equation}
    This is exactly the equation for the Wigner function obtained
    in~\cite{CalzettaVerdaguer06a}*{(30)}. 
  \end{enumerate}
\end{example}

\subsubsection{Quantum Probabilities}
\label{sec:quantum-probabilities}
For the elliptic character \(\chi_\myh(s)= \rme^{\rmi\myh s }\) we can use
the Cauchy--Schwartz inequality to demonstrate that the real number
\(A\) in the identity~\eqref{eq:addition-functional} is between \(-1\)
and \(1\). Thus, we can put \(A=\cos \alpha\) for some angle (phase)
\(\alpha\) to get the formula for counting quantum
probabilities%
\index{probability!quantum}%
\index{quantum!probability}, cf.~\cite{Khrennikov03a}*{(2)}:
\begin{equation}
  \label{eq:addition-functional-ell}
  l_{12}=l_1+l_2+2 \cos\alpha \,\sqrt{l_1l_2}
\end{equation}

\begin{remark}
  \label{re:sine-cosine}
  It is interesting to note that the both trigonometric functions are
  employed in quantum mechanics: sine is in the heart of the Moyal
  bracket%
  \index{Moyal bracket}%
  \index{bracket!Moyal}~\eqref{eq:repres-commutator} and cosine is responsible for
  the addition of probabilities~\eqref{eq:addition-functional-ell}. In
  the essence the commutator and probabilities took respectively the
  odd and even parts of the elliptic character \( \rme^{\rmi\myh s}\).
\end{remark}

\begin{example}
Take a  vector \(v_{(a,b)}\in\FSpace[\myh]{L}{2}(\Space{H}{})\) defined by a
Gaussian\index{Gaussian} with mean value \((a,b)\) in the phase space%
\index{phase!space}%
\index{space!phase} for a harmonic oscillator of the
mass \(m\) and the frequency \(k\):
\begin{equation}
  \label{eq:gauss-state}
  v_{(a,b)}(q,p)=\exp\left(-\frac{2\pi k m}{\myhbar}(q-a)^2-\frac{2\pi}{\myhbar
      k m}(p-b)^2\right).
\end{equation}
A direct calculation shows:
\begin{eqnarray*}
  \lefteqn{\scalar{v_{(a,b)}}{\uir{}{\myhbar}(s,x,y)v_{(a',b')}}=\frac{4}{\myhbar}
  \exp\left(
    \pi \rmi \left(2s\myhbar+x (a+a')+y (b+b')\right)\frac{}{}\right.}\\
  &&\left.{} -\frac{\pi}{2 \myhbar k m }((\myhbar x+b-b')^2
    +(b-b')^2)
-\frac{\pi k m}{2\myhbar} ((\myhbar y+a'-a)^2
  + (a'-a)^2)
  \right)\\
  &=&\frac{4}{\myhbar}
  \exp\left(
    \pi \rmi \left(2s\myhbar+x (a+a')+y (b+b')\right)\frac{}{}\right.\\
  &&\left.{}  -\frac{\pi}{\myhbar k m }((b-b'+{\textstyle\frac{\myhbar x}{2}})^2
    +({\textstyle\frac{\myhbar x}{2}})^2)
    -\frac{\pi k m}{\myhbar} ((a-a'-{\textstyle\frac{\myhbar y}{2}})^2
    + ({\textstyle\frac{\myhbar y}{2}})^2) 
  \right).
\end{eqnarray*}
Thus, the kernel
\(l_{(a,b)}=\scalar{v_{(a,b)}}{\uir{}{\myhbar}(s,x,y)v_{(a,b)}}\)~\eqref{eq:kernel-state}
for a state \(v_{(a,b)}\) is:
\begin{eqnarray}
  l_{(a,b)}&=&\frac{4}{\myhbar}
  \exp\left(
    2\pi \rmi (s\myhbar+xa+yb)\frac{}{}
    -\frac{\pi\myhbar}{2 k m }x^2
    -\frac{\pi k m \myhbar}{2\myhbar} y^2
  \right)
  \label{eq:single-slit}
\end{eqnarray}
An observable registering a particle at a point \(q=c\) of the
configuration space%
\index{configuration!space}%
\index{space!configuration} is \(\delta(q-c)\). On the Heisenberg group this
observable is given by the kernel:
\begin{equation}
  \label{eq:coordinate}
  X_c(s,x,y)= \rme^{2\pi\rmi (s\myhbar+x  c)}\delta(y).
\end{equation}
The measurement of \(X_c\) on the
state~\eqref{eq:gauss-state} (through the
kernel~\eqref{eq:single-slit}) predictably is: 
\begin{displaymath}
  \scalar{X_c}{l_{(a,b)}}=\sqrt{\frac{2k
  m}{\myhbar}}\exp\left(-\frac{2\pi k m}{\myhbar}(c-a)^2\right).
\end{displaymath}
\end{example}
\begin{example}
  Now take two states \(v_{(0,b)}\) and \(v_{(0,-b)}\), where for the
  simplicity we assume the mean values of coordinates vanish in the
  both cases.  Then the corresponding kernel~\eqref{eq:kernel-add} has
  the interference\index{interference} terms:
\begin{eqnarray*}
  l_i&=&  \scalar{v_{(0,b)}}{\uir{}{\myhbar}(s,x,y)v_{(0,-b)}}\\
  &=&\frac{4}{\myhbar}
  \exp\left(2\pi \rmi s\myhbar
    -\frac{\pi}{2 \myhbar k m }((\myhbar x+2b)^2
    +4b^2)
    -\frac{\pi\myhbar k m}{2} y^2
  \right).
\end{eqnarray*}
The measurement of \(X_c\)~\eqref{eq:coordinate} on this term contains
the oscillating part:
\begin{displaymath}
  \scalar{X_c}{l_i}=\sqrt{\frac{2k m}{\myhbar}} \exp\left(-\frac{2\pi k
    m }{\myhbar} c^2
  -\frac{2\pi}{ k
    m \myhbar}b^2+\frac{4\pi\rmi}{\myhbar} cb\right)
\end{displaymath}
Therefore on the kernel \(l\) corresponding to
the state \(v_{(0,b)}+v_{(0,-b)}\) the measurement is 
\begin{eqnarray*}
  \scalar{X_c}{l}
  &=&2\sqrt{\frac{2 k
      m}{\myhbar}}\exp\left(-\frac{2\pi k m}{\myhbar}c^2\right)
  \left(1+\exp\left(    -\frac{2\pi}{ k
      m \myhbar}b^2\right)\cos\left(\frac{4\pi}{\myhbar} cb\right)\right).
\end{eqnarray*}
The presence of the cosine term in the last expression can generate an
interference picture. In practise, it does not happen for the minimal
uncertainty state~\eqref{eq:gauss-state} which we are using here: it
rapidly vanishes outside of the neighbourhood of zero, where
oscillations of the cosine occurs, see Fig.~\ref{fig:quant-prob}(a).
\end{example}
\begin{figure}[htbp]
  \centering
  (a)\includegraphics[scale=.75]{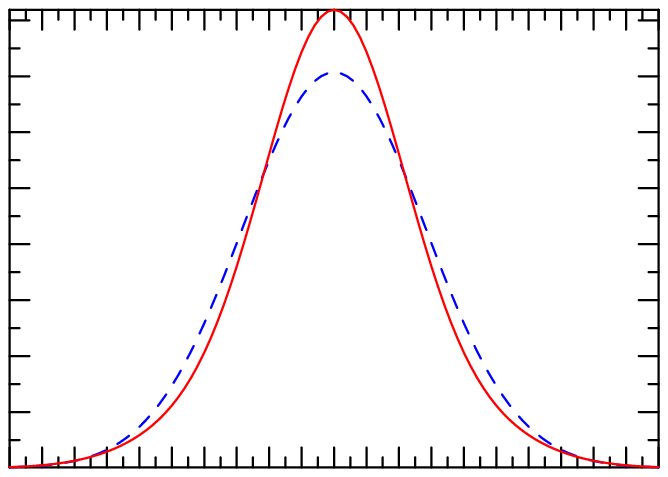}\hfill
  (b)\includegraphics[scale=.75]{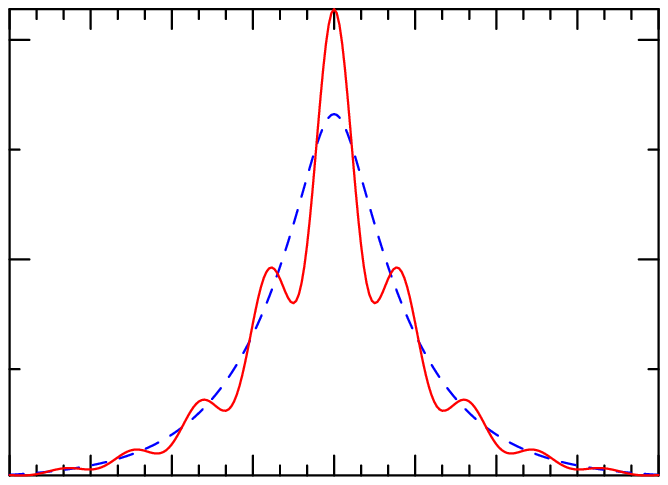}
  \caption{Quantum probabilities: the blue (dashed) graph shows the
    addition of probabilities without interaction, the red (solid)
    graph present the quantum interference. Left picture shows the
    Gaussian state~\eqref{eq:gauss-state}, the right\,--\,the rational
    state~\eqref{eq:poly-state}}
  \label{fig:quant-prob}
\end{figure}

\begin{example}
To see a traditional interference pattern one can use a state which is far
from the minimal uncertainty. For example, we can consider the state:
\begin{equation}
  \label{eq:poly-state}
  u_{(a,b)}(q,p)=\frac{\myhbar^2}{((q-a)^2+\myhbar/ k
    m)((p-b)^2+\myhbar k m)}.
\end{equation}
To evaluate the observable \(X_c\)~\eqref{eq:coordinate} on the
state \(l(g)=\scalar{u_1}{\uir{}{h}(g)u_2}\)~\eqref{eq:kernel-state}
we use the following formula: 
\begin{displaymath}
  \scalar{X_c}{l}=\frac{2}{\myhbar}\int_{\Space{R}{n}} \hat{u}_1(q,
  2(q-c)/\myhbar)\, 
  \overline{\hat{u}_2(q, 2(q-c)/\myhbar)}\,\rmd q\notingiq
\end{displaymath}
where \(\hat{u}_{i}(q,x)\) denotes  the partial Fourier transform
\(p\mapsto x\) of \(u_{i}(q,p)\). The formula is obtained by 
swapping order of integrations.  The numerical evaluation of the state
obtained by the addition \(u_{(0,b)}+u_{(0,-b)}\) is plotted on
Fig.~\ref{fig:quant-prob}(b), the red curve shows the canonical
interference pattern. 
\end{example}

\subsection{Ladder Operators and Harmonic Oscillator}
\label{sec:ladder-operators-oscillator}

Let \(\uir{}{}\) be a representation of the Schr\"odinger group%
\index{Schr\"odinger!group}%
\index{group!Schr\"odinger}
\(G=\Space{H}{}\rtimes\Mp\)~\eqref{eq:schrodinger-group}
in a space \(V\).
Consider the derived representation of the Lie algebra
\(\algebra{g}\)~\cite{Lang85}*{\S~VI.1}
and to simplify expressions we denote \(\tilde{X}=\uir{}{}(X)\)
for \(X\in\algebra{g}\).
To see the structure of the representation \(\uir{}{}\)
we can decompose the space \(V\)
into eigenspaces of the operator \(\tilde{X}\)
for some \(X\in \algebra{g}\).
The canonical example is the Taylor series in complex analysis.

We are going to consider three cases corresponding to three
non-isomorphic subgroups~(\ref{eq:a-subgroup}--\ref{eq:k-subgroup}) of
\(\SL\) starting from the compact case. Let \(X=Z\) be a
generator of the compact subgroup \(K\).  Corresponding
symplectomorphisms~\eqref{eq:sympl-auto} of the phase space%
\index{phase!space}%
\index{space!phase} are given
by orthogonal rotations with matrices \(\begin{pmatrix} \cos t & \sin
  t\\ -\sin t& \cos t
\end{pmatrix}\). The Shale--Weil
representation~\eqref{eq:shale-weil-der} coincides with the
Hamiltonian of the harmonic oscillator%
\index{harmonic!oscillator}%
\index{oscillator!harmonic} in Schr\"odinger representation.

Since \(\Mp\)
is a two-fold cover of \(\Mp\),
the corresponding eigenspaces of a compact group
\(\tilde{Z} v_k=\rmi k v_k\)
are parametrised by a half-integer \(k\in\Space{Z}{}/2\).
Explicitly for a half-integer \(k\) eigenvectors are:
\begin{equation}
  \label{eq:hermit-poly}
  v_k(q)=H_{k+\frac{1}{2}}\left(\sqrt{\frac{2\pi}{\myhbar}}q\right)  \rme^{-\frac{\pi}{\myhbar}q^2}\notingiq
\end{equation}
where \(H_k\) is the \emph{Hermite polynomial}%
\index{Hermite!polynomial}%
\index{polynomial!Hermite}%
\index{function!Hermite|see{Hermite polynomial}}%
~\citelist{\cite{Folland89}*{\S~1.7}  \cite{ErdelyiMagnusII}*{8.2(9)}}.  

From the point of view of quantum mechanics as well as the representation
theory, it is beneficial to introduce the
ladder operators \(\ladder{\pm}\)~\eqref{eq:raising-lowering}, known
also as \emph{creation/annihilation} in quantum
mechanics~\citelist{\cite{Folland89}*{p.~49} \cite{BoyerMiller74a}}.%
\index{ladder operator|(}%
\index{operator!ladder|(}
There are two ways to search for ladder operators: in
(complexified) Lie algebras \(\algebra{h}\) and \(\algebra{sl}_2\).
The later coincides with our consideration in
Section~\ref{sec:correspondence} in the essence.  

\subsubsection{Ladder Operators from the Heisenberg Group}
\label{sec:heis-group-oper}

Assuming \(\ladder{+}=a\tilde{X}+b\tilde{Y}\) we obtain from the
relations~\eqref{eq:cross-comm}--\eqref{eq:cross-comm1}
and~\eqref{eq:raising-lowering} the linear equations with unknown
\(a\) and \(b\):
\begin{displaymath}
  a=\lambda_+ b, \qquad -b=\lambda_+ a.
\end{displaymath}
The equations have a solution if and only if \(\lambda_+^2+1=0\), and
the raising/lowering operators are \(\ladder{\pm}=
\tilde{X}\mp\rmi\tilde{Y}\). 
\begin{remark}
  Here we have an interesting asymmetric response: due to the
  structure of the semidirect product
  \(\Space{H}{}\rtimes\Mp\) it is the
  symplectic group
  \index{symplectic!group}%
  \index{group!symplectic}%
  which acts on \(\Space{H}{}\), not vise versa.
  However the Heisenberg group has a weak action in the opposite
  direction: it shifts eigenfunctions of \(\SL\).
\end{remark}

In the Schr\"odinger representation%
\index{Schr\"odinger!representation}%
\index{representation!Schr\"odinger}~\eqref{eq:schroedinger-rep-conf-der}
the ladder operators are
\begin{equation}
  \label{eq:ell-ladder-heisen-rep}
  \uir{}{\myhbar}(\ladder{\pm})= 2\pi\rmi q\pm\rmi\myhbar \frac{\rmd}{\rmd q}.
\end{equation}
The standard treatment of the harmonic oscillator in quantum mechanics,
which can be found in many textbooks, e.g.~\citelist{
  \cite{Folland89}*{\S~1.7} \cite{Gazeau09a}*{\S~2.2.3}}, 
is as follows. The vector  \(v_{-1/2}(q)= \rme^{-\pi q^2/\myhbar}\) is an
eigenvector of \(\tilde{Z}\) with the eigenvalue
\(-\frac{\rmi}{2}\). In addition \(v_{-1/2}\) is annihilated by
\(\ladder{+}\). Thus the chain~\eqref{eq:ladder-chain-1D} terminates to
the right and the complete set of eigenvectors of the harmonic
oscillator Hamiltonian is presented by \((\ladder{-})^k v_{-1/2}\)
with \(k=0, 1, 2, \ldots\).

We can make a wavelet transform%
\index{wavelet!transform} generated by the Heisenberg group with
the mother wavelet \(v_{-1/2}\), and the image will be the
Fock--Segal--Bargmann (FSB) space%
\index{Fock--Segal--Bargmann!space}%
\index{space!Fock--Segal--Bargmann}%
\index{FSB!space|see{Fock--Segal--Bargmann space}}%
\index{space!FSB|see{Fock--Segal--Bargmann space}} \citelist{\cite{Howe80b}
  \cite{Folland89}*{\S~1.6}}. Since \(v_{-1/2}\) is the null solution
of \(\ladder{+}=\tilde{X}-\rmi \tilde{Y}\), then by
Cor.~\ref{co:cauchy-riemann-integ} the image of
the wavelet transform will be null-solutions of the corresponding linear
combination of the Lie derivatives~\eqref{eq:h-lie-algebra}:
\begin{equation}
  \label{eq:CR-Bargmann}
  D=\overline{X^{r} -\rmi  Y^{r}}=(\partial_{ x} +\rmi\partial_{y})-\pi\myhbar(x-\rmi
y)\notingiq
\end{equation}
which turns out to be the Cauchy--Riemann equation%
\index{Cauchy-Riemann operator}%
\index{operator!Cauchy-Riemann} on a weighted FSB-type space.

\subsubsection{Symplectic Ladder Operators}
\label{sec:sympl-ladd-oper}
We can also look for ladder operators within the Lie algebra
\(\algebra{sl}_2\), see~\S~\ref{sec:ellipt-ladd-oper} and~\cite{Kisil09c}*{\S~8}.
Assuming \(\ladder[2]{+}=a\tilde{A}+b\tilde{B}+c\tilde{Z}\) from the
relations~\eqref{eq:sl2-commutator} and defining
condition~\eqref{eq:raising-lowering} we obtain the linear equations
with unknown \(a\), \(b\) and \(c\): 
\begin{displaymath}
  c=0, \qquad 2a=\lambda_+ b, \qquad -2b=\lambda_+ a.
\end{displaymath}
The equations have a solution if and only if \(\lambda_+^2+4=0\), and
the raising/lowering operators are \(\ladder[2]{\pm}=\pm\rmi
\tilde{A}+\tilde{B}\). In the Shale--Weil
representation~\eqref{eq:shale-weil-der} they turn out to be:
\begin{equation}
  \label{eq:ell-ladder-symplect}
  \ladder[2]{\pm}=\pm\rmi\left(\frac{q}{2}\frac{\rmd}{\rmd q}+\frac{1}{4}\right)-\frac{\myhbar\rmi}{8\pi}\frac{\rmd ^2}{\rmd q^2}-\frac{\pi\rmi q^2}{2\myhbar}=-\frac{\rmi}{8\pi\myhbar}\left(\mp2\pi q+\myhbar\frac{\rmd}{\rmd q}\right)^2.
\end{equation}
Since this time \(\lambda_+=2\rmi\) the ladder operators
\(\ladder[2]{\pm}\) produce a shift on the
diagram~\eqref{eq:ladder-chain-1D} twice bigger than the operators
\(\ladder{\pm}\) from the Heisenberg group. After all, this is not
surprising since from the explicit
representations~\eqref{eq:ell-ladder-heisen-rep} and~\eqref{eq:ell-ladder-symplect} we get:
\begin{displaymath}
  \ladder[2]{\pm}=-\frac{\rmi}{8\pi\myhbar}(\ladder{\pm})^2.
\end{displaymath}
\index{ladder operator|)}%
\index{operator!ladder|)}%

\subsection{Hyperbolic Quantum Mechanics}
\label{sec:hyperb-repr-addt}

Now we turn to double numbers%
\index{number!double|indef}%
\index{double!number|indef} also known as hyperbolic, split-complex,
etc. numbers~\citelist{\cite{Yaglom79}*{App.~C} \cite{Ulrych05a}
  \cite{KhrennikovSegre07a}}. They form a two commutative associative
dimensional algebra \(\Space{O}{}\)
spanned by \(1\)
and the hyperbolic unit \(\rmh\)
with the property \(\rmh^2=1\).  There are zero divisors in  \(\Space{O}{}\):
\begin{displaymath}
  \rmh_\pm=\textstyle\frac{1}{\sqrt{2}}(1\pm j), \qquad\text{ such that }\quad
  \rmh_+ \rmh_-=0 
  \quad
  \text{ and }
  \quad
  \rmh_\pm^2=\rmh_\pm.
\end{displaymath}
Thus, double numbers algebraically isomorphic to two copies of
\(\Space{R}{}\) spanned by \(\rmh_\pm\). Being algebraically dull
double numbers are nevertheless interesting as a \(\SL\)-homogeneous
space~\cites{Kisil05a,Kisil09c} and they are relevant in
physics~\cites{Khrennikov05a,Ulrych05a,Ulrych08a,Ulrych2014a}.  The combination of
p-mechanical approach with hyperbolic quantum mechanics was already
discussed in~\cite{BrodlieKisil03a}*{\S~6}.

For the hyperbolic character \(\chi_{\rmh \myh}(s)= \rme^{\rmh \myh
  s}=\cosh \myh s +\rmh\sinh \myh s\)
of \(\Space{R}{}\) one can define
the hyperbolic Fourier-type transform:
\begin{displaymath}
  \hat{k}(q)=\int_{\Space{R}{}} k(x)\, \rme^{-\rmh q x}dx.
\end{displaymath}
It can be understood in the sense of distributions on the space dual
to the set of analytic functions~\cite{Khrennikov08a}*{\S~3}. Hyperbolic
Fourier transform intertwines the derivative \(\frac{\rmd}{\rmd x}\) and
multiplication by \(\rmh q\)~\cite{Khrennikov08a}*{Prop.~1}.
\begin{example}
  For the Gaussian\index{Gaussian} the hyperbolic Fourier transform is the ordinary
  function (note the sign  difference!):
  \begin{displaymath}
    \int_{\Space{R}{}}  \rme^{-x^2/2}  \rme^{-\rmh q x}dx= \sqrt{2\pi}\,  \rme^{q^2/2}.
  \end{displaymath}
  However the opposite identity:
  \begin{displaymath}
    \int_{\Space{R}{}}  \rme^{x^2/2}  \rme^{-\rmh q x}dx= \sqrt{2\pi}\,  \rme^{-q^2/2}
  \end{displaymath}
  is true only in a suitable distributional sense. To this end we may
  note that \( \rme^{x^2/2}\) and \( \rme^{-q^2/2}\) are null solutions to the
  differential operators \(\frac{\rmd}{\rmd x}-x\) and \(\frac{\rmd}{\rmd q}+q\)
  respectively, which are intertwined (up to the factor \(\rmh\)) by
  the hyperbolic Fourier transform. The above differential operators
  \(\frac{\rmd}{\rmd x}-x\) and \(\frac{\rmd}{\rmd q}+q\) are images of the ladder
  operators%
  \index{ladder operator}%
  \index{operator!ladder}~\eqref{eq:ell-ladder-heisen-rep} in the Lie algebra of the Heisenberg group.
  They are intertwining by the Fourier transform, since this is an
  automorphism of the Heisenberg group~\cite{Howe80a}.
\end{example}
An elegant theory of hyperbolic Fourier transform may be achieved by a
suitable adaptation of~\cite{Howe80a}, which uses representation
theory of the Heisenberg group.

\subsubsection{Hyperbolic Representations of the Heisenberg Group} 
\label{sec:segre-quatern-hyperb}

Consider the space
\(\FSpace[\rmh]{F}{\myh}(\Space{H}{})\) of \(\Space{O}{}\)-valued
functions on \(\Space{H}{}\) with the property:
\begin{equation}
  \label{eq:induced-prop-h}
  f(s+s',h,y)= \rme^{\rmh \myh s'} f(s,x,y), \qquad \text{ for all}\quad]
  (s,x,y)\in \Space{H}{},\ s'\in \Space{R}{} 
\end{equation}
and the square integrability condition~\eqref{eq:L2-condition}. Then
the hyperbolic representation of \(\Space{H}{}\)
is obtained by the restriction of the left shifts to
\(\FSpace[\rmh]{F}{\myh}(\Space{H}{})\).
To obtain an equivalent representation on the phase space%
\index{phase!space}%
\index{space!phase} we take the \(\Space{O}{}\)-valued
functional of the Lie algebra \(\algebra{h}\):
\begin{equation}
  \label{eq:hyp-character}
  \chi^j_{(\myh,q,p)}(s,x,y)= \rme^{\rmh(\myh s +qx+ py)}
  =\cosh (\myh s +qx+ py) + \rmh\sinh(\myh s +qx+ py).
\end{equation}
The hyperbolic Fock--Segal--Bargmann type representation%
  \index{Fock--Segal--Bargmann!representation!hyperbolic}%
  \index{representation!Fock--Segal--Bargmann!hyperbolic}%
  \index{hyperbolic!Fock--Segal--Bargmann representation}%
  \index{Heisenberg!group!Fock--Segal--Bargmann representation!hyperbolic}%
  \index{group!Heisenberg!Fock--Segal--Bargmann representation!hyperbolic} is intertwined
with the left group action by means of the Fourier
transform~\eqref{eq:fourier-transform} with the hyperbolic
functional~\eqref{eq:hyp-character}. Explicitly this representation is:
\begin{equation}
  \label{eq:segal-bargmann-hyp}
  \uir{}{\myhbar}(s,x,y): f (q,p) \mapsto 
  \textstyle  \rme^{-\rmh(\myh s+qx+py)}
  f \left(q-\frac{\myh}{2} y, p+\frac{\myh}{2} x\right).
\end{equation}
For a hyperbolic Schr\"odinger type representation%
\index{representation!Heisenberg group!hyperbolic}%
\index{Schr\"odinger!representation!hyperbolic}%
\index{hyperbolic!Schr\"odinger representation!}%
\index{representation!Schr\"odinger!hyperbolic} we again use the
scheme described in \S~\ref{sec:concl-induc-repr}. Similarly to the
elliptic case one obtains the formula,
resembling~\eqref{eq:schroedinger-rep}:
\begin{equation}
    \label{eq:schroedinger-rep-hyp}
    [\uir{\rmh}{\chi}(s',x',y') f](x)= \rme^{-\rmh\myh (s'+xy'-x'y'/2)}f(x-x').
\end{equation}
Application of the hyperbolic Fourier transform produces a
Schr\"odinger type representation on the configuration space%
\index{configuration!space}%
\index{space!configuration},
cf.~\eqref{eq:schroedinger-rep-conf}: 
\begin{displaymath}
  [\uir{\rmh}{\chi}(s',x',y') \hat{f}\,](q)= \rme^{-\rmh\myh (s'+x'y'/2)
    -\rmh x' q}\,\hat{f}(q+\myh y').  
\end{displaymath}
The extension of this representation to kernels according
to~\eqref{eq:rho-extended-to-L1} generates hyperbolic
pseudodifferential operators introduced
in~\cite{Khrennikov08a}*{(3.4)}.  

\subsubsection{Hyperbolic Dynamics}
\label{sec:hyperbolic-dynamics}

Similarly to the elliptic (quantum) case we consider a convolution
of two kernels on \(\Space{H}{}\) restricted to
\(\FSpace[\rmh]{F}{\myh}(\Space{H}{})\). The composition law becomes,
cf.~\eqref{eq:composition-ell}:
\begin{equation}
  \label{eq:composition-par}
  (k'*k)\hat{_s}
  =
  \int_{\Space{R}{2n}}  \rme^{ {\rmh \myh}{}(xy'-yx')}\, \hat{k}'_s(\myh ,x',y')\,
 \hat{k}_s(\myh ,x-x',y-y')\,\rmd x' \rmd y'. 
\end{equation}
This is close to the calculus of hyperbolic PDO obtained
in~\cite{Khrennikov08a}*{Thm.~2}.
Respectively for the commutator of two convolutions we get,
cf.~\eqref{eq:repres-commutator}:
\begin{equation}
  \label{eq:commut-par}
  [k',k]\hat{_s}
  = 
  \int_{\Space{R}{2n}}\!\! \sinh(\myh
   (xy'-yx'))\, \hat{k}'_s(\myh ,x',y')\,
 \hat{k}_s(\myh ,x-x',y-y')\,\rmd x'\rmd y'. 
\end{equation}
This is the hyperbolic version of the Moyal bracket%
\index{Moyal bracket!hyperbolic}%
\index{bracket!Moyal!hyperbolic}%
\index{hyperbolic!Moyal bracket},
cf.~\cite{Khrennikov08a}*{p.~849}, which generates the corresponding
image of the dynamic equation~\eqref{eq:universal}.
\begin{example}
  \begin{enumerate}
  \item 
    \label{it:hyp-harm-oscil} 
    For a quadratic Hamiltonian, e.g.  harmonic oscillator%
    \index{harmonic!oscillator}%
    \index{oscillator!harmonic} from
    Example~\ref{ex:p-harmonic}, the hyperbolic equation and
    respective dynamics is identical to quantum considered before.
  \item Since \(\frac{\partial}{\partial s}\) acts on
    \(\FSpace[\rmh]{F}{2}(\Space{H}{})\) as multiplication by \(\rmh
    \myh\) and \(\rmh^2=1\), the hyperbolic image of the unharmonic
    equation%
    \index{unharmonic!oscillator!hyperbolic}%
    \index{hyperbolic!unharmonic oscillator}%
    \index{oscillator!unharmonic!hyperbolic}~\eqref{eq:p-unharm-osc-dyn} becomes:
    \begin{displaymath}
      \dot{f}=     \left(m  k^2 q\frac{\partial}{\partial
          p}+\frac{\lambda}{6}\left(3q^2\frac{\partial}{\partial p}
          +\frac{\myhbar^2}{4}\frac{\partial^3}{\partial p^3}\right)-\frac{1}{m}
        p \frac{\partial}{\partial q} \right) f. 
    \end{displaymath}
    The difference with quantum mechanical
    equation~\eqref{eq:q-unhar-dyn} is in the sign of the cubic
    derivative. 
  \end{enumerate}
\end{example}
Notably, the hyperbolic setup allows us to linearise many non-linear
problems of classical mechanics%
\index{non-linear dynamics}%
\index{dynamics!non-linear}. It will be interesting to realise new
hyperbolic coordinates introduced to this end
in~\cites{Pilipchuk10a,Pilipchuk11a,PilipchukAndrianovMarkert16a} as a
hyperbolic phase space.

\subsubsection{Hyperbolic Probabilities}
\label{sec:hyperb-prob}
\begin{figure}[htbp]
  \centering
  (a)\includegraphics[scale=.75]{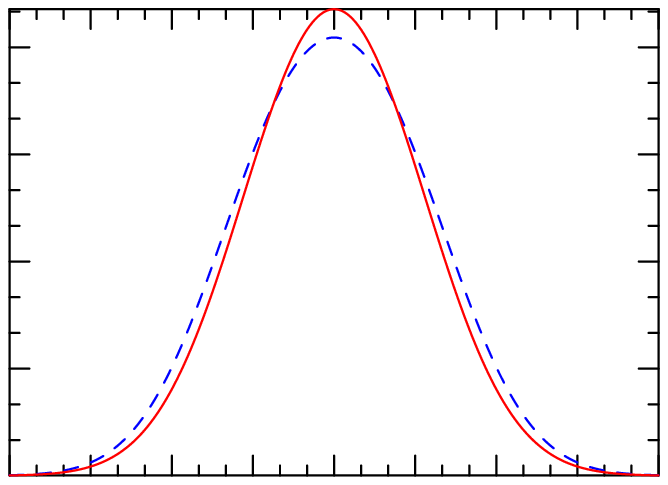}\hfill
  (b)\includegraphics[scale=.75]{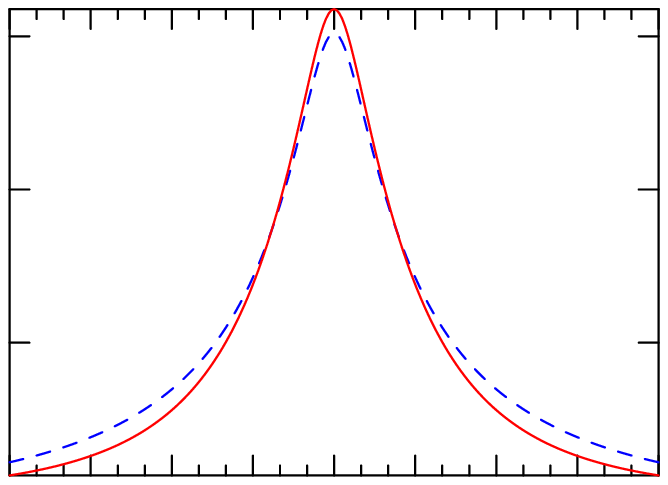}
  \caption{Hyperbolic probabilities: the blue (dashed) graph shows the
    addition of probabilities without interaction, the red (solid)
    graph present the hyperbolic quantum interference. Left picture
    shows the Gaussian state~\eqref{eq:gauss-state}, with the same
    distribution as in quantum mechanics,
    cf.~Fig.~\ref{fig:quant-prob}(a). The right picture shows the
    rational state~\eqref{eq:poly-state}, note the absence of
    interference oscillations in comparison with the quantum state
    on~Fig.~\ref{fig:quant-prob}(b).}
  \label{fig:hyp-prob}
\end{figure}

To calculate probability%
\index{probability!hyperbolic}%
\index{hyperbolic!probability} distribution generated by a hyperbolic state we are
using the general procedure from
Section~\ref{sec:states-probability}. The main differences with the
quantum case are as follows:
\begin{enumerate}
\item The real number \(A\) in the
  expression~\eqref{eq:addition-functional} for the addition of
  probabilities is bigger than \(1\) in absolute value. Thus, it can
  be associated with the hyperbolic cosine \(\cosh \alpha \),
  cf.~Rem.~\ref{re:sine-cosine}, for certain phase
  \(\alpha\in\Space{R}{}\)~\cite{Khrennikov08a}.
\item The nature of hyperbolic interference on two slits is affected
  by the fact that \( \rme^{\rmh \myh s}\) is not periodic and the
  hyperbolic exponent \( \rme^{\rmh t}\) and cosine \(\cosh t\) do not
  oscillate. It is worth to notice that for Gaussian\index{Gaussian} states the
  hyperbolic interference is exactly the same as quantum one,
  cf.~Figs.~\ref{fig:quant-prob}(a) and~\ref{fig:hyp-prob}(a). This is
  similar to coincidence of quantum and hyperbolic dynamics of
  harmonic oscillator.

  The contrast between two types of interference is prominent for
  the rational state~\eqref{eq:poly-state}, which is far from the
  minimal uncertainty, see the different patterns
  on Figs.~\ref{fig:quant-prob}(b) and~\ref{fig:hyp-prob}(b).
\end{enumerate}

\subsubsection{Ladder Operators for the Hyperbolic Subgroup}
\label{sec:hiperbolic-subgroup}%
\index{ladder operator|(}%
\index{operator!ladder|(}

Consider the case of the Hamiltonian \(H=2B\), which is a repulsive
(hyperbolic) harmonic oscillator%
\index{harmonic!oscillator!repulsive (hyperbolic)}%
\index{repulsive!harmonic oscillator}%
\index{hyperbolic!harmonic oscillator}%
\index{oscillator!harmonic!repulsive
  (hyperbolic)}~\cite{Wulfman10a}*{\S~3.8}. The corresponding
one-dimensional subgroup of symplectomorphisms produces hyperbolic
rotations of the phase space, see Fig.~\ref{fig:rotations}.%
\index{symplectic!transformation}%
\index{transformation!symplectic} The
eigenvectors \(v_\mu\) of the operator
\begin{displaymath}
  \uir{\text{SW}}{\myhbar}(2B)v_\nu
  =-\rmi\left(\frac{\myhbar}{4\pi}\frac{\rmd ^2}{\rmd q^2}+\frac{\pi q^2}{\myhbar}\right)v_\nu
  =\rmi\nu v_\nu 
\end{displaymath}
are \emph{Weber--Hermite}%
\index{Weber--Hermite function}%
\index{function!Weber--Hermite} (or \emph{parabolic cylinder}%
\index{parabolic!cylinder function|see{Weber--Hermite function}}%
\index{function!parabolic cylinder|see{Weber--Hermite function}}) functions
\(v_{\nu}=D_{\nu-\frac{1}{2}}\left(\pm2 \rme^{\rmi \frac{\pi}{4}}\sqrt{\frac{\pi}{\myhbar}} q\right)\),
see~\citelist{\cite{ErdelyiMagnusII}*{\S~8.2}
  \cite{SrivastavaTuanYakubovich00a}} for fundamentals of
Weber--Hermite functions and~\cite{ATorre08a} for further
illustrations and applications in optics.

The corresponding one-parameter group is not compact and the
eigenvalues of the operator \(2\tilde{B}\) are not restricted by any
integrality condition, but the raising/lowering operators are
still important~\citelist{\cite{HoweTan92}*{\S~II.1}
  \cite{Mazorchuk09a}*{\S~1.1}}. We again seek solutions in two
subalgebras \(\algebra{h}\) and \(\algebra{sl}_2\) separately.
However, the additional options will be provided by a choice of the
number system: either complex or double.

\begin{example}[Complex Ladder Operators]
\label{sec:compl-ladd-oper}
  
Assuming
\(\ladder[h]{+}=a\tilde{X}+b\tilde{Y}\) from the
commutators~\eqref{eq:cross-comm}--\eqref{eq:cross-comm1} we obtain
the linear equations:
\begin{equation}
  \label{eq:hyp-ladder-compatib}
  -a=\lambda_+ b, \qquad -b=\lambda_+ a.
\end{equation}
The equations have a solution if and only if \(\lambda_+^2-1=0\).
Taking the real roots \(\lambda=\pm1\) we obtain that the
raising/lowering operators are
\(\ladder[h]{\pm}=\tilde{X}\mp\tilde{Y}\).  In the Schr\"odinger
representation%
\index{Schr\"odinger!representation}%
\index{representation!Schr\"odinger}~\eqref{eq:schroedinger-rep-conf-der}
the ladder operators are
\begin{equation}
  \label{eq:ell-ladder-heisen-rep1}
  \ladder[h]{\pm}= 2\pi\rmi q\pm \myhbar \frac{\rmd}{\rmd q}.
\end{equation}
The null solutions \(v_{\pm\frac{1}{2}}(q)= \rme^{\pm\frac{\pi\rmi}{\myhbar}
  q^2}\) to operators \(\uir{}{\myhbar}(\ladder{\pm})\) are also
eigenvectors of the Hamiltonian \(\uir{\text{SW}}{\myhbar}(2B)\) with the
eigenvalue \(\pm\frac{1}{2}\).  However the important distinction from
the elliptic case is, that they are not square-integrable on the real line
anymore.

We can also look for ladder operators within the \(\algebra{sl}_2\),
that is in the form \(\ladder[2h]{+}=a\tilde{A}+b\tilde{B}+c\tilde{Z}\)
for the commutator \([2\tilde{B},\ladder[h]{+}]=\lambda
\ladder[h]{+}\), see \S~\ref{sec:hyperb-ladd-oper}. Within complex
numbers we get only the values \(\lambda=\pm 2\) with the ladder
operators \(\ladder[2h]{\pm}=\pm2\tilde{A}+\tilde{Z}/2\),
see~\citelist{\cite{HoweTan92}*{\S~II.1}
  \cite{Mazorchuk09a}*{\S~1.1}}. Each indecomposable \(\algebra{h}\)-
or \(\algebra{sl}_2\)-module is formed by a one-dimensional chain of
eigenvalues with a transitive action of ladder operators \(\ladder[h]{\pm}\)
or \(\ladder[2h]{\pm}\) respectively. And we again have a
quadratic relation between the ladder operators:
\begin{displaymath}
  \ladder[2h]{\pm}=\frac{\rmi}{4\pi\myhbar}(\ladder[h]{\pm})^2.
\end{displaymath}
\end{example}

\subsubsection{Double Ladder Operators}
\label{sec:double-ladd-oper}
  
There are extra possibilities in in the context of hyperbolic quantum
mechanics~\citelist{\cite{Khrennikov03a} \cite{Khrennikov05a}
  \cite{Khrennikov08a}}.  Here we use the representation of
\(\Space{H}{}\) induced by a hyperbolic character \( \rme^{\rmh \myh
  t}=\cosh (\myh t)+\rmh\sinh(\myh t)\), see
\eqref{eq:segal-bargmann-hyp} and~\cite{Kisil10a}*{(4.5)}, and obtain
the hyperbolic representation of \(\Space{H}{}\),
cf.~\eqref{eq:schroedinger-rep-conf}:  
\begin{equation}
  \label{eq:schroedinger-rep-conf-hyp}
    [\uir{\rmh}{\myh}(s',x',y') \hat{f}\,](q)= \rme^{\rmh\myh (s'-x'y'/2)
    +\rmh x' q}\,\hat{f}(q-\myh y').  
\end{equation}
The corresponding derived representation is
\begin{equation}
  \label{eq:schroedinger-rep-conf-der-hyp}
  \uir{\rmh}{\myh}(X)=\rmh q,\qquad \uir{\rmh}{\myh}(Y)=-\myh \frac{\rmd}{\rmd q},
  \qquad
  \uir{\rmh}{\myh}(S)=\rmh\myh I.
\end{equation}
Then the associated Shale--Weil derived  representation of \(\algebra {sp}_2\) in
the Schwartz space \(\FSpace{S}{}(\Space{R}{})\) is, cf.~\eqref{eq:shale-weil-der}:
\begin{equation}
  \label{eq:shale-weil-der-double}
  \uir{\text{SW}}{\myh}(A) =-\frac{q}{2}\frac{\rmd}{\rmd q}-\frac{1}{4},\quad
  \uir{\text{SW}}{\myh}(B)=\frac{\rmh\myh}{4}\frac{\rmd ^2}{\rmd q^2}-\frac{\rmh q^2}{4\myh},\quad
  \uir{\text{SW}}{\myh}(Z)=-\frac{\rmh\myh}{2}\frac{\rmd ^2}{\rmd q^2}-\frac{\rmh q^2}{2\myh}.
\end{equation}
Note that \(\uir{\text{SW}}{\myh}(B)\) now generates a usual harmonic
oscillator, not the repulsive one like 
\(\uir{\text{SW}}{\myhbar}(B)\) in \eqref{eq:shale-weil-der}. 
However, the expressions in the quadratic algebra are still the same (up to a factor),
cf.~\eqref{eq:quadratic-A}--\eqref{eq:quadratic-Z}:
\begin{align*}
  \qquad\uir{\text{SW}}{\myh}(A) &=
  -\frac{\rmh}{2\myh}(\uir{\rmh}{\myh}(X)\uir{\rmh}{\myh}(Y)
  -{\textstyle\frac{1}{2}}\uir{\rmh}{\myh}(S))\\
  &=-\frac{\rmh}{4\myh}(\uir{\rmh}{\myh}(X)\uir{\rmh}{\myh}(Y)
  +\uir{\rmh}{\myh}(Y)\uir{\rmh}{\myh}(X)),\nonumber \\ 
  \uir{\text{SW}}{\myh}(B) &=
  \frac{\rmh}{4\myh}(\uir{\rmh}{\myh}(X)^2-\uir{\rmh}{\myh}(Y)^2), \\
  \uir{\text{SW}}{\myh}(Z)
  &=-\frac{\rmh}{2\myh}(\uir{\rmh}{\myh}(X)^2+\uir{\rmh}{\myh}(Y)^2). 
\end{align*}
This is due to the Principle~\ref{pr:simil-corr-principle} of
similarity and correspondence%
\index{principle!similarity and correspondence}: we can swap operators \(Z\) and \(B\) with
simultaneous replacement of hypercomplex units \(\rmi\) and \(\rmh\).

The eigenspace of the operator \(2\uir{\text{SW}}{\myh}(B)\) with an
eigenvalue \(\rmh \nu\) are spanned by the Weber--Hermite
functions%
\index{Weber--Hermite function}%
\index{function!Weber--Hermite} \(D_{-\nu-\frac{1}{2}}\left(\pm\sqrt{\frac{2}{\myh}}x\right)\),
see~\cite{ErdelyiMagnusII}*{\S~8.2}.  Functions \(D_\nu\) are
generalisations of the Hermit functions%
\index{Hermite!polynomial}%
\index{polynomial!Hermite}~\eqref{eq:hermit-poly}.

The compatibility condition for a ladder operator within the Lie algebra
\(\algebra{h}\) will be~\eqref{eq:hyp-ladder-compatib} as before,
since it depends only on the
commutators~\eqref{eq:cross-comm}--\eqref{eq:cross-comm1}. Thus, we still
have the set of ladder operators corresponding to values
\(\lambda=\pm1\):
\begin{displaymath}
  \ladder[h]{\pm}=\tilde{X}\mp\tilde{Y}=\rmh q\pm\myh \frac{\rmd}{\rmd q}.
\end{displaymath}
Admitting double numbers%
\index{number!double}%
\index{double!number} we have an extra way to satisfy
\(\lambda^2=1\) in~\eqref{eq:hyp-ladder-compatib} with values
\(\lambda=\pm\rmh\).  Then there is an additional pair of hyperbolic
ladder operators, which are identical (up to factors)
to~\eqref{eq:ell-ladder-heisen-rep}:
\begin{displaymath}
  \ladder[\rmh]{\pm}=\tilde{X}\mp\rmh\tilde{Y}=\rmh q\pm\rmh\myh \frac{\rmd}{\rmd q}.
\end{displaymath}
Pairs \(\ladder[h]{\pm}\) and \(\ladder[\rmh]{\pm}\) shift
eigenvectors in the ``orthogonal'' directions changing their
eigenvalues by \(\pm1\) and \(\pm\rmh\).  Therefore an indecomposable
\(\algebra{sl}_2\)-module can be para\-metrised by a two-dimensional
lattice of eigenvalues in double numbers, see
Fig.~\ref{fig:2D-lattice}.

The following functions 
\begin{eqnarray*}
  v_{\frac{1}{2}}^{\pm\myh}(q)&=& \rme^{\mp\rmh
    q^2/(2\myh)}=\cosh\frac{q^2}{2\myh}\mp \rmh\sinh \frac{q^2}{2\myh},\\
  v_{\frac{1}{2}}^{\pm\rmh}(q)&=& \rme^{\mp  q^2/(2\myh)}
\end{eqnarray*}
are null solutions to the operators \(\ladder[h]{\pm}\) and
\(\ladder[\rmh]{\pm}\) respectively. They are also eigenvectors of
\(2\uir{\text{SW}}{\myh}(B)\) with eigenvalues \(\mp\frac{\rmh}{2}\)
and \(\mp\frac{1}{2}\) respectively. If these functions are used as
mother wavelets for the wavelet transforms generated by the Heisenberg
group, then the image space will consist of the null-solutions of the
following differential operators, see Cor.~\ref{co:cauchy-riemann-integ}:
\begin{align*}\textstyle
  D_{h}&=\overline{X^{r} - Y^{r}}=(\partial_{ x} -\partial_{y})+\frac{\myh}{2}(x+y)\notingiq
\\
  D_{\rmh}&=\overline{X^{r} - \rmh Y^{r}}=(\partial_{ x} +\rmh\partial_{y})-\frac{\myh}{2}(x-\rmh
y)\notingiq
\end{align*}
for \(v_{\frac{1}{2}}^{\pm\myh}\) and \(v_{\frac{1}{2}}^{\pm\rmh}\)
respectively. This is again in line with the classical
result~\eqref{eq:CR-Bargmann}. However, annihilation of the eigenvector
by a ladder operator does not mean that the part of the 2D-lattice
becomes void since it can be reached via alternative routes on this
lattice. Instead of multiplication by a zero, as it happens in the
elliptic case, a half-plane of eigenvalues will be multiplied by the
divisors of zero%
\index{divisor!zero}%
\index{zero!divisor} \(1\pm\rmh\).

We can also search ladder operators within the algebra
\(\algebra{sl}_2\) and admitting double numbers we will again find two
sets of them, cf.~\S~\ref{sec:hyperb-ladd-oper}:
\begin{eqnarray*}
  \ladder[2h]{\pm} &=&\pm\tilde{A}+\tilde{Z}/2 =
   \mp\frac{q}{2}\frac{\rmd}{\rmd q}\mp\frac{1}{4}- \frac{\rmh\myh}{4}\frac{\rmd ^2}{\rmd q^2}-\frac{\rmh q^2}{4\myh}=-\frac{\rmh}{4\myh}(\ladder[h]{\pm})^2,\\
  \ladder[2\rmh]{\pm}&=&\pm\rmh\tilde{A}+\tilde{Z}/2=  
  \mp\frac{\rmh q}{2}\frac{\rmd}{\rmd q}\mp\frac{\rmh}{4}-\frac{\rmh\myh}{4}\frac{\rmd ^2}{\rmd q^2}-\frac{\rmh q^2}{4\myh}=-\frac{\rmh}{4\myh}(\ladder[\rmh]{\pm})^2.
\end{eqnarray*}
Again the operators \(\ladder[2h]{\pm}\) and \(\ladder[2h]{\pm}\) produce
double shifts in the orthogonal directions on the same two-dimensional
lattice in Fig.~\ref{fig:2D-lattice}.%
\index{ladder operator|)}%
\index{operator!ladder|)}%

\subsection{Parabolic (Classical) Representations on the Phase Space}
\label{sec:class-repr-phase}
After the previous two cases it is natural to link classical mechanics
with dual numbers%
\index{dual!number|indef}%
\index{number!dual|indef}%
\index{dual!number|(}%
\index{number!dual|(} generated by the parabolic unit \(\rmp^2=0\).
Connection of the parabolic unit \(\rmp\)
with the Galilean group of symmetries of classical mechanics is around
for a while~\cite{Yaglom79}*{App.~C}, for other applications
see~\citelist{\cite{CatoniCannataNichelatti04}
  \cite{Zejliger34}*{\S~I.2(10)} \cite{Gromov90a} \cite{Dimentberg78a}
  \cite{Dimentberg78b}}.

However, the nilpotency of the parabolic unit \(\rmp\)
makes it difficult if we will work with dual number valued functions
only.  To overcome this issue we consider a commutative and
associative four-dimensional real algebra \(\algebra{C}\)
spanned by \(1\),
\(\rmi\),
\(\rmp\)
and \(\rmi\rmp\)
with identities \(\rmi^2=-1\)
and \(\rmp^2=0\). A seminorm on \(\algebra{C}\) is defined as follows:
\begin{displaymath}
  \modulus{a+b\rmi+c\rmp+d\rmi\rmp}^2=a^2+b^2.
\end{displaymath}

\subsubsection{Classical Non-Commutative Representations}
\label{sec:class-non-comm}
We wish to build a representation of the Heisenberg group which will
be a classical analog of the Fock--Segal--Barg\-mann
representation~\eqref{eq:stone-inf}.  To this end, we introduce the
space \(\FSpace[\rmp]{F}{\myh}(\Space{H}{})\) of
\(\algebra{C}\)-valued functions on \(\Space{H}{}\) with the
property:
\begin{equation}
  \label{eq:induced-prop-p}
  f(s+s',h,y)= \rme^{\rmp \myh s'} f(s,x,y), \qquad \text{ for all}\quad]
  (s,x,y)\in \Space{H}{},\ s'\in \Space{R}{} 
\end{equation}
and the square integrability condition~\eqref{eq:L2-condition}. Here
as before, \( \rme^{\rmp \myh s'}=1+\rmp \myh s'\) in line with the Taylor
expansion of the exponent. The described space  is
invariant under the left shifts and we restrict the left group action to
\(\FSpace[\rmp]{F}{\myh}(\Space{H}{})\). 

An equivalent form of the induced representation acts on
\(\FSpace[\rmp]{F}{\myh}(\Space{R}{2n})\), where \(\Space{R}{2n})\) is
the homogeneous space of \(\Space{H}{}\) over its centre.  The
Fourier transform \((x,y)\mapsto(q,p)\) intertwines the last
representation with the following action on \(\algebra{C}\)-valued
functions on the phase space:%
\index{Fock--Segal--Bargmann!representation!classic (parabolic)}%
\index{representation!Fock--Segal--Bargmann!classic (parabolic)}%
\index{classic!Fock--Segal--Bargmann representation}%
\index{parabolic!Fock--Segal--Bargmann representation}%
\index{Heisenberg!group!Fock--Segal--Bargmann representation!classic
  (parabolic)}%
\index{group!Heisenberg!Fock--Segal--Bargmann representation!classic
  (parabolic)}
\begin{eqnarray}
  \label{eq:dual-repres}
  \uir{\rmp}{\myh}(s,x,y): f(q,p) &\mapsto&  \rme^{-2\pi\rmi(xq+yp)}(f(q,p)\\
  &&\qquad +\rmp\hbar(2\pi s f(q,p)
  -\frac{y\rmi}{2}f'_q(q,p)+\frac{x\rmi}{2}f'_p(q,p))).
  \nonumber
\end{eqnarray}
Note, that for any real polynomial  \(p(x)\) 
algebraic manipulations show that \(p(x+\rmp y)=p(x)+\rmp y
p'(x)\). If extend this rule to any differentiable function then
\eqref{eq:dual-repres} can be re-written as:
\begin{equation}
  \label{eq:dual-as-SB-1}
  \uir{}{\myhbar}(s,x,y): f (q,p) \mapsto 
   \rme^{-2\pi(\rmp \myhbar s+\rmi(qx+py))}
  f \left(q-\frac{\rmi\myhbar}{2} \rmp y, p+\frac{\rmi\myhbar}{2} \rmp
    x\right).
\end{equation}
The later form completely agrees with FSB representation~\eqref{eq:stone-inf}. 
\begin{remark}
  \label{re:classic-rep}
  Comparing the traditional
  infinite-dimensional~\eqref{eq:stone-inf} and
  one-dimen\-sional~\eqref{eq:commut-repres} representations of
  \(\Space{H}{}\) we can note that the properties of the
  representation~\eqref{eq:dual-repres} are a non-trivial mixture of
  the former:  
  \begin{enumerate}
  \item \label{it:class-non-commut} 
    The action~\eqref{eq:dual-repres}
    is non-commutative, similarly to the quantum
    representation~\eqref{eq:stone-inf} and unlike the classical
    one~\eqref{eq:commut-repres}. This non-commutativity will produce
    the Hamilton equations below in a way very similar to Heisenberg
    equation, see Rem.~\ref{re:hamilton-from-nc}.
  \item \label{it:class-locality} The
    representation~\eqref{eq:dual-repres} does not change the support
    of a function \(f\) on the phase space, similarly to the
    classical representation~\eqref{eq:commut-repres} and unlike the
    quantum one~\eqref{eq:stone-inf}. Such a localised action will be
    responsible later for an absence of an interference in classical
    probabilities.
  \item The parabolic representation~\eqref{eq:dual-repres} can not be
    derived from either the elliptic~\eqref{eq:stone-inf} or
    hyperbolic~\eqref{eq:segal-bargmann-hyp} by the plain substitution
    \(\myh=0\).
  \end{enumerate}
\end{remark}
We may also write a classical Schr\"odinger type representation.%
\index{representation!Heisenberg group!classic (parabolic)}%
\index{Schr\"odinger!representation!classic (parabolic)}%
\index{parabolic!Schr\"odinger representation!}%
\index{classic!Schr\"odinger representation!}%
\index{representation!Schr\"odinger!classic (parabolic)}
According to \S~\ref{sec:concl-induc-repr} we get a representation formally
very similar to the elliptic~\eqref{eq:schroedinger-rep} and
hyperbolic versions~\eqref{eq:schroedinger-rep-hyp}:
\begin{eqnarray}
    \label{eq:schroedinger-rep-par}
    [\uir{\rmp}{\chi}(s',x',y') f](x)&=& \rme^{-\rmp\myh
      (s'+xy'-x'y'/2)}f(x-x')\\
    &=&(1-\rmp\myh (s'+xy'-\textstyle\frac{1}{2}x'y')) f(x-x').\nonumber 
\end{eqnarray}
However due to nilpotency of \(\rmp\) the (complex) Fourier transform
\(x\mapsto q\) produces a different formula for parabolic
Schr\"odinger type representation in the configuration space%
\index{configuration!space}%
\index{space!configuration},
cf.~\eqref{eq:schroedinger-rep-conf}
and~\eqref{eq:schroedinger-rep-conf-hyp}:
\begin{align}
  [\uir{\rmp}{\chi}(s',x',y') \hat{f}](q)&=  \rme^{2\pi\rmi x' q}\left(
                                           \left(1-\rmp\myh (s'-{\textstyle\frac{1}{2}}x'y')\right)    \hat{f}(q)
                                           +\rmp\rmi\hbar
                                           y'\hat{f}'(q)\right)
                                           \nonumber \\
                                         &=  \rme^{2\pi(-\rmp\hbar (s'-{\textstyle\frac{1}{2}}x'y')+\rmi x' q)}
                                           \hat{f}(q+\rmp\rmi\hbar
                                           y').
                                           \label{eq:schroedinger-rep-conf-der-par}
\end{align}
This representation shares all properties mentioned in
Rem.~\ref{re:classic-rep} as well.

\subsubsection{Hamilton Equation}
\label{sec:hamilton-equation}


The identity \( \rme^{ \rmp t}- \rme^{ -\rmp t}= 2\rmp t\)
suggests that a parabolic version of the sine function is the identity
function, while the parabolic cosine is identically equal to one,
cf. \S~\ref{sec:hyperc-char}
and~\cites{HerranzOrtegaSantander99a,Kisil07a}.  From this we obtain
the parabolic version of the commutator~\eqref{eq:repres-commutator}:
\begin{eqnarray*}
  [k',k]\hat{_s}(\rmp \myh, x,y) 
  &=& 
  \rmp \myh\int_{\Space{R}{2n}}
 (xy'-yx') \\
 &&{}\times\, \hat{k}'_s(\rmp \myh,x',y')  \,
 \hat{k}_s(\rmp \myh,x-x',y-y')\,\rmd x'\rmd y', \nonumber 
\end{eqnarray*}
for the partial parabolic Fourier-type transform \(\hat{k}_s\) of the
kernels.  Thus, the parabolic representation of the dynamical
equation~\eqref{eq:universal} becomes:
\begin{eqnarray}
  \label{eq:dynamics-par}
  \lefteqn{
  \rmp\myh \frac{\rmd \hat{f}_s}{\rmd t}(\rmp\myh,x,y;t)}\quad&\\
  =&\rmp \myh \int_{\Space{R}{2n}}
 (xy'-yx')\, 
\hat{H}_s(\rmp \myh,x',y')  \,
 \hat{f}_s(\rmp \myh,x-x',y-y';t)\,\rmd x'\rmd y'. \nonumber 
\end{eqnarray}
Although there is no possibility to divide by \(\rmp\) (since it is a
zero divisor) we can obviously eliminate \(\rmp \myh \) from the both
sides if the rest of the expressions are real.  Moreover, this can be
done ``in advance'' through a kind of the antiderivative operator
considered in~\cite{Kisil02e}*{(4.1)}. This will prevent ``imaginary
parts'' of the remaining expressions (which contain the factor
\(\rmp\)) from vanishing.
\begin{remark}
  It is noteworthy that the Planck constant%
  \index{Planck!constant}%
  \index{constant!Planck} completely disappeared
  from the dynamical equation. Thus the only prediction about it
  following from our construction is \(\myh\neq 0\), which was
  confirmed by experiments, of course. 
\end{remark}
Using the duality between the Lie algebra \(\algebra{h}\)
of \(\Space{H}{}\)
and the phase space we can find an adjoint equation for observables on
the phase space. To this end we apply the usual Fourier transform
\((x,y)\mapsto(q,p)\).
It turn to be the Hamilton equation \eqref{eq:hamilton-poisson}%
\index{Hamilton!equation}%
\index{equation!Hamilton}~\cite{Kisil02e}*{(4.7)}.  However, the
transition to the phase space is more a custom rather than a necessity
and in many cases we can efficiently work on the Heisenberg group
itself.

\begin{remark}
  \label{re:hamilton-from-nc}
  It is noteworthy, that the non-commutative
  representation~\eqref{eq:dual-repres} allows to obtain the Hamilton
  equation directly from the commutator
  \([\uir{\rmp}{\myh}(k_1),\uir{\rmp}{\myh}(k_2)]\). Indeed, its
  straightforward  evaluation will produce exactly the above expression. On
  the contrast such a commutator for the commutative
  representation~\eqref{eq:commut-repres} is zero and to obtain the
  Hamilton equation we have to work with an additional tools, e.g. an
  anti-derivative~\cite{Kisil02e}*{(4.1)}. 
\end{remark}

\begin{example}
  \begin{enumerate}
  \item For the harmonic oscillator%
    \index{harmonic!oscillator}%
    \index{oscillator!harmonic} in Example~\ref{ex:p-harmonic} the
    equation~\eqref{eq:dynamics-par} again reduces to the
    form~\eqref{eq:p-harm-osc-dyn} with the solution given
    by~\eqref{eq:p-harm-sol}. The adjoint equation of the harmonic
    oscillator on the phase space is not different from the quantum
    written in
    Example~\ref{ex:quntum-oscillators}(\ref{it:q-harmonic}). This is
    true for any Hamiltonian of at most quadratic order.
  \item 
    For non-quadratic Hamiltonians classical and quantum dynamics
    are different, of course. For example, 
    the cubic term of \(\partial_s\) in the
    equation~\eqref{eq:p-unharm-osc-dyn} will generate the factor
    \(\rmp^3=0\) and thus vanish. Thus the
    equation~\eqref{eq:dynamics-par} of the unharmonic oscillator%
    \index{unharmonic!oscillator}%
    \index{oscillator!unharmonic} on
    \(\Space{H}{}\) becomes:
    \begin{displaymath}
      \dot{f}=     \left(m  k^2 y\frac{\partial}{\partial x}
        +\frac{\lambda y}{2}\frac{\partial^2}{\partial x^2} 
          -\frac{1}{m} x
        \frac{\partial}{\partial y} \right) f. 
   \end{displaymath}
   The adjoint equation on the phase space is:
    \begin{displaymath}
      \dot{f}=     \left(\left(m  k^2
          q+\frac{\lambda}{2}q^2\right)
        \frac{\partial}{\partial p} -\frac{1}{m}  p \frac{\partial}{\partial q} \right) f. 
    \end{displaymath}
    The last equation is the classical
    Hamilton equation generated by the
    cubic potential~\eqref{eq:unharmonic-hamiltonian}. Qualitative
    analysis of its dynamics can be found in many textbooks
    \citelist{\cite{Arnold91}*{\S~4.C, Pic.~12} \cite{PercivalRichards82}*{\S~4.4}}. 
  \end{enumerate}
\end{example}

\begin{remark}
  We have obtained the \emph{Poisson bracket}%
  \index{Poisson!bracket}%
  \index{bracket!Poisson} from the commutator of
  convolutions on \(\Space{H}{}\) without any quasiclassical limit
  \(\myh\rightarrow 0\). This has a common source with the deduction
  of main calculus theorems in~\cite{CatoniCannataNichelatti04} based
  on dual numbers. As explained in~\cite{Kisil05a}*{Rem.~6.9} this is
  due to the similarity between the parabolic unit \(\rmp\) and the
  infinitesimal number used in non-standard analysis~\cite{Devis77}.
  In other words, we never need to take care about terms of order
  \(O(\myh^2)\) because they will be wiped out by \(\rmp^2=0\).
\end{remark}
An alternative derivation of classical dynamics from the Heisenberg
group is given in the recent paper~\cite{Low09a}.

\subsubsection{Classical Probabilities}
\label{sec:class-prob}
It is worth to notice that dual numbers are not only helpful in
reproducing classical Hamiltonian dynamics, they also provide the
classic rule for addition of probabilities.%
\index{probability!classic (parabolic)}%
\index{classic!probability}%
\index{parabolic!probability|see{classic probability}}
We use the same formula~\eqref{eq:kernel-state} to calculate kernels of
the states. The important difference now that the
representation~\eqref{eq:dual-repres} does not change the support of
functions. Thus if we calculate the correlation term
\(\scalar{v_1}{\uir{}{}(g)v_2}\) in~\eqref{eq:kernel-add}, then it
will be zero for every two vectors \(v_1\) and \(v_2\) which have 
disjoint supports in the phase space. Thus no interference%
\index{interference} similar to
quantum or hyperbolic cases (Subsection~\ref{sec:quantum-probabilities})
is possible.

\subsubsection{Ladder Operator for the Nilpotent Subgroup}
\label{sec:nilpotent-subgroup}%
\index{ladder operator|(}%
\index{operator!ladder|(}

Finally we look for ladder operators for the Hamiltonian
\(\tilde{B}+\tilde{Z}/2\) or, equivalently,
\(-\tilde{B}+\tilde{Z}/2\). It can be identified with a free
particle~\cite{Wulfman10a}*{\S~3.8}. 

We can search for ladder operators in the
representation~\eqref{eq:schroedinger-rep-conf-der}--\eqref{eq:shale-weil-der}
within the Lie algebra \(\algebra{h}\) in the form
\(\ladder[\rmp]{\pm}=a\tilde{X}+b\tilde{Y}\). This is possible if and only if
\begin{equation}
  \label{eq:compatib-parab}
  -b=\lambda a,\quad 0=\lambda b.
\end{equation}
The compatibility condition \(\lambda^2=0\) implies \(\lambda=0\)
within complex numbers. However, such a ``ladder'' operator
\(\ladder[\rmp]{\pm}=a\tilde{X}\) produces
only the zero shift on the eigenvectors, cf.~\eqref{eq:ladder-action}.

Another possibility appears if we consider the representation of the
Heisenberg group induced by dual-valued characters. On the
configuration space%
\index{configuration!space}%
\index{space!configuration} such a representation
is~\eqref{eq:schroedinger-rep-conf-der-par}: 
\begin{equation}
  \label{eq:schroedinger-rep-conf-par}
  [\uir{\rmp}{\chi}(s',x',y') \hat{f}](q)=  \rme^{2\pi(-\rmp\hbar (s'-{\textstyle\frac{1}{2}}x'y')+\rmi x' q)}
                                           \hat{f}(q+\rmp\rmi\hbar
                                           y').
\end{equation}
The corresponding derived representation of \(\algebra{h}\) is 
\begin{displaymath}
  \uir{\rmp}{\myh}(X)=2\pi\rmi q,\qquad
  \uir{\rmp}{\myh}(Y)={-\rmi\rmp\hbar} \frac{\rmd}{\rmd q},
  \qquad
  \uir{\rmp}{\myh}(S)=-2\pi\rmp\hbar I.
\end{displaymath}
However, the Shale--Weil extension generated by this representation is
inconvenient.  It is better to consider the FSB--type parabolic
representation~\eqref{eq:dual-repres}%
\index{Fock--Segal--Bargmann!representation}%
\index{representation!Fock--Segal--Bargmann} on the phase space
induced by the same dual-valued character.
Then the derived representation of \(\algebra{h}\) is:
\begin{equation}
  \label{eq:schroedinger-rep-conf-der-par1}
  \uir{p}{\myh}(X)=-2\pi\rmi q+\frac{\rmi\rmp\hbar}{2}\partial_{p},\qquad
  \uir{p}{\myh}(Y)=-2\pi\rmi p-\frac{\rmi\rmp\hbar}{2}\partial_{q},
  \qquad
  \uir{p}{\myh}(S)=2\pi\rmp\hbar I.
\end{equation}
An advantage of the FSB representation%
  \index{Fock--Segal--Bargmann!representation}%
  \index{representation!Fock--Segal--Bargmann} is that the
derived form of the parabolic Shale--Weil representation coincides
with the elliptic one~\eqref{eq:shale-weil-der-ell}.

Eigenfunctions with the eigenvalue \(\mu\) of the parabolic
Hamiltonian \(\tilde{B}+\tilde{Z}/2=q\partial_p\) have the form
\begin{equation}
  \label{eq:par-eigenfunctions}
  v_\mu (q,p)= \rme^{\mu p/q} f(q), \quad \text{ with an arbitrary
    function}\quad f(q).
\end{equation}

The linear equations defining the corresponding ladder operator
\(\ladder[\rmp]{\pm}=a\tilde{X}+b\tilde{Y}\) in the algebra
\(\algebra{h}\) are~\eqref{eq:compatib-parab}.  The compatibility
condition \(\lambda^2=0\) implies \(\lambda=0\) within complex numbers
again. Admitting dual numbers we have additional values
\(\lambda=\pm\rmp\lambda_1\) with \(\lambda_1\in\Space{C}{}\) with the
corresponding ladder operators
\begin{displaymath}
  \ladder[\rmp]{\pm}=\tilde{X}\mp\rmp\lambda_1\tilde{Y}=
  -2\pi\rmi q+\frac{\rmi\rmp\hbar}{2}\partial_{p}\pm 2\pi\rmi\rmp
  \lambda_1 p= 
  -2\pi\rmi q+   \rmp\rmi( \pm 2\pi\lambda_1 p+\frac{\hbar}{2}\partial_{p}).
\end{displaymath}
For the eigenvalue \(\mu=\mu_0+\rmp\mu_1\) with \(\mu_0\),
\(\mu_1\in\Space{C}{}\) the
eigenfunction~\eqref{eq:par-eigenfunctions} can be rewritten as:
\begin{equation}
  \label{eq:par-eigenfunctions-1}
  v_\mu (q,p)= \rme^{\mu  p/q} f(q)=  \rme^{\mu_0  p/q}\left(1+\rmp\mu_1
    \frac{p}{q}\right) f(q)
\end{equation}
due to the nilpotency of \(\rmp\).  Then the ladder action of
\(\ladder[\rmp]{\pm}\) is \(\mu_0+\rmp\mu_1\mapsto \mu_0+\rmp(\mu_1\pm
\lambda_1)\).  Therefore these operators are suitable for building
\(\algebra{sl}_2\)-modules with a one-dimensional chain of
eigenvalues.

Finally, consider the ladder operator for the same element \(B+Z/2\)
within the Lie algebra \(\algebra{sl}_2\), cf.
\S~\ref{sec:parab-ladd-oper}.  There is the only operator
\(\ladder[p]{\pm}=\tilde{B}+\tilde{Z}/2\) corresponding to complex
coefficients, which does not affect the eigenvalues.  However the dual
numbers lead to the operators
\begin{displaymath}
  \ladder[\rmp]{\pm}=\pm \rmp\lambda_2\tilde{A}+\tilde{B}+\tilde{Z}/2
  =
  \pm\frac{\rmp\lambda_2}{2}\left(q\partial_{q}-p\partial_{p}\right)+q\partial_{p}, 
  \qquad \lambda_2\in\Space{C}{}. 
\end{displaymath}
These operator act on eigenvalues in a non-trivial way.%
\index{dual!number|)}%
\index{number!dual|)}

\section{Wavelet Transform, Uncertainty Relation and Analyticity}
\label{sec:gaussian}

There are two and a half main examples of reproducing kernel spaces of
analytic function. One is the Fock--Segal--Bargmann (FSB) space and
others (one and a half)\,--\,the Bergman and Hardy spaces on the upper
half-plane. The first
space is generated by the Heisenberg
group~\citelist{\cite{Kisil11c}*{\S~7.3} \cite{Folland89}*{\S~1.6}},
two others\,--\,by the group \(\SL\)~\cite{Kisil11c}*{\S~4.2} (this
explains our way of counting). 

Those spaces have the following properties, which make their study
particularly pleasant and fruitful:
\begin{enumerate}
\item \label{item:group}
  There is a group, which acts transitively on functions'
  domain.
\item \label{item:reproducing-kernel}
  There is a reproducing kernel.
\item \label{item:spac-cons-holom}
  The space consists of holomorphic functions.
\end{enumerate}
Furthermore, for FSB space there is the following property:
\begin{enumerate}
\item[iv.] \label{item:uncertainity}
  The reproducing kernel is generated by a function, which minimises
  the uncertainty for coordinate and momentum observables.
\end{enumerate}
It is known, that a transformation group is responsible for the
appearance of the reproducing
kernel~\cite{AliAntGaz00}*{Thm.~8.1.3}. This paper shows that the last
two properties are equivalent and connected to the group as well.

\subsection{Induced Wavelet (Covariant) Transform}
\label{sec:induc-wavel-transf}

The following object is common in quantum
mechanics~\cite{Kisil02e}, signal processing, harmonic
analysis~\cite{Kisil12d}, operator theory~\cites{Kisil12b,Kisil13a}
and many other areas~\cite{Kisil11c}.  Therefore, it has various
names~\cite{AliAntGaz00}: coherent states, wavelets, matrix
coefficients, etc.  In the most fundamental
situation~\cite{AliAntGaz00}*{Ch.~8}, we start from an irreducible
unitary representation \(\uir{}{}\) of a Lie group \(G\) in a Hilbert
space \(\FSpace{H}{}\). For a vector \(f\in \FSpace{H}{}\) (called mother wavelet, vacuum
state, etc.), we define the map \(\oper{W}_f\) from \(\FSpace{H}{}\) to a space
of functions on \(G\) by:
\begin{equation}
  \label{eq:wavelet-trans}
  [\oper{W}_f v](g)=\tilde{v}(g):=\scalar{v}{\uir{}{}(g)f}.
\end{equation}
Under the above assumptions, \(\tilde{v}(g)\) is a bounded continuous
function on \(G\). The map \(\oper{W}_f\) intertwines \(\uir{}{}(g)\)
with the left shifts on \(G\):
\begin{equation}
  \label{eq:left-shift-itertwine}
  \oper{W}_f \circ \uir{}{}(g) =\Lambda(g)\circ   \oper{W}_f ,\qquad
  \text{ where}\quad] \Lambda(g):\tilde{v}(g')\mapsto
  \tilde{v}(g^{-1}g').
\end{equation}
Thus, the image \(\oper{W}_f \FSpace{H}{}\)
is invariant under the left shifts on \(G\).
If \(\uir{}{}\)
is square integrable and \(f\)
is admissible~\cite{AliAntGaz00}*{\S~8.1}, then \(\tilde{v}(g)\)
is square-integrable with respect to the Haar measure on \(G\).
Moreover, it is a reproducing kernel Hilbert space and the kernel is
\(k(g)= [\oper{W}_f f](g)\).
At this point, none of admissible vectors has an advantage over
others.

It is common~\cite{Kisil11c}*{\S~5.1}, that there exists a closed subgroup
\(H\subset G\) and a respective \(f\in \FSpace{H}{}\) such that \(\uir{}{}(h)
f=\chi(h) f\) for some character \(\chi\) of \(H\). In this case, it
is enough to know values of \(\tilde{v}(\map{s}(x))\), for any
continuous section \(\map{s}\) from the homogeneous space \(X=G/H\) to
\(G\). The map \(v \mapsto \tilde{v}(x)=\tilde{v}(\map{s}(x))\)
intertwines \(\uir{}{}\) with the representation \(\uir{}{\chi}\) in a
certain function space on \(X\) induced by the character \(\chi\) of
\(H\)~\cite{Kirillov76}*{\S~13.2}. We call the map
\begin{equation}
  \label{eq:induce-wavelet-transform}
  \oper{W}_f: v \mapsto \tilde{v}(x) = \scalar{v}{\uir{}{}(\map{s}(x))f},\qquad
  \text{where} \quad x\in G/H
\end{equation}
the \emph{induced wavelet
  transform}~\cite{Kisil11c}*{\S~5.1}.

For example, if \(G=\Space{H}{}\), \(H=\{(s,0,0)\in\Space{H}{}:\
s\in\Space{R}{}\}\) and its character
\(\chi_\myhbar(s,0,0)= \rme^{2\pi\rmi\myhbar s}\), then any vector
\(f\in\FSpace{L}{2}(\Space{R}{})\) satisfies \(\uir{}{\myhbar}(s,0,0)
f=\chi_\myhbar(s) f\) for the
representation~\eqref{eq:H1-schroedinger-rep-q-space}. Thus, we still
do not have a reason to prefer any admissible vector to
others. 

\subsection{The Uncertainty Relation}
\label{sec:uncert-relat}

In quantum mechanics~\cite{Folland89}*{\S~1.1}, an observable (that
is, a self-adjoint operator on a Hilbert space \(\FSpace{H}{}\))
\(A\)
produces the expectation value \(\bar{A}\)
on a pure state (that is, a unit vector) \(\phi\in \FSpace{H}{}\)
by \(\bar{A}=\scalar{A\phi}{\phi}\).
Then, the dispersion is evaluated as follow:
\begin{equation}
  \label{eq:dispersion}
  \Delta_\phi^2(A)=\scalar{(A-\bar{A})^2\phi}{\phi}=
  \scalar{(A-\bar{A})\phi}{(A-\bar{A})\phi}=
  \norm{(A-\bar{A})\phi}^2.
\end{equation}
The next theorem links obstructions of exact simultaneous
measurements with non-commutativity of observables.
\begin{theorem}[The Uncertainty relation]
  \label{th:uncertainty}
  If \(A\) and \(B\) are self-adjoint operators on a Hilbert space \(\FSpace{H}{}\), then
  \begin{equation}
    \label{eq:uncertainty-abstract}
    \textstyle
    \norm{(A-a)u}\norm{(B-b)u}\geq \frac{1}{2} \modulus{\scalar{(AB-BA)u}{u}}\notingiq
  \end{equation}
  for any \(u\in \FSpace{H}{}\) from the domains of \(AB\) and \(BA\) and \(a\),
  \(b\in\Space{R}{}\). Equality holds precisely when \(u\) is a
  solution of \(((A-a)+\rmi r (B-b))u=0\) for some real \(r\).
\end{theorem}
\begin{proof} The proof is well-known~\cite{Folland89}*{\S~1.3}, but
  it is short, instructive and relevant for the following discussion,
  thus we include it in full. We start from simple algebraic
  transformations:
  \begin{eqnarray}
    \scalar{(AB-BA)u}{u}&=&\scalar{((A-a)(B-b)-(B-b)(A-a))u}{u}\nonumber \\
    &=&\scalar{(B-b)u}{(A-a)u}-\scalar{(A-a))u}{(B-b)u}\nonumber \\
    &=&2\rmi \Im\scalar{(B-b)u}{(A-a)u}
  \end{eqnarray}
  Then by the Cauchy--Schwartz inequality:
  \begin{displaymath}
    \label{eq:uncert-rel-proof}
    \textstyle\frac{1}{2}\scalar{(AB-BA)u}{u}\leq
    \modulus{\scalar{(B-b)u}{(A-a)u}}\leq\norm{(B-b)u}\norm {(A-a)u}.
  \end{displaymath}
  The equality holds if and only if \((B-b)u\) and \((A-a)u\) are
  proportional by a \emph{purely imaginary} scalar.
\end{proof}
The famous application of the above theorem is the following
fundamental relation in quantum
mechanics. We use~\cite{Kisil10a}*{(3.5)} the 
Schr\"odinger representation%
\index{Schr\"odinger!representation|indef}%
\index{representation!Schr\"odinger|indef}
\eqref{eq:schroedinger-rep-conf} of the Heisenberg group~\eqref{eq:H1-schroedinger-rep-q-space}: 
\begin{equation}
  \label{eq:H1-schroedinger-rep-q-space1}
  [\uir{}{\myhbar}(s',x',y') \hat{f}\,](q)= \rme^{-2\pi\rmi\myhbar (s'+x'y'/2)
    -2\pi\rmi x' q}\,\hat{f}(q+\myhbar y').  
\end{equation}
Elements of the Lie algebra \(\algebra{h}\), corresponding to the
infinitesimal generators \(X\) and \(Y\) of one-parameters subgroups
\((0,t/(2\pi),0)\) and \((0,0,t)\) in \(\Space{H}{}\), are
represented in~\eqref{eq:H1-schroedinger-rep-q-space1} by the
(unbounded) operators \(\tilde{M}\) and \(\tilde{D}\) on
\(\FSpace{L}{2}(\Space{R}{})\):
\begin{equation}
  \label{eq:lie-algebra-repres}
  \textstyle
  \tilde{M}=-\rmi q,\quad \tilde{D}=\myhbar\frac{\rmd}{\rmd q},\quad \text{with the
    commutator} \quad [\tilde{M},\tilde{D}]= \rmi\myhbar I.
\end{equation}
In the Schr\"odinger model of quantum mechanics,
\(f(q)\in\FSpace{L}{2}(\Space{R}{})\) is interpreted as a wave
function (a state) of a particle, with \(M=\rmi \tilde{M}\) and
\(\frac{1}{\rmi} \tilde{D}\) are the observables of its coordinate and
momentum.
\begin{corollary}[Heisenberg--Kennard uncertainty relation]
  For the coordinate \(M\) and momentum \(D\) observables we have the
  \emph{Heisenberg--Kennard uncertainty} relation:
  \begin{equation}
    \label{eq:heisenberg-uncertainty}
    \Delta_\phi(M) \cdot \Delta_\phi(D) \geq \frac{\myh}{2}.
  \end{equation}
  The equality holds if and only if
  \(\phi(q)= \rme^{-c q^2}\),
  \(c\in\Space[+]{R}{}\)
  is the Gaussian\index{Gaussian}---the vacuum state in the
  Schr\"odinger model.
\end{corollary}
\begin{proof}
  The relation follows from the commutator \([M,D]=\rmi\myhbar I\),
  which, in turn, is the representation of the Lie algebra
  \(\algebra{h}\) of the Heisenberg group. By
  Thm.~\ref{th:uncertainty}, the minimal uncertainty
  state in the Schr\"odinger representation is a solution of the
  differential equation: \((M-\rmi r D)\phi=0\) for some
  \(r\in\Space{R}{}\), or, explicitly:
  \begin{equation}
    \label{eq:gaussian-uncertainty-eq}
    (M-\rmi r D)\phi=-\rmi\left(q+r{\myhbar}\frac{\rmd}{\rmd q}\right)\phi(q)=0.
  \end{equation}
  The solution is the Gaussian \(\phi(q)= \rme^{-c q^2}\),
  \(c=\frac{1}{2r\myhbar}\). For \(c>0\), this function  is in
  the state space \(\FSpace{L}{2}(\Space{R}{})\). 
\end{proof}
It is common to say that the Gaussian \(\phi(q)= \rme^{-c q^2}\) represents
the ground state, which minimises the uncertainty of coordinate and
momentum.

\subsection{Right Shifts and Analyticity}

To discover some preferable mother wavelets, we use the following 
general result from~\cite{Kisil11c}*{\S~5}.  Let \(G\) be a locally
compact group and \(\uir{}{}\) be its representation in a Hilbert
space \(\FSpace{H}{}\). Let \([\oper{W}_f v](g)=\scalar{v}{\uir{}{}(g)f}\) be the
wavelet transform defined by a vacuum state \(f\in \FSpace{H}{}\).  Then, the
right shift \(R(g): [\oper{W}_f v](g')\mapsto [\oper{W}_f v](g'g)\) for
\(g\in G\) coincides with the wavelet transform
\([\oper{W}_{f_g} v](g')=\scalar{v}{\uir{}{}(g')f_g}\) defined by the vacuum
state \(f_g=\uir{}{}(g) f\).  In other words, the covariant transform
intertwines%
\index{intertwining operator}%
\index{operator!intertwining} right shifts on the group \(G\) with the
associated action \(\uir{}{}\) on vacuum states,
cf.~\eqref{eq:left-shift-itertwine}:
\begin{equation}
  \label{eq:wave-intertwines-right}
  R(g) \circ \oper{W}_f= \oper{W}_{\uir{}{}(g)f}.
\end{equation}
Although, the above observation is almost trivial, applications of the following
corollary are not.
\begin{corollary}[Analyticity of the wavelet transform,~\cite{Kisil11c}*{\S~5}] 
  \label{co:cauchy-riemann-integ}
  Let \(G\) be a group and \(dg\) be a measure on \(G\). Let
  \(\uir{}{}\) be a unitary representation of \(G\), which can be
  extended by integration to a vector space \(V\) of functions or
  distributions on \(G\).  Let a mother wavelet \(f\in \FSpace{H}{}\) satisfy the
  equation
  \begin{displaymath}
    \int_{G} a(g)\, \uir{}{}(g) f\,\rmd g=0,
  \end{displaymath}
  for a fixed distribution \(a(g) \in V\). Then  any wavelet transform
  \(\tilde{v}(g)=\scalar{v}{\uir{}{}(g)f}\) obeys the condition:
  \begin{equation}
     \label{eq:dirac-op}
     D\tilde{v}=0,\qquad \text{where} \quad D=\int_{G} \bar{a}(g)\, R(g) \,\rmd g\notingiq
  \end{equation}
  with \(R\) being the right regular representation of \(G\).
\end{corollary}

Some applications (including discrete one) produced by the \(ax+b\) group
can be found in~\cite{Kisil12d}*{\S~6}.  We turn to the Heisenberg
group now.
\begin{example}[Gaussian and FSB transform]
  \label{ex:gaussian-fsb}
  The Gaussian\index{Gaussian} \(\phi(x)= \rme^{- c q^2/2}\) is a null-solution of the
  operator \(\myhbar c M-\rmi D\). For the centre \(Z=\{(s,0,0):\
  s\in\Space{R}{}\}\subset \Space{H}{}\), we define the section
  \(\map{s}:\Space{H}{}/Z\rightarrow \Space{H}{}\) by
  \(\map{s}(x,y)=(0,x,y)\). Then, the corresponding induced wavelet
  transform~\eqref{eq:induce-wavelet-transform} is:
  \begin{equation}
    \label{eq:fsb-transform}
   \tilde{v}(x,y)=\scalar{v}{\uir{}{}(\map{s}(x,y))f}= \int_{\Space{R}{}} v(q)\,  \rme^{\pi\rmi\myhbar xy
    -2\pi\rmi x q}\, \rme^{-c(q+{\myhbar} y)^2/2}\,\rmd q.
  \end{equation}
  The transformation intertwines the Schr\"odinger and
  Fock--Segal--Bargmann representations%
  \index{Fock--Segal--Bargmann!representation}%
  \index{representation!Fock--Segal--Bargmann}%
  \index{Heisenberg!group!Fock--Segal--Bargmann representation}%
  \index{group!Heisenberg!Fock--Segal--Bargmann representation} The
  infinitesimal generators \(X\)
  and \(Y\)
  of one-parameters subgroups \((0,t/(2\pi),0)\)
  and \((0,0,t)\)
  are represented through the right shift in~\eqref{eq:H-n-group-law}
  by
  \begin{displaymath}
    \textstyle
    R_*(X)=-\frac{1}{4\pi}y\partial_s+\frac{1}{2\pi}\partial_x,\quad
    R_*(Y)=\frac{1}{2}x\partial_s+\partial_y.
  \end{displaymath}
  For the representation induced
  by the character \(\chi_\myhbar(s,0,0)= \rme^{2\pi\rmi\myhbar s}\) we
  have \(\partial_s= 2\pi\rmi\myhbar
  I\). Cor.~\ref{co:cauchy-riemann-integ} ensures that the operator
  \begin{equation}
    \label{eq:fsb-cauchy-riemann}
    \myhbar c\cdot R_*(X)+\rmi\cdot R_*(Y)=
   -\frac{\myhbar}{2}(2\pi x+ {\rmi\myhbar c}y) +\frac{\myhbar
      c}{2\pi}\partial_x +\rmi \partial_y
  \end{equation}
  annihilate any \(\tilde{v}(x,y)\)
  from~\eqref{eq:fsb-transform}.  The
  integral~\eqref{eq:fsb-transform} is known as Fock--Segal--Bargmann
  (FSB) transform%
  \index{Fock--Segal--Bargmann!transform}%
  \index{transform!Fock--Segal--Bargmann} and in the most common case
  the values \(\myhbar=1\)
  and \(c=2\pi\)
  are used. For these, operator~\eqref{eq:fsb-cauchy-riemann} becomes
  \(-\pi(x+\rmi y)+(\partial_x+\rmi\partial _y)=-\pi
  z+2\partial_{\bar{z}}\)
  with \(z=x+\rmi y\).
  Then the function
  \(V(z)= \rme^{\pi z \bar{z}/2}\,\tilde{v}(z)= \rme^{\pi(x^2+y^2)/2}\,
  \tilde{v}(x,y)\)
  satisfies the Cauchy--Riemann equation
  \(\partial_{\bar{z}} V(z)=0\).
\end{example}
This example
shows, that the Gaussian is a preferred vacuum state (as producing
analytic functions through FSB transform) exactly for the same reason
as being the minimal uncertainty state: the both are derived from the
identity \((\myhbar c M+\rmi D) \rme^{- c q^2/2}=0\).

\subsection{Uncertainty and Analyticity}
\label{sec:uncert-analyt}
The main result of this paper is a generalisation of the previous
observation, which bridges together Cor.~\ref{co:cauchy-riemann-integ} and
Thm.~\ref{th:uncertainty}. Let \(G\), \(H\), \(\uir{}{}\) and
\(\FSpace{H}{}\) be
as before. Assume, that the homogeneous space
\(X=G/H\) has a (quasi-)invariant measure
\(d\mu(x)\)~\cite{Kirillov76}*{\S~13.2}. Then, for a function (or a
suitable distribution) \(k\) on \(X\) we can define the integrated
representation:
\begin{equation}
  \label{eq:integrated-rep}
  \uir{}{}(k)=\int_X k(x)\uir{}{}(\map{s}(x))\,\rmd \mu(x)\notingiq
\end{equation}
which is (possibly, unbounded) operators on (possibly, dense subspace
of) \(\FSpace{H}{}\). It is a homomorphism of the convolution algebra
\(\FSpace{L}{1}(G,dg)\) to an algebra of bounded operators on \(\FSpace{H}{}\).
In particular, \(R(k)\) denotes the
integrated right shifts,  for \(H=\{e\}\).

\begin{theorem}[\cite{Kisil13c}]
  Let \(k_1\) and \(k_2\) be two distributions on
  \(X\) with the respective integrated representations
  \(\uir{}{}(k_1)\) and \(\uir{}{}(k_2)\).  The following are
  equivalent:
  \begin{enumerate}
  \item A vector \(f\in \FSpace{H}{}\) satisfies the identity
    \begin{displaymath}
      \Delta_f(\uir{}{}(k_1))\cdot
    \Delta_f(\uir{}{}(k_2))=\modulus{\scalar{[\uir{}{}(k_1),\uir{}{}(k_1)]
      f}{f}}.
    \end{displaymath}
  \item The image of the wavelet transform \(\oper{W}_f: v \mapsto
    \tilde {v}(g)=\scalar{v}{\uir{}{}(g)f}\) consists of functions
    satisfying the equation \(R(k_1+\rmi r k_2) \tilde {v}=0\) for
    some \(r\in\Space{R}{}\), where \(R\) is the integrated
    form~\eqref{eq:integrated-rep} of the right regular representation
    on \(G\).
  \end{enumerate}
\end{theorem}
\begin{proof}
  This is an immediate consequence of a combination of
  Thm.~\ref{th:uncertainty} and Cor.~\ref{co:cauchy-riemann-integ}.
\end{proof}
Example~\ref{ex:gaussian-fsb} is a particular case of this theorem
with \(k_1(x,y)=\delta'_x(x,y)\) and \(k_2(x,y)=\delta'_y(x,y)\)
(partial derivatives of the delta function), which represent vectors
\(X\) and \(Y\) from the Lie algebra \(\algebra{h}\). The next
example will be of this type as well.

\subsection{Hardy Space on the Real Line}
\label{sec:hardy-space}

We consider the induced representation
\(\uir{}{1}\)~\eqref{eq:discrete} for \(k=1\) of the group
\(\SL\). A \(\SL\)-quasi-invariant measure on the real line is
\(\modulus{cx+d}^{-2}\,\rmd x\). Thus, the following form of the representation~\eqref{eq:discrete}
\begin{equation}
  \label{eq:discrete1}
  \uir{}{1}{}(g) f(w)=\frac{1}{cx+d}\,f\left(\frac{ax+b}{cx+d}\right),
  \quad \text{ where}\quad g^{-1}=\begin{pmatrix}a&b\\c&d
  \end{pmatrix}\notingiq
\end{equation}
is unitary in
\(\FSpace{L}{2}(\Space{R}{})\) with the Lebesgue measure \(dx\).  

We can calculate the derived representations for the basis of
\(\algebra{sl}_2\)~\eqref{eq:sl2-basis}:
\begin{align*}
  \rmd\uir{A}{1}  &= \textstyle \frac{1}{2}\cdot I+x \partial_x,\\
  \rmd\uir{B}{1}  &= \textstyle \frac{1}{2} x \cdot I+\frac{1}{2}(x^2-1)\partial_x,\\
  \rmd\uir{Z}{1}  &= - x\cdot I-(x^2+1)\partial_x.
\end{align*}

The linear combination of the above vector
fields producing \emph{ladder}  operators \( \ladder{\pm}=\pm\rmi A+B\) are,
cf.~\eqref{eq:elliptic-ladder}:
\begin{eqnarray}
  \label{eq:A+iB-e-rl}
 \rmd\uir{\ladder{\pm}}{1 } & = & \textstyle \frac{1}{2} \left( (x\pm\rmi)\cdot I
   +(x\pm \rmi)^2\cdot\partial _x \right).
\end{eqnarray}
Obviously, the function
\(f_+(x)=(x+\rmi)^{-1}\) satisfies \(\rmd\uir{\ladder{+}}{1} f_+=0\). Recalling
the commutator \([A,B]=-\frac{1}{2}Z\) we note that
\(\rmd\uir{Z}{1}f_+=-\rmi f_+\). Therefore, there is the
following identity for dispersions on this state:
\begin{displaymath}
  \textstyle
  \Delta_{f_+}(\uir{A}{1})\cdot \Delta_{f_+}(\uir{B}{1}) =\frac{1}{2}\notingiq
\end{displaymath}
with the minimal value of uncertainty among all eigenvectors of the operator
\(\rmd\uir{Z}{1}\).  

Furthermore, the vacuum state \(f_+\) generates the induced wavelet
transform for the subgroup \(K=\{ \rme^{tZ} \such t\in\Space{R}{}\}\).  We
identify \(\SL/K\) with the upper half-plane%
\index{half-plane!upper}%
\index{upper half-plane} \(\Space[+]{C}{}=\{z\in\Space{C}{} \such \Im
z > 0\}\)~\citelist{\cite{Kisil11c}*{\S~5.5} \cite{Kisil13a}}.  The map
\(\map{s}: \SL/K \rightarrow \SL\) is defined as
\(\map{s}(x+\rmi y)=\frac{1}{\sqrt{y}}
\begin{pmatrix}
  y&x\\
  0&1
\end{pmatrix}\)~\eqref{eq:s-map}. Then, the induced wavelet transform~\eqref{eq:induce-wavelet-transform} is:
\begin{align*}
  \tilde{v}(x+\rmi y)=\scalar{v}{\uir{}{1}(\map{s}(x+\rmi y)) f_+}&=
                                                             \frac{1}{2\pi\sqrt{y}}\int_{\Space{R}{}}
                                                             \frac{v(t)\,\rmd t}{\frac{t-x}{y}-\rmi}\\
                                                           &=
                                                             \frac{\sqrt{y}}{2\pi}\int_{\Space{R}{}}
                                                             \frac{v(t)\,\rmd t}{t-(x+\rmi
                                                             y)}. 
\end{align*}
Clearly, this is the Cauchy integral up to the factor
\(\sqrt{y}\), which is related to the  conformal
metric on the upper half-plane.  Similarly, we can consider the operator
\(\uir{B-\rmi A}{1}=\frac{1}{2} \left( (x\pm\rmi)\cdot I
   +(x\pm \rmi)^2\cdot\partial _x \right)\) and the function
\(f_-(z)=\frac{1}{x-\rmi}\) simultaneously solving the equations
\(\uir{B-\rmi A}{1} f_-=0\) and \(\rmd\uir{Z}{1}f_-=\rmi
f_-\). It produces the integral with the conjugated Cauchy kernel.

Finally, we can calculate the operator~\eqref{eq:dirac-op}
annihilating the image of the wavelet transform. In the coordinates
\((x+\rmi y,t)\in (\SL/K)\times K\), the restriction to the induced
subrepresentation is, cf.~\cite{Lang85}*{\S~IX.5}:
\begin{eqnarray}
  \label{eq:left-A-e-uhp}
  \linv{A}&=&
   \textstyle \frac{\rmi}{2}\sin 2t\cdot I - y \sin 2 t \cdot\partial_x- y \cos 2 t \cdot\partial_y,\\
  \linv{B}&=&
 \textstyle -\frac{\rmi }{2}\cos 2t\cdot  I+y \cos 2 t \cdot\partial_x- y \sin 2 t \cdot\partial_y.
\end{eqnarray}
Then, the left-invariant vector field corresponding to the ladder
operator contains the Cauchu--Riemann operator as the main ingredient: 
\begin{equation}
  \label{eq:CR-elliptic}
  \linv{\ladder{-}} =   \rme^{2\rmi t}(\textstyle -\frac{\rmi}{2} I+ y
  (\partial_x+\rmi\partial_y)),\qquad \text{ where}\quad \ladder{-}=\overline{\ladder{+}}=-\rmi A + B.
\end{equation}

Furthermore, if \(\linv{-\rmi A+B}{}\tilde{v}(x+\rmi y)=0\), then
\((\partial_{x}+\rmi\partial_y)(\tilde{v}(x+\rmi y)/\sqrt{y})=0\). That is,
\(V(x+\rmi y)=\tilde{v}(x+\rmi y)/\sqrt{y}\) is a holomorphic
function on the upper half-plane.

Similarly, we can treat representations of \(\SL\) in the space of
square integrable functions on the upper half-plane. The irreducible
components of this representation are isometrically
isomorphic~\cite{Kisil11c}*{\S~4--5} to the weighted Bergman spaces of
(purely poly-)analytic functions on the unit disk,
cf.~\cite{Vasilevski99}. Further connections between analytic function
theory and group representations can be found in~\cites{Kisil97c,Kisil11c}.

\subsection{Contravariant Transform and Relative Convolutions}
\label{sec:bound-oper}
For a square integrable unitary irreducible representation
\(\uir{}{}\) and a fixed admissible vector \(\psi\in V\),
the integrated representation~\eqref{eq:integrated-rep} produces the
\emph{contravariant transform}%
\index{contravariant transform}%
\index{transform!contravariant}
\(\oper{M}_\psi: \FSpace{L}{1}(G) \rightarrow V\),
cf.~\cite{Kisil98a,Kisil13a}:
\begin{equation}
  \label{eq:contravariant-transform}
  \oper{M}_\psi^{\uir{}{}} (k)= \uir{}{}(k)\psi, \qquad
  \text{ where}\quad k\in \FSpace{L}{1}(G).
\end{equation}

The contravariant transform \(\oper{M}_{\psi}^{\uir{}{}}\) intertwines the
left regular representation \( \Lambda  \)
on \( \FSpace{L}{2}(G)\) and \( \uir{}{} \):
\begin{equation}
  \label{eq:inv-transform-intertwine}
  \oper{M}_{\psi}^{\uir{}{}}\, \Lambda (g) = \uir{}{}(g)\, \oper{M}_{\psi}^{\uir{}{}}.
\end{equation}
Combining with~\eqref{eq:left-shift-itertwine}, we see that the
composition \(\oper{M}_\psi^{\uir{}{}} \circ
\oper{W}_\phi^{\uir{}{}}\) of the covariant and contravariant
transform intertwines \(\uir{}{}\) with itself.  For an irreducible
square integrable \(\uir{}{}\) and suitably normalised admissible
\(\phi\) and \(\psi\), we use the Schur's lemma
\cite[Lem.~4.3.1]{AliAntGaz00}, \cite[Thm.~8.2.1]{Kirillov76} to
conclude that:
\begin{equation}
  \label{eq:wave-trans-inverse}
  \oper{M}_\psi^{\uir{}{}} \circ
  \oper{W}_\phi^{\uir{}{}}=\scalar{\psi}{\phi} I.
\end{equation}  

Similarly to induced wavelet
transform~\eqref{eq:induce-wavelet-transform}, we may define
integrated representation and contravariant transform for a
homogeneous space. Let \(H\)
be a subgroup of \(G\)
and \(X=G/H\)
be the respective homogeneous space with a (quasi-)\-invariant measure
\(dx\)~\cite[\S~9.1]{Kirillov76}.
For the natural projection \(\map{p}: G \rightarrow X\)
we fix a continuous section
\(\map{s}: X \rightarrow G\)~\cite[\S~13.2]{Kirillov76},
which is a right inverse to \(\map{p}\).
Then, we define an operator of \emph{relative convolution}%
\index{relative convolution}%
\index{convolution!relative} on \(V\)~\cite{Kisil94e,Kisil13a},
cf.~\eqref{eq:integrated-rep}:
\begin{equation}
  \label{eq:relative-conv}
  \uir{}{}(k)=\int_{X} k(x)\,\uir{}{}(\map{s}(x))\,\rmd x\notingiq
\end{equation}
with a kernel \(k\) defined on \(X=G/H\).  
There are many important classes of operators described
by~\eqref{eq:relative-conv}, notably pseudodifferential operators
(PDO) and Toeplitz
operators~\cite{Howe80b,Kisil94e,Kisil98a,Kisil13a}.  Thus, it is
important to have various norm estimations of \(\uir{}{}(k)\). We
already mentioned a straightforward inequality
\(\norm{\uir{}{}(k)}\leq C \norm[1]{k}\) for
\(k\in\FSpace{L}{1}(G,dg)\), however, other classes are of interest as
well.

\subsection{Norm Estimations of Relative Convolutions}
\label{sec:cald-vaill-type}

If \(G\)
is the Heisenberg group and \(\uir{}{}\)
is its Schr\"odinger representation, then \(\uir{}{}(\hat{a})\)~\eqref{eq:relative-conv}
is a PDO \(a(X,D)\)
with the symbol \(a\)~\cite{Howe80b,Folland89,Kisil13a},
which is the Weyl quantization~\eqref{eq:weyl-quantisation} of a classical
observable \(a\) defined on phase space \(\Space{R}{2}\).
Here, \(\hat{a}\)
is the Fourier transform of \(a\),
as usual. The Calder\'on--Vaillancourt
theorem~\cite[Ch.~XIII]{MTaylor81} estimates \(\norm{a(X,D)}\)
by \(\FSpace{L}{\infty}\)-norm
of a finite number of partial derivatives of \(a\).

In this section we revise the method used in~\cite[\S~3.1]{Howe80b} to
prove the Calder\'on--Vaillancourt estimations. It was described as ``rather
magical'' in~\cite[\S~2.5]{Folland89}. We hope, that  a usage of the covariant
transform dispel the mystery without undermining the power of the method. 
\vskip 0.2cm

We start from the following lemma, which has a transparent proof in
terms of covariant transform, cf. ~\cite[\S~3.1]{Howe80b}
and~\cite[(2.75)]{Folland89}. For the rest of the section we assume that
\(\uir{}{}\) is an irreducible square integrable representation of an
exponential Lie group \(G\) in \(V\) and mother wavelet \(\phi,\psi
\in V\) are admissible.
\begin{lemma}
  Let \(\phi\in V\) be such that, for \(\Phi=\oper{W}_\phi \phi\), the
  reciprocal \(\Phi^{-1}\) is bounded on \(G\) or \(X=G/H\).  Then,
  for the integrated representation~\eqref{eq:integrated-rep} or
  relative convolution~\eqref{eq:relative-conv}, we have the
  inequality:
  \begin{equation}
    \label{eq:norm-inequal-repres}
    \norm{\uir{}{}(f)}\leq \norm{\Lambda \otimes R(f \Phi^{-1})}\notingiq
  \end{equation}
  where \((\Lambda \otimes R)(g): k(g')\mapsto k(g^{-1}g'g)\)
  acts on the image of \(\oper{W}_\phi\).
\end{lemma}
\begin{proof}
  We know from~\eqref{eq:wave-trans-inverse} that \(\oper{M}_\phi
  \circ \oper{W}_{\uir{}{}(g)\phi} =
  \scalar{\phi}{\uir{}{}(g)\phi} I\) on \(V\), thus:
  \begin{displaymath}
    \oper{M}_\phi \circ \oper{W}_{\uir{}{}(g)\phi} \circ \uir{}{}(g) =
    \scalar{\phi}{\uir{}{}(g)\phi} \uir{}{}(g) =\Phi (g) \uir{}{}(g). 
  \end{displaymath}
  On the other hand, the intertwining
  properties~\eqref{eq:left-shift-itertwine}
  and~\eqref{eq:wave-intertwines-right} of the wavelet transform 
  imply:
  \begin{displaymath}
    \oper{M}_\phi \circ \oper{W}_{\uir{}{}(g)\phi} \circ \uir{}{}(g) =
    \oper{M}_\phi \circ (\Lambda \otimes R)(g)  \circ\oper{W}_{\phi}.
  \end{displaymath}
  Integrating the identity \(\Phi (g) \uir{}{}(g)=  \oper{M}_\phi
  \circ (\Lambda \otimes R)(g)  \circ\oper{W}_{\phi}\) with the
  function \(f\Phi^{-1}\) and use the partial isometries \(\oper{W}_\phi\)  and
  \(\oper{M}_\phi\) we get the inequality.
\end{proof}

The Lemma is most efficient if \(\Lambda \otimes R\) acts in a simple
way. Thus, we give he following
\begin{definition}
  We say that the subgroup \(H\) has the \emph{complemented commutator
    property}, if there exists a continuous section \(\map{s}: X \rightarrow
  G\) such that:
  \begin{equation}
    \label{eq:H-upper-triangle-prop}
    \map{p}(\map{s}(x)^{-1}g \map{s}(x))=\map{p}(g),\qquad \text{ for
      all}\quad x\in X=G/H, \ g\in G.
  \end{equation}
\end{definition}
For a Lie group \(G\) with the Lie algebra \(\mathfrak{g}\) define the
Lie algebra \(\mathfrak{h}=[\mathfrak{g},\mathfrak{g}]\). The subgroup
\(H=\exp(\mathfrak{h})\) (as well as any larger subgroup) has the
complemented commutator property~\eqref{eq:H-upper-triangle-prop}. Of
course, \(X=G/H\) is non-trivial if \(H\neq G\) and this happens, for
example, for a nilpotent \(G\).  In particular, for the Heisenberg
group, its centre has the complemented commutator property.

Note, that the complemented commutator
property~\eqref{eq:H-upper-triangle-prop} implies:
\begin{equation}
  \label{eq:h-identity-defn}
  \Lambda\otimes R(\map{s}(x)): g \mapsto gh, \quad \text{ for the
  unique}\quad h=g^{-1}\map{s}(x)^{-1}g \map{s}(x)\in H.
\end{equation}
For a character \(\chi\) of the subgroup \(H\), we introduce an integral
transformation \(\wideparen{\ }:{L}_{1}(X)\rightarrow {C}(G)\):
\begin{equation}
  \label{eq:paren-transf}
  \wideparen{k}(g)=\int_X k(x)\, \chi(g^{-1}\map{s}(x)^{-1}g \map{s}(x))\,\rmd x\notingiq
\end{equation}
where \(h(x,g)=g^{-1}\map{s}(x)^{-1}g \map{s}(x)\) is in \(H\) due to
the relations~\eqref{eq:H-upper-triangle-prop}. This transformation
generalises the isotropic symbol defined for the Heisenberg group in
\cite[\S~2.1]{Howe80b}.
\begin{proposition}[\cite{Kisil13b}]
  Let a subgroup \(H\) of \(G\) have the complemented commutator
  property~\eqref{eq:H-upper-triangle-prop} and \(\uir{}{\chi}\) be
  an irreducible representation of \(G\) induced from a character
  \(\chi\) of \(H\), then
  \begin{equation}
    \label{eq:norm-ineq-Hn-operat}
    \norm{\uir{}{\chi}(f)}\leq 
    \norm[\infty]{\wideparen{f\Phi^{-1}}}\notingiq
  \end{equation}
  with the \(\sup\)-norm of the function \(\wideparen{f\Phi^{-1}}\)
  on the right.
\end{proposition}
\begin{proof}
  For an induced representation \(\uir{}{\chi}\)~\cite[\S~13.2]{Kirillov76}, the covariant
  transform \(\oper{W}_\phi\) maps  \(V\) to a space
  \(\FSpace[\chi]{L}{2}(G)\) of functions having the property
  \(F(gh)=\chi(h)F(g)\)~\cite[\S~3.1]{Kisil13a}. 
  From~\eqref{eq:h-identity-defn}, the restriction of \(\Lambda\otimes
  R\) to the space \(\FSpace[\chi]{L}{2}(G)\) is:
  \begin{displaymath}
        \Lambda\otimes R(\map{s}(x)): \psi(g) \mapsto \psi(gh)=\chi(h(x,g))\psi(g).
  \end{displaymath}
  In other words, \(\Lambda\otimes R\) acts by multiplication on
  \(\FSpace[\chi]{L}{2}(G)\). 
  Then, integrating the representation \(\Lambda \otimes R\) over \(X\)
  with a function \(k\) we get an operator \((L\otimes R)(k)\),
  which reduces on the irreducible component to multiplication by the
  function \(\wideparen{k}(g)\).
  Put \(k=f\Phi^{-1}\) for \(\Phi=\oper{W}_\phi \phi\). Then, from
  the inequality~\eqref{eq:norm-inequal-repres}, the norm of operator
  \(\uir{}{\chi}(f)\) can be estimated by \(\norm{\Lambda \otimes
    R(f \Phi^{-1})}=\norm[\infty]{\wideparen{f\Phi^{-1}}}\).
\end{proof}

For a nilpotent step \(2\) Lie group, the
transformation~\eqref{eq:paren-transf} is almost the Fourier
transform, cf. the case of the Heisenberg group
in~\cite[\S~2.1]{Howe80b}. This allows to estimate
\(\norm[\infty]{\wideparen{f\Phi^{-1}}}\) through
\(\norm[\infty]{\wideparen{f}}\), where \(\wideparen{f}\) is in the
essence the symbol of the respective PDO. For other groups, the
expression \(g^{-1}\map{s}(x)^{-1}g \map{s}(x)\)
in~\eqref{eq:paren-transf} contains non-linear terms and its analysis
is more difficult. In some circumstance the integral Fourier
operators~\cite[Ch.~VIII]{MTaylor81} may be useful for this purpose.

\section*{Acknowledgement}
Author is grateful to Prof.~Ivailo Mladenov for the kind invitation and
support.


\begin{bibdiv}
\begin{biblist}

\bib{AgostiniCapraraCiccotti07a}{article}{
      author={Agostini, F.},
      author={Caprara, S.},
      author={Ciccotti, G.},
       title={Do we have a consistent non-adiabatic quantum-classical
  mechanics?},
        date={2007},
        ISSN={0295-5075},
     journal={Europhys. Lett. EPL},
      volume={78},
      number={3},
       pages={Art. 30001, 6},
        note={\doi{10.1209/0295-5075/78/30001}},
      review={\MR{MR2366698 (2008k:81004)}},
}

\bib{AliAntGaz00}{book}{
      author={Ali, Syed~Twareque},
      author={Antoine, Jean-Pierre},
      author={Gazeau, Jean-Pierre},
       title={Coherent states, wavelets and their generalizations},
      series={Graduate Texts in Contemporary Physics},
   publisher={Springer-Verlag},
     address={New York},
        date={2000},
        ISBN={0-387-98908-0},
      review={\MR{2002m:81092}},
}

\bib{Arnold91}{book}{
      author={Arnol{\cprime}d, V.~I.},
       title={Mathematical methods of classical mechanics},
      series={Graduate Texts in Mathematics},
   publisher={Springer-Verlag},
     address={New York},
        date={1991},
      volume={60},
        ISBN={0-387-96890-3},
        note={Translated from the 1974 Russian original by K. Vogtmann and A.
  Weinstein, corrected reprint of the second (1989) edition},
      review={\MR{96c:70001}},
}

\bib{ArovDym08}{book}{
      author={Arov, Damir~Z.},
      author={Dym, Harry},
       title={{$J$}-contractive matrix valued functions and related topics},
      series={Encyclopedia of Mathematics and its Applications},
   publisher={Cambridge University Press},
     address={Cambridge},
        date={2008},
      volume={116},
        ISBN={978-0-521-88300-9},
      review={\MR{MR2474532}},
}

\bib{Bell08a}{book}{
      author={Bell, John~L.},
       title={A primer of infinitesimal analysis},
     edition={Second},
   publisher={Cambridge University Press},
     address={Cambridge},
        date={2008},
        ISBN={978-0-521-88718-2; 0-521-88718-6},
         url={http://dx.doi.org/10.1017/CBO9780511619625},
      review={\MR{2398446 (2009c:03075)}},
}

\bib{BocCatoniCannataNichZamp07}{book}{
      author={Boccaletti, Dino},
      author={Catoni, Francesco},
      author={Cannata, Roberto},
      author={Catoni, Vincenzo},
      author={Nichelatti, Enrico},
      author={Zampetti, Paolo},
       title={The mathematics of {Minkowski} space-time and an introduction to
  commutative hypercomplex numbers},
   publisher={Birkh\"auser Verlag},
     address={Basel},
        date={2008},
}

\bib{BoyerMiller74a}{article}{
      author={Boyer, Charles~P.},
      author={Miller, Willard, Jr.},
       title={A classification of second-order raising operators for
  {H}amiltonians in two variables},
        date={1974},
        ISSN={0022-2488},
     journal={J. Mathematical Phys.},
      volume={15},
       pages={1484\ndash 1489},
      review={\MR{0345542 (49 \#10278)}},
}

\bib{BrodlieKisil03a}{incollection}{
      author={Brodlie, Alastair},
      author={Kisil, Vladimir~V.},
       title={Observables and states in {$p$}-mechanics},
        date={2003},
   booktitle={Advances in mathematics research. vol. 5},
   publisher={Nova Sci. Publ.},
     address={Hauppauge, NY},
       pages={101\ndash 136},
        note={\arXiv{quant-ph/0304023}},
      review={\MR{MR2117375}},
}

\bib{CalzettaVerdaguer06a}{article}{
      author={Calzetta, Esteban},
      author={Verdaguer, Enric},
       title={Real-time approach to tunnelling in open quantum systems:
  Decoherence and anomalous diffusion},
        date={2006},
        ISSN={0305-4470},
     journal={J. Phys. A},
      volume={39},
      number={30},
       pages={9503\ndash 9532},
         url={http://dx.doi.org/10.1088/0305-4470/39/30/008},
      review={\MR{MR2246702 (2007f:82059)}},
}

\bib{CatoniCannataNichelatti04}{article}{
      author={Catoni, Francesco},
      author={Cannata, Roberto},
      author={Nichelatti, Enrico},
       title={The parabolic analytic functions and the derivative of real
  functions},
        date={2004},
     journal={Adv. Appl. Clifford Algebras},
      volume={14},
      number={2},
       pages={185\ndash 190},
}

\bib{CnopsKisil97a}{article}{
      author={Cnops, Jan},
      author={Kisil, Vladimir~V.},
       title={Monogenic functions and representations of nilpotent {L}ie groups
  in quantum mechanics},
        date={1999},
        ISSN={0170-4214},
     journal={Math. Methods Appl. Sci.},
      volume={22},
      number={4},
       pages={353\ndash 373},
        note={\arXiv{math/9806150}. \Zbl{1005.22003}},
      review={\MR{MR1671449 (2000b:81044)}},
}

\bib{ConstalesFaustinoKrausshar11a}{article}{
      author={Constales, Denis},
      author={Faustino, Nelson},
      author={Krau\ss{}har, Rolf~S\"oren},
       title={{F}ock spaces, {L}andau operators and the time-harmonic {M}axwell
  equations},
        date={2011},
     journal={Journal of Physics A: Mathematical and Theoretical},
      volume={44},
      number={13},
       pages={135303},
         url={http://stacks.iop.org/1751-8121/44/i=13/a=135303},
}

\bib{Cuntz01a}{article}{
      author={Cuntz, Joachim},
       title={Quantum spaces and their noncommutative topology},
        date={2001},
        ISSN={0002-9920},
     journal={Notices Amer. Math. Soc.},
      volume={48},
      number={8},
       pages={793\ndash 799},
      review={\MR{1847023 (2002g:58006)}},
}

\bib{Devis77}{book}{
      author={Davis, Martin},
       title={Applied nonstandard analysis},
   publisher={Wiley-Interscience [John Wiley \& Sons]},
     address={New York},
        date={1977},
        ISBN={0-471-19897-8},
      review={\MR{MR0505473 (58 \#21590)}},
}

\bib{deGosson08a}{article}{
      author={de~Gosson, Maurice~A.},
       title={Spectral properties of a class of generalized {L}andau
  operators},
        date={2008},
        ISSN={0360-5302},
     journal={Comm. Partial Differential Equations},
      volume={33},
      number={10--12},
       pages={2096\ndash 2104},
         url={http://dx.doi.org/10.1080/03605300802501434},
      review={\MR{MR2475331 (2010b:47128)}},
}

\bib{Dimentberg78a}{article}{
      author={Dimentberg, F.~M.},
       title={The method of screws and calculus of screws applied to the theory
  of three-dimensional mechanisms},
        date={1978},
        ISSN={0137-3722},
     journal={Adv. in Mech.},
      volume={1},
      number={3-4},
       pages={91\ndash 106},
      review={\MR{530595 (83h:70003)}},
}

\bib{Dimentberg78b}{book}{
      author={Dimentberg, F.~M.},
       title={{\cyr Teoriya vintov i ee prilozheniya} [{Theory} of screws and
  its applications]},
    language={Russian},
   publisher={``Nauka'', Moscow},
        date={1978},
      review={\MR{510196 (80b:70004)}},
}

\bib{Dirac26b}{article}{
      author={Dirac, P. A.~M.},
       title={On the theory of quantum mechanics},
        date={1926},
     journal={Proceedings of the Royal Society of London. Series A},
      volume={112},
      number={762},
       pages={661\ndash 677},
  eprint={http://rspa.royalsocietypublishing.org/content/112/762/661.full.pdf+html},
  url={http://rspa.royalsocietypublishing.org/content/112/762/661.short},
}

\bib{Dirac26a}{article}{
      author={Dirac, P. A.~M.},
       title={Quantum mechanics and a preliminary investigation of the hydrogen
  atom},
    language={English},
        date={1926},
        ISSN={09501207},
     journal={Proceedings of the Royal Society of London. Series A, Containing
  Papers of a Mathematical and Physical Character},
      volume={110},
      number={755},
       pages={561\ndash 579},
  eprint={http://rspa.royalsocietypublishing.org/content/110/755/561.full.pdf+html},
         url={http://www.jstor.org/stable/94410},
}

\bib{DiracPrinciplesQM}{book}{
      author={Dirac, P. A.~M.},
       title={The {P}rinciples of {Q}uantum {M}echanics},
     edition={4},
   publisher={Oxford University Press},
     address={London},
        date={1958},
      review={\MR{0023198 (9,319d)}},
}

\bib{DiracDirections}{book}{
      author={Dirac, Paul A.~M.},
       title={Directions in physics},
   publisher={Wiley-Interscience [John Wiley \& Sons]},
     address={New York},
        date={1978},
        ISBN={0-471-02997-1},
        note={Five lectures delivered during a visit to Australia and New
  Zealand, August--September, 1975, With a foreword by Mark Oliphant, Edited by
  H. Hora and J. R. Shepanski},
      review={\MR{0479067 (57 \#18520)}},
}

\bib{ErdelyiMagnusII}{book}{
      author={Erd{\'e}lyi, Arthur},
      author={Magnus, Wilhelm},
      author={Oberhettinger, Fritz},
      author={Tricomi, Francesco~G.},
       title={Higher transcendental functions. {V}ol. {II}},
   publisher={Robert E. Krieger Publishing Co. Inc.},
     address={Melbourne, Fla.},
        date={1981},
        ISBN={0-89874-069-X},
        note={Based on notes left by Harry Bateman, Reprint of the 1953
  original},
      review={\MR{698780 (84h:33001b)}},
}

\bib{FaddeevYakubovskii09}{book}{
      author={Faddeev, L.~D.},
      author={Yakubovski\u\i, O.~A.},
       title={Lectures on quantum mechanics for mathematics students},
    language={English},
      series={Student Mathematical Library},
   publisher={{American Mathematical Society (AMS)}},
     address={Providence, RI},
        date={2009},
      volume={47},
        note={Translated by Harold McFaden. xii, 234~p.},
}

\bib{Feynman1990qed}{book}{
      author={Feynman, R.P.},
       title={{QED}: the strange theory of light and matter},
      series={Penguin Press Science Series},
   publisher={Penguin},
        date={1990},
        ISBN={9780140125054},
         url={http://books.google.com/books?id=2X-3QgAACAAJ},
}

\bib{FeynHibbs65}{book}{
      author={Feynman, R.P.},
      author={Hibbs, A.R.},
       title={Quantum mechanics and path integral},
   publisher={McGraw-Hill Book Company},
     address={New York},
        date={{\noopsort{}}1965},
}

\bib{Folland89}{book}{
      author={Folland, Gerald~B.},
       title={Harmonic analysis in phase space},
      series={Annals of Mathematics Studies},
   publisher={Princeton University Press},
     address={Princeton, NJ},
        date={1989},
      volume={122},
        ISBN={0-691-08527-7; 0-691-08528-5},
      review={\MR{92k:22017}},
}

\bib{Gazeau09a}{book}{
      author={Gazeau, Jean-Pierre},
       title={{Coherent States in Quantum Physics}},
   publisher={Wiley-VCH Verlag},
        date={2009},
        ISBN={9783527407095},
}

\bib{Gromov90a}{book}{
      author={Gromov, N.~A.},
       title={{\cyr Kontraktsii i analiticheskie prodolzheniya klassicheskikh
  grupp. {E}dinyi podkhod}. ({R}ussian) [{C}ontractions and analytic extensions
  of classical groups. {U}nified approach]},
   publisher={Akad. Nauk SSSR Ural. Otdel. Komi Nauchn. Tsentr},
     address={Syktyvkar},
        date={1990},
      review={\MR{MR1092760 (91m:81078)}},
}

\bib{Gromov90b}{article}{
      author={Gromov, N.~A.},
       title={Transitions: Contractions and analytical continuations of the
  {C}ayley--{K}lein groups},
        date={1990},
        ISSN={0020-7748},
     journal={Int. J. Theor. Phys.},
      volume={29},
       pages={607\ndash 620},
         url={http://dx.doi.org/10.1007/BF00672035},
}

\bib{GromovKuratov05b}{article}{
      author={Gromov, N.~A.},
      author={Kuratov, V.~V.},
       title={All possible {C}ayley-{K}lein contractions of quantum orthogonal
  groups},
        date={2005},
     journal={Yadernaya Fiz.},
      volume={68},
      number={10},
       pages={1752\ndash 1762},
         url={http://dx.doi.org/10.1134/1.2121918},
      review={\MR{MR2189521 (2006g:81101)}},
}

\bib{Gromov12a}{book}{
      author={{Gromov}, N.A.},
       title={{\cyr Kontraktsii Klassicheskikh i Kvantovykh Grupp.}
  [contractions of classic and quanrum groups]},
    language={Russian},
   publisher={Moskva: Fizmatlit},
        date={2012},
        ISBN={978-5-9221-1398-4/hbk; 978-5-7691-2325-2/hbk},
}

\bib{GuentherKuzhel10a}{article}{
      author={G\"unther, Uwe},
      author={Kuzhel, Sergii},
       title={$\mathcal{{P}}\mathcal{{T}}$--symmetry, {C}artan decompositions,
  {L}ie triple systems and {K}rein space-related {C}lifford algebras},
        date={2010},
     journal={Journal of Physics A: Mathematical and Theoretical},
      volume={43},
      number={39},
       pages={392002},
         url={http://stacks.iop.org/1751-8121/43/i=39/a=392002},
}

\bib{HerranzOrtegaSantander99a}{article}{
      author={Herranz, Francisco~J.},
      author={Ortega, Ram{\'o}n},
      author={Santander, Mariano},
       title={Trigonometry of spacetimes: a new self-dual approach to a
  curvature/signature (in)dependent trigonometry},
        date={2000},
        ISSN={0305-4470},
     journal={J. Phys. A},
      volume={33},
      number={24},
       pages={4525\ndash 4551},
        note={\arXiv{math-ph/9910041}},
      review={\MR{MR1768742 (2001k:53099)}},
}

\bib{HerranzSantander02b}{article}{
      author={Herranz, Francisco~J.},
      author={Santander, Mariano},
       title={Conformal compactification of spacetimes},
        date={2002},
        ISSN={0305-4470},
     journal={J. Phys. A},
      volume={35},
      number={31},
       pages={6619\ndash 6629},
        note={\arXiv{math-ph/0110019}},
      review={\MR{MR1928852 (2004b:53123)}},
}

\bib{Howe80a}{article}{
      author={Howe, Roger},
       title={On the role of the {H}eisenberg group in harmonic analysis},
        date={1980},
        ISSN={0002-9904},
     journal={Bull. Amer. Math. Soc. (N.S.)},
      volume={3},
      number={2},
       pages={821\ndash 843},
      review={\MR{81h:22010}},
}

\bib{Howe80b}{article}{
      author={Howe, Roger},
       title={Quantum mechanics and partial differential equations},
        date={1980},
        ISSN={0022-1236},
     journal={J. Funct. Anal.},
      volume={38},
      number={2},
       pages={188\ndash 254},
      review={\MR{83b:35166}},
}

\bib{HoweTan92}{book}{
      author={Howe, Roger},
      author={Tan, Eng-Chye},
       title={Nonabelian harmonic analysis. {Applications of
  ${{\rm{S}}L}(2,{{\bf{R}}})$}},
   publisher={Springer-Verlag},
     address={New York},
        date={1992},
        ISBN={0-387-97768-6},
      review={\MR{1151617 (93f:22009)}},
}

\bib{Hudson66a}{thesis}{
      author={Hudson, Robin},
       title={Generalised translation-invariant mechanics},
        type={D. Phil. thesis},
     address={Bodleian Library, Oxford},
        date={1966},
}

\bib{Hudson04a}{incollection}{
      author={Hudson, Robin},
       title={Translation invariant phase space mechanics},
        date={2004},
   booktitle={Quantum theory: reconsideration of foundations---2},
      editor={Khrennikov, A.},
      series={Math. Model. Phys. Eng. Cogn. Sci.},
      volume={10},
   publisher={V\"axj\"o Univ. Press, V\"axj\"o},
       pages={301\ndash 314},
      review={\MR{2111131 (2006e:81134)}},
}

\bib{Khrennikov05a}{article}{
      author={Khrennikov, A.~Yu.},
       title={Hyperbolic quantum mechanics},
        date={2005},
        ISSN={0869-5652},
     journal={Dokl. Akad. Nauk},
      volume={402},
      number={2},
       pages={170\ndash 172},
      review={\MR{MR2162434 (2006d:81118)}},
}

\bib{Khrennikov01a}{article}{
      author={Khrennikov, Andrei},
       title={Linear representations of probabilistic transformations induced
  by context transitions},
        date={2001},
        ISSN={0305-4470},
     journal={J. Phys. A},
      volume={34},
      number={47},
       pages={9965\ndash 9981},
         url={http://dx.doi.org/10.1088/0305-4470/34/47/304},
        note={\arXiv{quant-ph/0106073}},
      review={\MR{1871755 (2002m:60007)}},
}

\bib{Khrennikov03a}{article}{
      author={Khrennikov, Andrei},
       title={{Hyperbolic Quantum Mechanics}},
    language={English},
        date={2003},
     journal={Adv. Appl. Clifford Algebras},
      volume={13},
      number={1},
       pages={1\ndash 9},
        note={\arXiv{quant-ph/0101002}},
}

\bib{Khrennikov08a}{article}{
      author={Khrennikov, Andrei},
       title={Hyperbolic quantization},
        date={2008},
        ISSN={0188-7009},
     journal={Adv. Appl. Clifford Algebr.},
      volume={18},
      number={3--4},
       pages={843\ndash 852},
      review={\MR{MR2490591}},
}

\bib{Khrennikov09book}{book}{
      author={Khrennikov, Andrei},
       title={Contextual approach to quantum formalism},
      series={Fundamental Theories of Physics},
   publisher={Springer},
     address={New York},
        date={2009},
      volume={160},
        ISBN={978-1-4020-9592-4},
         url={http://dx.doi.org/10.1007/978-1-4020-9593-1},
      review={\MR{2676217 (2011i:81012)}},
}

\bib{KhrennikovSegre07a}{incollection}{
      author={Khrennikov, Andrei},
      author={Segre, Gavriel},
       title={Hyperbolic quantization},
        date={2007},
   booktitle={Quantum probability and infinite dimensional analysis},
      editor={Accardi, L.},
      editor={Freudenberg, W.},
      editor={Sch\"urman, M.},
      series={QP--PQ: Quantum Probab. White Noise Anal.},
      volume={20},
   publisher={World Scientific Publishing, Hackensack, NJ},
       pages={282\ndash 287},
      review={\MR{MR2359402}},
}

\bib{KhrenVol01}{article}{
      author={Khrennikov, A.Yu.},
      author={Volovich, Ya.I.},
       title={Numerical experiment on interference for macroscopic particles},
        date={2001},
        note={\arXiv{quant-ph/0111159}},
}

\bib{Kirillov76}{book}{
      author={Kirillov, A.~A.},
       title={Elements of the theory of representations},
   publisher={Springer-Verlag},
     address={Berlin},
        date={1976},
        note={Translated from the Russian by Edwin Hewitt, Grundlehren der
  Mathematischen Wissenschaften, Band 220},
      review={\MR{54 \#447}},
}

\bib{Kirillov94a}{incollection}{
      author={Kirillov, A.~A.},
       title={Introduction to the theory of representations and noncommutative
  harmonic analysis [{\MR{90a:22005}}]},
        date={1994},
   booktitle={Representation theory and noncommutative harmonic analysis, i},
   publisher={Springer},
     address={Berlin},
       pages={1\ndash 156, 227\ndash 234},
        note={\MR{1311488}.},
      review={\MR{1 311 488}},
}

\bib{Kirillov99}{article}{
      author={Kirillov, A.~A.},
       title={Merits and demerits of the orbit method},
        date={1999},
        ISSN={0273-0979},
     journal={Bull. Amer. Math. Soc. (N.S.)},
      volume={36},
      number={4},
       pages={433\ndash 488},
      review={\MR{2000h:22001}},
}

\bib{Kisil08a}{article}{
      author={Kisil, Anastasia~V.},
       title={Isometric action of {${\rm SL}_2(\mathbb{R})$} on homogeneous
  spaces},
        date={2010},
     journal={Adv. App. Clifford Algebras},
      volume={20},
      number={2},
       pages={299\ndash 312},
        note={\arXiv{0810.0368}. \MR{2012b:32019}},
}

\bib{Kisil93c}{incollection}{
      author={Kisil, Vladimir~V.},
       title={{Clifford} valued convolution operator algebras on the
  {Heisenberg} group. {A} quantum field theory model},
        date={1993},
   booktitle={{C}lifford algebras and their applications in mathematical
  physics, proceedings of the {T}hird international conference held in
  {D}einze},
      editor={Brackx, F.},
      editor={Delanghe, R.},
      editor={Serras, H.},
      series={Fundamental Theories of Physics},
      volume={55},
   publisher={Kluwer Academic Publishers Group},
     address={Dordrecht},
       pages={287\ndash 294},
        note={\MR{1266878}},
}

\bib{Kisil94d}{incollection}{
      author={Kisil, Vladimir~V.},
       title={Quantum probabilities and non-commutative {F}ourier transform on
  the {H}eisenberg group},
        date={1995},
   booktitle={Interaction between functional analysis, harmonic analysis and
  probability ({C}olumbia, {MO}, 1994)},
      editor={Kalton, Nigel},
      editor={Saab, Elias},
      editor={Montgomery-Smith},
      series={Lecture Notes in Pure and Appl. Math.},
      volume={175},
   publisher={Dekker},
     address={New York},
       pages={255\ndash 266},
        note={\MR{97b:81060}},
}

\bib{Kisil96a}{article}{
      author={Kisil, Vladimir~V.},
       title={Plain mechanics: Classical and quantum},
        date={1996},
        ISSN={0963-2654},
     journal={J. Natur. Geom.},
      volume={9},
      number={1},
       pages={1\ndash 14},
        note={\arXiv{funct-an/9405002}},
      review={\MR{MR1374912 (96m:81112)}},
}

\bib{Kisil97c}{article}{
      author={Kisil, Vladimir~V.},
       title={Analysis in {$\mathbf{R}\sp {1,1}$} or the principal function
  theory},
        date={1999},
        ISSN={0278-1077},
     journal={Complex Variables Theory Appl.},
      volume={40},
      number={2},
       pages={93\ndash 118},
        note={\arXiv{funct-an/9712003}},
      review={\MR{MR1744876 (2000k:30078)}},
}

\bib{Kisil94e}{article}{
      author={Kisil, Vladimir~V.},
       title={Relative convolutions. {I}. {P}roperties and applications},
        date={1999},
        ISSN={0001-8708},
     journal={Adv. Math.},
      volume={147},
      number={1},
       pages={35\ndash 73},
        note={\arXiv{funct-an/9410001},
  \href{http://www.idealibrary.com/links/doi/10.1006/aima.1999.1833}{On-line}.
  \Zbl{933.43004}},
      review={\MR{MR1725814 (2001h:22012)}},
}

\bib{Kisil98a}{article}{
      author={Kisil, Vladimir~V.},
       title={Wavelets in {B}anach spaces},
        date={1999},
        ISSN={0167-8019},
     journal={Acta Appl. Math.},
      volume={59},
      number={1},
       pages={79\ndash 109},
        note={\arXiv{math/9807141},
  \href{http://dx.doi.org/10.1023/A:1006394832290}{On-line}},
      review={\MR{MR1740458 (2001c:43013)}},
}

\bib{Kisil00a}{article}{
      author={Kisil, Vladimir~V.},
       title={Quantum and classical brackets},
        date={2002},
        ISSN={0020-7748},
     journal={Internat. J. Theoret. Phys.},
      volume={41},
      number={1},
       pages={63\ndash 77},
        note={\arXiv{math-ph/0007030}.
  \href{http://dx.doi.org/10.1023/A:1013269432516}{On-line}},
      review={\MR{2003b:81105}},
}

\bib{Kisil01c}{inproceedings}{
      author={Kisil, Vladimir~V.},
       title={Two slits interference is compatible with particles'
  trajectories},
        date={2002},
   booktitle={Quantum theory: Reconsideration of foundations},
      editor={Khrennikov, Andrei},
      series={Mathematical Modelling in Physics, Engineering and Cognitive
  Science},
      volume={2},
   publisher={V\"axj\"o University Press},
       pages={215\ndash 226},
        note={\arXiv{quant-ph/0111094}},
}

\bib{Kisil02e}{article}{
      author={Kisil, Vladimir~V.},
       title={{$p$}-{M}echanics as a physical theory: an introduction},
        date={2004},
        ISSN={0305-4470},
     journal={J. Phys. A},
      volume={37},
      number={1},
       pages={183\ndash 204},
        note={\arXiv{quant-ph/0212101},
  \href{http://stacks.iop.org/0305-4470/37/183}{On-line}. \Zbl{1045.81032}},
      review={\MR{MR2044764 (2005c:81078)}},
}

\bib{Kisil05c}{article}{
      author={Kisil, Vladimir~V.},
       title={A quantum-classical bracket from {$p$}-mechanics},
        date={2005},
        ISSN={0295-5075},
     journal={Europhys. Lett.},
      volume={72},
      number={6},
       pages={873\ndash 879},
        note={\arXiv{quant-ph/0506122},
  \href{http://dx.doi.org/10.1209/epl/i2005-10324-7}{On-line}},
      review={\MR{MR2213328 (2006k:81134)}},
}

\bib{Kisil06a}{article}{
      author={Kisil, Vladimir~V.},
       title={Erlangen program at large--0: Starting with the group {${\rm
  SL}\sb 2({\bf R})$}},
        date={2007},
        ISSN={0002-9920},
     journal={Notices Amer. Math. Soc.},
      volume={54},
      number={11},
       pages={1458\ndash 1465},
        note={\arXiv{math/0607387},
  \href{http://www.ams.org/notices/200711/tx071101458p.pdf}{On-line}.
  \Zbl{1137.22006}},
      review={\MR{MR2361159}},
}

\bib{Kisil06b}{article}{
      author={Kisil, Vladimir~V.},
       title={Two-dimensional conformal models of space-time and their
  compactification},
        date={2007},
        ISSN={0022-2488},
     journal={J. Math. Phys.},
      volume={48},
      number={7},
       pages={\href{http://link.aip.org/link/?JMP/48/073506}{073506}, 8},
        note={\arXiv{math-ph/0611053}. \Zbl{1144.81368}},
      review={\MR{MR2337687}},
}

\bib{Kisil09a}{article}{
      author={Kisil, Vladimir~V.},
       title={Comment on ``{D}o we have a consistent non-adiabatic
  quantum-classical mechanics?'' by {A}gostini {F}. et al},
        date={2010},
     journal={Europhys. Lett. EPL},
      volume={89},
       pages={50005},
        note={\arXiv{0907.0855}},
}

\bib{Kisil09b}{article}{
      author={Kisil, Vladimir~V.},
       title={Computation and dynamics: {C}lassical and quantum},
        date={2010},
     journal={AIP Conference Proceedings},
      volume={1232},
      number={1},
       pages={306\ndash 312},
         url={http://link.aip.org/link/?APC/1232/306/1},
        note={\arXiv{0909.1594}},
}

\bib{Kisil07a}{article}{
      author={Kisil, Vladimir~V.},
       title={Erlangen program at large---2: {Inventing} a wheel. {The}
  parabolic one},
        date={2010},
     journal={Zb. Pr. Inst. Mat. NAN Ukr. (Proc. Math. Inst. Ukr. Ac. Sci.)},
      volume={7},
      number={2},
       pages={89\ndash 98},
        note={\arXiv{0707.4024}},
}

\bib{Kisil05a}{article}{
      author={Kisil, Vladimir~V.},
       title={Erlangen program at large--1: Geometry of invariants},
        date={2010},
     journal={SIGMA, Symmetry Integrability Geom. Methods Appl.},
      volume={6},
      number={076},
       pages={45},
        note={\arXiv{math.CV/0512416}. \MR{2011i:30044}. \Zbl{1218.30136}},
}

\bib{Kisil09c}{article}{
      author={Kisil, Vladimir~V.},
       title={{E}rlangen program at large---2 1/2: {I}nduced representations
  and hypercomplex numbers},
        date={2011},
     journal={{\cyr Izvestiya Komi nauchnogo centra UrO RAN} [Izvestiya Komi
  nauchnogo centra UrO RAN]},
      volume={1},
      number={5},
       pages={4\ndash 10},
        note={\arXiv{0909.4464}},
}

\bib{Kisil11a}{article}{
      author={Kisil, Vladimir~V.},
       title={{E}rlangen {P}rogramme at {L}arge 3.2: {L}adder operators in
  hypercomplex mechanics},
        date={2011},
     journal={Acta Polytechnica},
      volume={51},
      number={4},
       pages={\href{http://ctn.cvut.cz/ap/download.php?id=614}{44\ndash 53}},
        note={\arXiv{1103.1120}},
}

\bib{Kisil11c}{incollection}{
      author={Kisil, Vladimir~V.},
       title={{E}rlangen programme at large: an {O}verview},
        date={2012},
   booktitle={Advances in applied analysis},
      editor={Rogosin, S.V.},
      editor={Koroleva, A.A.},
   publisher={Birkh\"auser Verlag},
     address={Basel},
       pages={1\ndash 94},
        note={\arXiv{1106.1686}},
}

\bib{Kisil12a}{book}{
      author={Kisil, Vladimir~V.},
       title={Geometry of {M}\"obius transformations: {E}lliptic, parabolic and
  hyperbolic actions of {$\mathrm{SL}_2(\mathbf{R})$}},
   publisher={Imperial College Press},
     address={London},
        date={2012},
        note={Includes a live DVD. \Zbl{1254.30001}},
}

\bib{Kisil10a}{article}{
      author={Kisil, Vladimir~V.},
       title={Hypercomplex representations of the {H}eisenberg group and
  mechanics},
        date={2012},
        ISSN={0020-7748},
     journal={Internat. J. Theoret. Phys.},
      volume={51},
      number={3},
       pages={964\ndash 984},
         url={http://dx.doi.org/10.1007/s10773-011-0970-0},
        note={\arXiv{1005.5057}. \Zbl{1247.81232}},
      review={\MR{2892069}},
}

\bib{Kisil12c}{article}{
      author={Kisil, Vladimir~V.},
       title={Is commutativity of observables the main feature, which separate
  classical mechanics from quantum?},
        date={2012},
     journal={{\cyr Izvestiya Komi nauchnogo centra UrO RAN} [Izvestiya Komi
  nauchnogo centra UrO RAN]},
      volume={3},
      number={11},
       pages={4\ndash 9},
        note={\arXiv{1204.1858}},
}

\bib{Kisil12b}{article}{
      author={Kisil, Vladimir~V.},
       title={Operator covariant transform and local principle},
        date={2012},
     journal={J. Phys. A: Math. Theor.},
      volume={45},
       pages={244022},
        note={\arXiv{1201.1749}.
  \href{http://stacks.iop.org/1751-8121/45/244022}{On-line}},
}

\bib{Kisil13b}{article}{
      author={Kisil, Vladimir~V.},
       title={Boundedness of relative convolutions on nilpotent {L}ie groups},
        date={2013},
     journal={Zb. Pr. Inst. Mat. NAN Ukr. (Proc. Math. Inst. Ukr. Ac. Sci.)},
      volume={10},
      number={4--5},
       pages={185\ndash 189},
        note={\arXiv{1307.3882}},
}

\bib{Kisil13a}{article}{
      author={Kisil, Vladimir~V.},
       title={Calculus of operators: {C}ovariant transform and relative
  convolutions},
        date={2014},
     journal={Banach J. Math. Anal.},
      volume={8},
      number={2},
       pages={156\ndash 184},
         url={\url{http://www.emis.de/journals/BJMA/tex_v8_n2_a15.pdf}},
        note={\arXiv{1304.2792},
  \href{http://www.emis.de/journals/BJMA/tex_v8_n2_a15.pdf}{on-line}},
}

\bib{Kisil12d}{article}{
      author={Kisil, Vladimir~V.},
       title={The real and complex techniques in harmonic analysis from the
  point of view of covariant transform},
        date={2014},
     journal={Eurasian Math. J.},
      volume={5},
       pages={95\ndash 121},
        note={\arXiv{1209.5072}.
  \href{http://emj.enu.kz/images/pdf/2014/5-1-4.pdf}{On-line}},
}

\bib{Kisil13c}{incollection}{
      author={Kisil, Vladimir~V.},
       title={Uncertainty and analyticity},
    language={English},
        date={2015},
   booktitle={Current trends in analysis and its applications},
      editor={Mityushev, Vladimir~V.},
      editor={Ruzhansky, Michael~V.},
      series={Trends in Mathematics},
   publisher={Springer International Publishing},
       pages={583\ndash 590},
         url={http://dx.doi.org/10.1007/978-3-319-12577-0_64},
        note={\arXiv{1312.4583}},
}

\bib{Kisil14a}{article}{
      author={Kisil, Vladimir~V.},
       title={Remark on continued fractions, {M\"obius} transformations and
  cycles},
        date={2016},
     journal={{\cyr Izvestiya Komi nauchnogo centra UrO RAN} [Izvestiya Komi
  nauchnogo centra UrO RAN]},
      volume={25},
      number={1},
       pages={11\ndash 17},
         url={http://www.izvestia.komisc.ru/Archive/i25_ann.files/kisil.pdf},
        note={\arXiv{1412.1457},
  \href{http://www.izvestia.komisc.ru/Archive/i25_ann.files/kisil.pdf}{on-line}},
}

\bib{Lang85}{book}{
      author={Lang, Serge},
       title={{${\rm SL}\sb 2({\bf R})$}},
      series={Graduate Texts in Mathematics},
   publisher={Springer-Verlag},
     address={New York},
        date={1985},
      volume={105},
        ISBN={0-387-96198-4},
        note={Reprint of the 1975 edition},
      review={\MR{803508 (86j:22018)}},
}

\bib{LavrentShabat77}{book}{
      author={Lavrent{\cprime}ev, M.~A.},
      author={Shabat, B.~V.},
       title={{\cyr Problemy Gidrodinamiki i Ikh Matematicheskie Modeli}.
  [{P}roblems of hydrodynamics and their mathematical models]},
    language={Russian},
     edition={Second},
   publisher={Izdat. ``Nauka'', Moscow},
        date={1977},
      review={\MR{56 \#17392}},
}

\bib{LevyLeblond65a}{article}{
      author={L{\'e}vy-Leblond, Jean-Marc},
       title={Une nouvelle limite non-relativiste du groupe de {P}oincar\'e},
        date={1965},
     journal={Ann. Inst. H. Poincar\'e Sect. A (N.S.)},
      volume={3},
       pages={1\ndash 12},
      review={\MR{0192900 (33 \#1125)}},
}

\bib{Low09a}{article}{
      author={{Low}, S.~G.},
       title={{Noninertial Symmetry Group of Hamilton's Mechanics}},
        date={2009-03},
      eprint={0903.4397},
        note={\arXiv{0903.4397}},
}

\bib{Mackey63}{book}{
      author={Mackey, George~W.},
       title={Mathematical foundations of quantum mechanics},
   publisher={W.A.~Benjamin, Inc.},
     address={New York, Amsterdam},
        date={{\noopsort{}}1963},
}

\bib{Mazorchuk09a}{book}{
      author={Mazorchuk, Volodymyr},
       title={Lectures on {$\germ{sl}_2(\Bbb C)$}-modules},
   publisher={Imperial College Press, London},
        date={2010},
        ISBN={978-1-84816-517-5; 1-84816-517-X},
      review={\MR{2567743 (2011b:17019)}},
}

\bib{Niederer73a}{article}{
      author={Niederer, U.},
       title={The maximal kinematical invariance group of the free
  {S}chr\"odinger equation},
        date={1972/73},
        ISSN={0018-0238},
     journal={Helv. Phys. Acta},
      volume={45},
      number={5},
       pages={802\ndash 810},
      review={\MR{0400948 (53 \#4778)}},
}

\bib{Penrose78a}{inproceedings}{
      author={Penrose, Roger},
       title={The complex geometry of the natural world},
        date={1980},
   booktitle={Proceedings of the {I}nternational {C}ongress of
  {M}athematicians. {V}ol. 1, 2},
      editor={Lehto, Olli},
   publisher={Academia Scientiarum Fennica, Helsinki},
       pages={189\ndash 194},
        note={Held in Helsinki, August 15--23, 1978},
      review={\MR{562607}},
}

\bib{PercivalRichards82}{book}{
      author={Percival, Ian},
      author={Richards, Derek},
       title={{Introduction to Dynamics.}},
    language={English},
   publisher={{Cambridge etc.: Cambridge University Press. VIII, 228 p.}},
        date={1982},
}

\bib{Pilipchuk10a}{book}{
      author={Pilipchuk, Valery~N.},
       title={{Nonlinear dynamics. Between linear and impact limits.}},
    language={English},
      series={Lecture Notes in Applied and Computational Mechanics},
   publisher={{Springer}},
     address={Berlin},
        date={2010},
      volume={52},
}

\bib{Pilipchuk11a}{article}{
      author={Pilipchuk, Valery~N.},
       title={Non-smooth spatio-temporal coordinates in nonlinear dynamics},
        date={2011-01},
      eprint={1101.4597},
        note={\arXiv{1101.4597}},
}

\bib{PilipchukAndrianovMarkert16a}{article}{
      author={Pilipchuk, Valery~N.},
      author={Andrianov, Igor~V.},
      author={Markert, Bernd},
       title={Analysis of micro-structural effects on phononic waves in layered
  elastic media with periodic nonsmooth coordinates},
        date={2016},
        ISSN={0165-2125},
     journal={Wave Motion},
      volume={63},
       pages={149\ndash 169},
         url={http://dx.doi.org/10.1016/j.wavemoti.2016.01.007},
      review={\MR{3477689}},
}

\bib{Pimenov65a}{article}{
      author={Pimenov, R.I.},
       title={Unified axiomatics of spaces with maximal movement group},
    language={Russian},
        date={1965},
     journal={Litov. Mat. Sb.},
      volume={5},
       pages={457\ndash 486},
        note={\Zbl{0139.37806}},
}

\bib{Pontryagin86a}{book}{
      author={Pontryagin, L.~S.},
       title={{\cyr Obobshcheniya chisel} [{G}eneralisations of numbers]},
    language={Russian},
      series={{\cyr Bibliotechka ``Kvant''} [Library ``Kvant'']},
   publisher={``Nauka''},
     address={Moscow},
        date={1986},
      volume={54},
      review={\MR{MR886479 (88c:00005)}},
}

\bib{Porteous95}{book}{
      author={Porteous, Ian~R.},
       title={Clifford algebras and the classical groups},
      series={Cambridge Studies in Advanced Mathematics},
   publisher={Cambridge University Press},
     address={Cambridge},
        date={1995},
      volume={50},
        ISBN={0-521-55177-3},
      review={\MR{MR1369094 (97c:15046)}},
}

\bib{Prezhdo-Kisil97}{article}{
      author={Prezhdo, Oleg~V.},
      author={Kisil, Vladimir~V.},
       title={Mixing quantum and classical mechanics},
        date={1997},
        ISSN={1050-2947},
     journal={Phys. Rev. A (3)},
      volume={56},
      number={1},
       pages={162\ndash 175},
        note={\arXiv{quant-ph/9610016}},
      review={\MR{MR1459700 (99j:81010)}},
}

\bib{SrivastavaTuanYakubovich00a}{article}{
      author={Srivastava, H.~M.},
      author={Tuan, Vu~Kim},
      author={Yakubovich, S.~B.},
       title={The {C}herry transform and its relationship with a singular
  {S}turm-{L}iouville problem},
        date={2000},
        ISSN={0033-5606},
     journal={Q. J. Math.},
      volume={51},
      number={3},
       pages={371\ndash 383},
         url={http://dx.doi.org/10.1093/qjmath/51.3.371},
      review={\MR{1782100 (2001g:44010)}},
}

\bib{MTaylor81}{book}{
      author={Taylor, Michael~E.},
       title={Pseudodifferential operators},
      series={Princeton Mathematical Series},
   publisher={Princeton University Press},
     address={Princeton, N.J.},
        date={1981},
      volume={34},
        ISBN={0-691-08282-0},
      review={\MR{82i:35172}},
}

\bib{MTaylor86}{book}{
      author={Taylor, Michael~E.},
       title={Noncommutative harmonic analysis},
      series={Mathematical Surveys and Monographs},
   publisher={American Mathematical Society},
     address={Providence, RI},
        date={1986},
      volume={22},
        ISBN={0-8218-1523-7},
      review={\MR{88a:22021}},
}

\bib{ATorre08a}{article}{
      author={Torre, A},
       title={A note on the general solution of the paraxial wave equation: a
  {L}ie algebra view},
        date={2008},
     journal={Journal of Optics A: Pure and Applied Optics},
      volume={10},
      number={5},
       pages={055006 (14pp)},
         url={http://dx.doi.org/10.1088/1464-4258/10/5/055006},
}

\bib{ATorre10a}{article}{
      author={Torre, A},
       title={Linear and quadratic exponential modulation of the solutions of
  the paraxial wave equation},
        date={2010},
     journal={Journal of Optics A: Pure and Applied Optics},
      volume={12},
      number={3},
       pages={035701 (11pp)},
         url={http://stacks.iop.org/2040-8986/12/035701},
}

\bib{Ulrych05a}{article}{
      author={Ulrych, S.},
       title={Relativistic quantum physics with hyperbolic numbers},
        date={2005},
        ISSN={0370-2693},
     journal={Phys. Lett. B},
      volume={625},
      number={3--4},
       pages={313\ndash 323},
      review={\MR{MR2170329 (2006e:81103a)}},
}

\bib{Ulrych08a}{article}{
      author={Ulrych, S.},
       title={Representations of {C}lifford algebras with hyperbolic numbers},
        date={2008},
        ISSN={0188-7009},
     journal={Adv. Appl. Clifford Algebr.},
      volume={18},
      number={1},
       pages={93\ndash 114},
         url={http://dx.doi.org/10.1007/s00006-007-0057-4},
      review={\MR{MR2377525 (2009d:81139)}},
}

\bib{Ulrych10a}{article}{
      author={Ulrych, S.},
       title={Considerations on the hyperbolic complex {K}lein-{G}ordon
  equation},
        date={2010},
        ISSN={0022-2488},
     journal={J. Math. Phys.},
      volume={51},
      number={6},
       pages={063510, 8},
         url={http://dx.doi.org/10.1063/1.3397456},
      review={\MR{2676487 (2011k:81083)}},
}

\bib{Ulrych2014a}{article}{
      author={Ulrych, S.},
       title={Conformal relativity with hypercomplex variables},
        date={2014},
        ISSN={1364-5021},
     journal={Proceedings of the Royal Society of London A: Mathematical,
  Physical and Engineering Sciences},
      volume={470},
      number={2168},
  eprint={http://rspa.royalsocietypublishing.org/content/470/2168/20140027.full.pdf},
  url={http://rspa.royalsocietypublishing.org/content/470/2168/20140027},
}

\bib{Vasilevski99}{article}{
      author={Vasilevski, N.~L.},
       title={On the structure of {B}ergman and poly-{B}ergman spaces},
        date={1999},
        ISSN={0378-620X},
     journal={Integral Equations Operator Theory},
      volume={33},
      number={4},
       pages={471\ndash 488},
         url={http://dx.doi.org/10.1007/BF01291838},
      review={\MR{1682807 (2000g:46032)}},
}

\bib{Vourdas06a}{article}{
      author={Vourdas, A.},
       title={Analytic representations in quantum mechanics},
        date={2006},
        ISSN={0305-4470},
     journal={J. Phys. A},
      volume={39},
      number={7},
       pages={R65\ndash R141},
         url={http://dx.doi.org/10.1088/0305-4470/39/7/R01},
      review={\MR{MR2210163 (2007g:81069)}},
}

\bib{Wulfman10a}{book}{
      author={Wulfman, Carl~E.},
       title={{Dynamical Symmetry}},
   publisher={World Scientific},
        date={2010},
        ISBN={978-981-4291-36-1},
         url={http://www.worldscibooks.com/physics/7548.html},
}

\bib{Yaglom79}{book}{
      author={Yaglom, I.~M.},
       title={A simple non-{E}uclidean geometry and its physical basis},
      series={Heidelberg Science Library},
   publisher={Springer-Verlag},
     address={New York},
        date={1979},
        ISBN={0-387-90332-1},
        note={Translated from the Russian by Abe Shenitzer, with the editorial
  assistance of Basil Gordon},
      review={\MR{MR520230 (80c:51007)}},
}

\bib{Zachos02a}{article}{
      author={Zachos, Cosmas},
       title={Deformation quantization: Quantum mechanics lives and works in
  phase-space},
        date={2002},
        ISSN={0217-751X},
     journal={Internat. J. Modern Phys. A},
      volume={17},
      number={3},
       pages={297\ndash 316},
        note={\arXiv{hep-th/0110114}},
      review={\MR{1 888 937}},
}

\bib{Zejliger34}{book}{
      author={Zejliger, D.N.},
       title={{\cyr Kompleksnaya Linei0chataya Geometriya. Poverhnosti i
  Kongruencii.} [{C}omplex lined geometry. {S}urfaces and congruency]},
    language={Russian},
   publisher={GTTI},
     address={Leningrad},
        date={1934},
}

\end{biblist}
\end{bibdiv}
\providecommand{\noopsort}[1]{} \providecommand{\printfirst}[2]{#1}
  \providecommand{\singleletter}[1]{#1} \providecommand{\switchargs}[2]{#2#1}
  \providecommand{\irm}{\textup{I}} \providecommand{\iirm}{\textup{II}}
  \providecommand{\vrm}{\textup{V}} \providecommand{\cprime}{'}
  \providecommand{\eprint}[2]{\texttt{#2}}
  \providecommand{\myeprint}[2]{\texttt{#2}}
  \providecommand{\arXiv}[1]{\myeprint{http://arXiv.org/abs/#1}{arXiv:#1}}
  \providecommand{\doi}[1]{\href{http://dx.doi.org/#1}{doi:
  #1}}\providecommand{\CPP}{\texttt{C++}}
  \providecommand{\NoWEB}{\texttt{noweb}}
  \providecommand{\MetaPost}{\texttt{Meta}\-\texttt{Post}}
  \providecommand{\GiNaC}{\textsf{GiNaC}}
  \providecommand{\pyGiNaC}{\textsf{pyGiNaC}}
  \providecommand{\Asymptote}{\texttt{Asymptote}}

\printindex
\end{document}